\documentclass[12pt,reqno]{amsart}
\advance\voffset by -2.0cm \advance \hoffset by -1.5cm
\usepackage{hyperref}

\usepackage{cite}
\usepackage{amsmath}
\usepackage{amssymb}
\usepackage{amsthm}
\usepackage{graphicx,color}
\usepackage{verbatim}

\usepackage{mathtools}

\usepackage{bbm}
\newcommand{\id}{\mathbbm{1}}

\newcommand{\CC}{\mathbb{C}} 
\newcommand{\RR}{\mathbb{R}} 
\newcommand{\QQ}{\mathbb{Q}} 
 
\newcommand{\HH}{\mathbb{H}}
\newcommand{\ZZ}{\mathbb{Z}} 
 
\newcommand{\PP}{\mathbf{P}} 
\newcommand{\MM}{\mathbb{M}}
\newcommand{\TT}{\mathbb{T}}
\newcommand{\bfT}{\mathbf{T}}

\newcommand{\E}{\mathcal{E}}
\newcommand{\F}{\mathcal{F}}

\newcommand{\Hh}{\mathcal{H}}

\newcommand{\Q}{\mathcal{Q}}

\newcommand{\V}{\mathcal{V}}
\newcommand{\W}{\mathcal{W}}

\newcommand{\g}{\mathfrak{g}}
\newcommand{\h}{\mathfrak{h}}
\newcommand{\cc}{\mathfrak{c}}

\usepackage{xcolor}

\newtheorem{theorem}{Theorem}

\newtheorem{lemma}[theorem]{Lemma}

\DeclareMathOperator{\Aut}{Aut}

\DeclareMathOperator{\Tr}{Tr}

\numberwithin{equation}{section}

\def\be{\begin{equation}}
\def\ee{\end{equation}}
\newcommand{\bea}{\begin{eqnarray}}
\newcommand{\eea}{\end{eqnarray}}

\allowdisplaybreaks[3]
\begin{document}
\title[Monstrous BPS-algebras and the superstring origin of moonshine]{Monstrous BPS-algebras and \\ the superstring origin of moonshine }
\author{Natalie M. Paquette}
\address{Stanford  Institute  for  Theoretical  Physics,  Department  of  Physics,  Stanford  University,
Stanford, CA 94305, USA}
\email{npaquett@stanford.edu}

\author{Daniel Persson}
\address{Department of Physics, Chalmers University of Technology,
412 96, Gothenburg, Sweden}
\email{daniel.persson@chalmers.se}

\author{Roberto Volpato}
\address{Theory  Group,  SLAC  National  Accelerator  Laboratory, Menlo  Park,
CA 94025, USA \\ 
\indent \& Stanford  Institute  for  Theoretical  Physics,  Department  of  Physics,  Stanford 
 University, Stanford, CA 94305, USA}
\email{volpato@slac.stanford.edu}

\begin{abstract} We provide a physics derivation of Monstrous moonshine. We show that the McKay-Thompson series $T_g$, $g\in \MM$, can be interpreted as supersymmetric indices counting spacetime BPS-states in certain heterotic string models. The invariance groups  of these series arise naturally as spacetime T-duality groups and their genus zero property descends from the behaviour of these heterotic models in suitable decompactification limits. We also show that the space of BPS-states forms a module for the Monstrous Lie algebras $\mathfrak{m}_g$, constructed by Borcherds and Carnahan. We argue that $\mathfrak{m}_g$ arise in the heterotic models as algebras of spontaneously broken gauge symmetries, whose generators are in exact correspondence with BPS-states. This gives $\mathfrak{m}_g$  an interpretation as a kind of BPS-algebra. 
\end{abstract}
\maketitle
\tableofcontents
\section{Introduction and summary}

\noindent The famous Monstrous moonshine conjecture \cite{Conway:1979kx} has motivated a host of new developments at the intersection between theoretical physics, algebra, number theory, and group theory. In its basic formulation, the conjecture associates with each element $g$ of the Monster group $\MM$ (the largest sporadic finite simple group) a modular function $T_g$, the McKay-Thompson series. The invariance groups $\Gamma_g\subset SL(2,\RR)$ of the series $T_g$ were conjectured to satisfy very constraining properties: in particular, the quotient $\HH/\Gamma_g$ of the upper half plane by $\Gamma_g$ is expected to have genus zero.  The conjecture has been proved by Borcherds \cite{Borcherds}, based on previous contributions by many authors, in particular by Frenkel, Lepowsky, and Meurman \cite{FLM}. While the proof dates back to almost 25 years ago, many aspects of Monstrous moonshine are still unclear.

This paper aims to provide a natural physical framework where some of these open issues can be understood. The main idea is to interpret the McKay-Thompson series as supersymmetric indices in certain heterotic string compactifications. This leads to the two main results of the paper. First, we show that the modular groups $\Gamma_g$ can be understood as groups of dualities in these models, and we  provide a physical derivation of  their genus zero property. Second, we show that the Monster Lie algebra $\mathfrak{m}$, introduced by Borcherds in his proof of Monstrous moonshine, is the algebra of spacetime BPS-states in these heterotic models. More precisely, the algebra arises as a spontaneously broken gauge symmetry, whose generators are in exact correspondence with the BPS single-particle states. In the remainder of the introduction we shall provide some motivational background, and give a more detailed overview of 
the results. 

\subsection{Monstrous CHL-models}
Monstrous moonshine \cite{Conway:1979kx,FLM,Borcherds} associates to each element $g$ in the Monster group $\mathbb{M}$, a modular function (McKay-Thompson series)
\be
T_g(\tau)=\text{Tr}_{V^\natural} (g q^{L_0-1})=\sum_{n=-1}^{\infty} \text{Tr}_{V^\natural_n}(g)q^{n}, \hspace{1cm} q:=e^{2\pi i \tau},
\ee
where the coefficients are characters of $g$  in  the graded components of the Frenkel-Lepowsky-Meurman Monster module $V^\natural=\bigoplus_{n=-1}^{\infty} V_n^\natural$. In physics language, $V^\natural$ is a holomorphic two dimensional conformal field theory (CFT) of central charge $24$ with symmetry group $\MM$, and $T_g$ are its $g$-twisted partition functions. By the cyclicity property of the trace, the McKay-Thompson series are actually class functions (i.e., they depend only on the conjugacy class of $g$).
A key ingredient in Borcherds' proof of Monstrous moonshine was the construction of an infinite-dimensional Lie algebra $\mathfrak{m}$, known as the \emph{Monster Lie algebra}, obtained by applying a certain functor to the Monster module $V^\natural$. In this paper we show that $\mathfrak{m}$ is the ``algebra of BPS-states" of a certain heterotic string theory. The fact that BPS-states in string theory form an algebra was first proposed by Harvey and Moore \cite{Harvey:1995fq,Harvey:1996gc}, but the precise structure of this algebra is still poorly understood (see \cite{Neumann:1997pr,Fiol:2000wx,Kontsevich:2010px,Gukov:2011ry,Chuang:2013wt} for various attempts). In the original work, Harvey and Moore envisioned that the algebraic structure was captured by the OPEs between BPS-vertex operators, and the algebra should be closely related to a Borcherds-Kac-Moody algebra (BKM-algebra). A slightly different construction was proposed in \cite{Gaberdiel:2011qu,Hohenegger:2011ff}, where the space of BPS-states in a certain $\mathcal{N}=4$ string theory realized a \emph{module} for a BKM-algebra. In this work we take the latter approach, and show that the space of BPS-states in a certain heterotic orbifold forms a representation of the Monster Lie algebra $\mathfrak{m}$ constructed from $V^\natural$. The algebra is generated by BRST-exact string states, that are in one to one correspondence with the BPS string states. This provides the first realization of the Monster Lie algebra as an algebra of BPS-states, answering a long-standing question raised by, in particular, Harvey and Moore \cite{Harvey:1995fq}, and Carnahan \cite{Carnahan2014}. In fact, our construction is more general and applies to the entire class of Lie algebras $\mathfrak{m}_g$, $g\in \mathbb{M}$, constructed by Carnahan \cite{2012arXiv1208.6254C} in the context of  proving ``generalized moonshine" (the full proof is now complete and constitutes \cite{2008arXiv0812.3440C,Carnahan2014,2012arXiv1208.6254C}). Generalized moonshine was proposed by Norton \cite{Norton1987} and concerns, for each commuting pair $g,h\in \mathbb{M}$, the following ``twisted'' and ``twined'' generalizations of the McKay-Thompson series: 
\be 
T_{g,h}(\tau)=\text{Tr}_{V^\natural_g}(h q^{L_0-1}),
\ee
where $V^\natural_g$ is the $g$-twisted sector in the orbifold of $V^\natural$ by $g$. The specialization\footnote{We will always use $\id$ to denote the identity element $\id \in \MM$.} $T_{\id,h}$ recovers the McKay-Thompson series, while  $T_{g,\id}$ constitutes part of the denominator formula for the Lie algebra $\mathfrak{m}_g$ \cite{2012arXiv1208.6254C}. Each such Lie algebra is the algebra of BPS-states for a heterotic string model labeled by $g$. 

The key ingredient in our analysis is a new type of ``Monstrous CHL'' model, where the name is taken from the analogous construction by Chaudhuri, Hockney, and Lykken \cite{Chaudhuri:1995fk,Chaudhuri:1995ee}. We consider the heterotic string compactified to 1+1 dimensions, with the internal CFT of the form $V^\natural\times \bar V^{s\natural}$, where $\bar V^{s\natural}$ is the  super-moonshine module for the Conway group, envisaged by \cite{FLM} and constructed by Duncan \cite{duncan2007super}. It has no NS-sector states of conformal weight 1/2, but 24 Ramond ground states of weight 1/2. We then consider the further compactification of this theory on the spatial circle $S^1$ of radius $R$ and we take a $\mathbb{Z}_N$-orbifold of this theory by $(\delta, g)$, where $\delta$  is an order $N$ shift along $S^1$ and $g\in \mathbb{M}$. 
The resulting theory has $(0,24)$ spacetime supersymmetry, and the spectrum contains two kinds of irreducible representations: a short (BPS) $1$-dimensional representation and a long representation of dimension $2^{12}$. 

These constructions raise an immediate puzzle. There are no local massless states in the spectrum since in lightcone quantization these precisely correspond to currents  in the internal CFT, and we have just seen that these are absent in our models. In particular, there is no dilaton and hence, at first sight, no string coupling! We will propose a resolution to this puzzle in Section \ref{sec:dilaton} and henceforth tune the string coupling to zero; obstacles to turning on finite string coupling are discussed briefly in Section \ref{sec:anomaly}.

\subsection{The supersymmetric index}

One of the main points of the paper is that in these models we can compute a supersymmetric index $Z$ which counts (with signs) the number of  BPS-states. This index will allow us to provide a spacetime interpretation of Monstrous moonshine and use this to shed light on the elusive genus zero property of moonshine. 

Let us first consider the unorbifolded case. We compute the index in three different ways which each provide certain useful clues. First, using a Fock space construction we define $Z$ as a weighted trace over the second-quantized BPS-Hilbert space $\mathcal{H}_{BPS}$: 
\be 
Z(T,U):=\text{Tr}_{\mathcal{H}_{BPS}}\big((-1)^F e^{2\pi i TW} e^{2\pi  i UM}\big),
\ee
where $F$ is the fermion number, $(W,M)$ represent winding and momentum operators along $S^1$, and $(T,U)\in \mathbb{H}\times \mathbb{H}$ are the associated (complexified) chemical potentials, involving in particular the radius $R$ of $S^1$ and the inverse temperature $\beta$. Evaluating this index we find the explicit formula
\be
Z(T,U)=\Bigl(e^{2\pi i (w_0T+m_0U)}\prod_{\substack{m> 0\\w\in \ZZ}}(1-e^{2\pi i Um}e^{2\pi iTw})^{c(mw)}\Bigr)^{24},
\label{susyindexFock}
\ee
where $c(mw)$ are the Fourier coefficients of the modular-invariant $J$-function
\be
J(\tau)=T_\id(\tau)=\sum_{n=-1}^{\infty}c(n)q^n=q^{-1} + 196884q +\cdots. 
\ee and we allow for possible winding and momentum $w_0,m_0$ of the ground state. We will argue below that the correct values are $w_0=-1$, $m_0=0$.

We stress that in the original Monstrous moonshine the $J$-function is the graded dimension of the moonshine vertex operator algebra $V^\natural$, or in physics parlance, the partition function of the Monster CFT, and therefore it is intrinsically a \emph{worldsheet}  object. On the other hand, the supersymmetric index $Z(T,U)$ is a \emph{spacetime} object, and therefore provides a new spacetime interpretation of Monstrous moonshine.

Inspired by generalized moonshine we also extend this analysis to all Monstrous CHL-models and we define for each commuting pair $g,h\in \mathbb{M}$ the twisted twined index
\be
Z_{g,h}(T,U):=\text{Tr}_{\mathcal{H}^g_{BPS}}\big(h(-1)^F e^{2\pi i TW} e^{2\pi i UM}\big),
\ee
where $\mathcal{H}^g_{BPS}$ is the space of BPS-states in the CHL-model. We will be particularly interested in the specialization $Z_{g,\id}$,  for which we derive an infinite product formula analogous to \eqref{susyindexFock}.

In the second approach, the index is computed in terms of an Euclidean path integral, with the Euclidean time compactified on a circle with period the inverse temperature $\beta$. We argue that the path-integral is one loop exact and show that it reproduces the modulus squared of the supersymmetric index.
\be
e^{-S_{1-loop}}=|Z(T,U)|^2. 
\ee
The holomorphic part of the one-loop integral is the free energy $\mathcal{F}$ and thus we have the expected relation: 
\be
\mathcal{F}(T,U)=\log\, Z(T,U),
\ee
with similar results for all the twisted indices $Z_{g,\id}(T,U)$. In this second approach, the invariance of the index under the T-duality group is manifest.

The third approach is purely algebraic. By the work of Borcherds it is  known that, for suitable values of $w_0,m_0$, the infinite product formulas \eqref{susyindexFock} and its twined generalizations are the denominator formulas for the Monster Lie algebra $\mathfrak{m}$ and its generalizations $\mathfrak{m}_g$, respectively. We give a spacetime BPS state interpretation of this algebra and show that $\mathfrak{m}$ is the algebra of BRST-exact string states associated with the supersymmetric partners of the (first quantized)  string BPS-states.  BRST-exact states in string theory are expected to generate a gauge symmetry, though in general this might be spontaneously broken for finite string tension $\alpha'<\infty$. Therefore, $\mathfrak{m}$ can be interpreted as a kind of gauge algebra in this theory. As a consequence, we show that the supersymmetric index $Z_{g,\id}(T,U)$ exactly coincides with (the 24th power of) the denominator of the algebra $\mathfrak{m}_g$.\footnote{As we will discuss in section \ref{s:Algebras}, the space of physical BPS-states in the spectrum of a single first quantized string form a representation of $\mathfrak{m}_g$. We conjecture that the second quantized Fock space $\Hh_{BPS}$ is also a $\mathfrak{m}_g$-module, but we were not able to prove this statement.}

Starting from the famous product formula for the $J$-function (see, e.g., \cite{MR772491,Borcherds}):
\be 
J(\sigma)-J(\tau)=p^{-1}\prod_{m>0, n\in \mathbb{Z}} (1-p^mq^n)^{c(mn)}, \qquad p:=e^{2\pi i \sigma},
\label{denominator}\ee
and using the description as an algebra denominator, we obtain a new formula for the supersymmetric index of the associated CHL-model:
\be\label{additivedenomin}
Z_{g,\id}(T,U)=\Big(T_{\id,g}(T)-T_{g,\id}(U)\Big)^{24}.
\ee

\subsection{The genus zero property of moonshine}
An outstanding puzzle in Monstrous moonshine has been to find an explanation for the so called \emph{genus zero property}, namely that all the McKay-Thompson series are ``Hauptmoduln'' $-$ generators of the function field on $\mathbb{H}/\Gamma$ $-$ for genus zero congruence subgroups $\Gamma$ of $SL(2,\mathbb{R})$.  For $g=\id$ and some other elements in the Monster group, all the properties of moonshine (except genus zero) have a natural physical interpretation in the Monster CFT $V^\natural$. Indeed, for these elements, the space $\HH/\Gamma_g$ is simply the moduli space parameterizing complex tori with additional flat $\MM$-bundles and the McKay-Thompson series, as twisted partition functions in the CFT $V^\natural$, are naturally defined on such spaces. For other $g\in\MM$, however, the modular groups na\"ively expected by CFT arguments are strictly smaller than the actual groups $\Gamma_g$ and not necessarily of genus zero. In these cases, there is no good interpretation of the spaces $\HH/\Gamma_g$ as moduli spaces where the McKay-Thompson series should be naturally defined.

Our  approach allows us to shed light on this issue, by reinterpreting it in the context of the spacetime BPS-states of CHL-models.  We first notice that the group $G_g$ of T-dualities  of the heterotic model is a subgroup of $SL(2,\RR)\times SL(2,\RR)$ acting on $(T,U)\in \HH\times \HH$. The index $Z_{g,\id}(T,U)$ is naturally defined as a function on the moduli space $(\HH\times \HH)/G_g$ of the heterotic model. Thus, for fixed $U$,  $Z_{g,\id}(T,U)$ as a function of $T$ is defined on $\HH/\text{proj}_1(G_g)$, where $\text{proj}_1(G_g)$ is the projection of $G_g\subset SL(2,\RR)\times SL(2,\RR)$ on the first $SL(2,\RR)$ factor. We will show that $\text{proj}_1(G_g)\subset SL(2,\RR)$ is exactly the modular group $\Gamma_g$ (or, more precisely, the eigengroup $\Gamma'_g$, see section \ref{s:McHaupt}). This observation, together with the explicit formula \eqref{additivedenomin} relating $Z_{g,\id}(T,U)$ to the McKay-Thompson series $T_{\id,g}$, provides a natural string theory interpretation for the modular group $\Gamma_g$.

\bigskip

Our arguments also give a new understanding of the genus zero property of the groups $\Gamma_g$. The Hauptmodul property for the McKay-Thompson series $T_g$ is equivalent to the fact that $T_g$ has only one single pole on $\overline{ \HH/\Gamma}_g$. In turn, this is equivalent to the statement that $Z_{g,\id}(T,U)$, considered as a function of $T$ for fixed $U$, has only one pole modulo T-dualities. 

The index $Z_{g,\id}(T,U)$ can only diverge in the limit where $T$ approaches one of the cusps at the boundary of the moduli space. From the physics perspective, these cusps can always be interpreted, in a suitable duality frame, as decompactification limits at  low temperature, where the index is dominated by the ground state contribution. For example, the limit $T\to i\infty$ always corresponds to a model in two uncompactified space-time dimensions, namely heterotic strings on $V^\natural\times \bar V^{s\natural}$. The index $Z_{g,\id}$ diverges in this limit, due to the contribution $e^{-2\pi i T}$ of the ground state.  Suppose that the index $Z_{g,\id}$ diverges also at another cusp (say, for $T\to 0$), different from $T\to i\infty$. The contribution $e^{-2\pi i T}$ of the ground state is finite as $T\to 0$, so this cannot be the dominant term if the index diverges in this limit. This means that, if we vary the moduli smoothly from $T\to i\infty$ to $T\to 0$,   the model undergoes a \emph{phase transition} at a certain critical value of the moduli, where the energy of some excited state gets lower than the ground state. The contribution of this excited state becomes dominant in the `small $T$ phase' and eventually diverges for $T\to 0$. Furthermore, whenever such a phase transition occurs, the two phases are always related by a T-duality of the model. The reason is that, at the critical manifold, new massless string modes appear. The latter generate an enhanced gauge symmetry that contains, in particular, the relevant T-duality.

To summarize, whenever the index $Z_{g,\id}$ diverges at some cusp, such a cusp must be related to $i\infty$ by a T-duality. The latter is part of an enhanced gauge symmetry that exists at some critical value of the moduli, and relates two different phases for the CHL model. This implies that, up to dualities, the only divergence of the index $Z_{g,\id}$ is at the cusp $T\to i\infty$, and this property is equivalent to the Hauptmodul property for $T_{g}$.

\medskip

In many respects, our approach is very similar to Tuite's reformulation of the genus zero property in terms of orbifolds of conformal field theories \cite{Tuite1995}. Tuite noticed that a McKay-Thompson series $T_{g}(\tau)$ has a pole at $\tau\to 0$ (is \emph{unbounded}, using Gannon's terminology) if and only if the orbifold $V^\natural/\langle g\rangle$ is a VOA without currents. Furthermore, assuming that any holomorphic VOA of central charge $24$ and with no currents is isomorphic to $V^\natural$ itself, he showed that the McKay-Thompson series $T_{\id,g}(\tau)$ is unbounded at $0$ if and only if it is invariant under the \emph{Fricke involution} $\tau\to -\frac{1}{N\tau}$, where $N$ is the order of $g$. Finally, he showed that this property implies that $T_g$ is a Hauptmodul for a genus zero group.

In our picture, the decompactification limit $T\to 0$ corresponds  to a two dimensional heterotic string model on $V^\natural/\langle g\rangle\times \bar V^{s\natural}$. Using the representation of the index as an algebra denominator, we show that $Z_{g,\id}$ diverges at $T\to 0$ if and only if the orbifold $V^\natural/\langle g\rangle$ has no currents. Furthermore, we show that the Fricke involution is contained in the T-duality group $G_g$ if and only if $V^\natural/\langle g\rangle$ is isomorphic to $V^\natural$ (with some additional conditions). These results reproduce the first part of Tuite's argument. However, in order to complete the proof, we do not need any assumption about the uniqueness of $V^\natural$: as explained above, the Hauptmodul property (and, in particular, Fricke invariance) follows from the existence of critical manifolds with enhanced gauge symmetries. 

In fact, one can reverse Tuite's argument and use our construction to actually \emph{prove} the following results:

\vspace{.1cm} 

\noindent {\it (i) When the orbifold $V^{\natural}/\left<g\right>$ is consistent and has no currents,  it is isomorphic to $V^\natural$.}

\vspace{.4cm} 

\noindent {\it (ii) When the orbifold $V^\natural/\left<g\right>$ is consistent and has currents, it is isomorphic to $V^{Leech}$ (the vertex operator algebra based on the Leech lattice).}

\vspace{.4cm} 

The first statement was part of Tuite's assumptions. The second statement was proved by Tuite through a case by case analysis, while we obtain a conceptual proof: our construction shows that the possible number of currents for a consistent orbifold $V^\natural/\left<g\right>$ is either $0$ or $24$, and in the latter case it is well-known that $V^{Leech}$ is the only possibility.

\vspace{.3cm}

\subsection{Outline}
\noindent Our paper is organized as follows. In section \ref{sec:setup} we describe the basic features of our Monstrous CHL-models, which form the core of the results in subsequent sections. In section \ref{s:indexI} we discuss the BPS-spectrum in our model, and give  the Fock space construction of the supersymmetric index $Z(T,U)$. In section \ref{s:oneloop} we define and evaluate a one-loop integral that reproduces the same index. We also provide an extensive analysis of the T-dualities satisfied by the index. In section \ref{s:Algebras}, we argue that each Monstrous CHL model contains an infinite dimensional Lie algebra of spontaneously broken gauge symmetries and show that this algebra is isomorphic to the corresponding Monstrous Lie algebra $\mathfrak{m}_g$. We identify each supersymmetric index $Z_{g,\id}$ with the algebraic index of the associated $\mathfrak{m}_g$, and show that it reproduces
the denominator formula.  
 In section \ref{sec:examples} we provide a number of detailed examples where we calculate the twisted index $Z_{g,\id}$ for elements of low order in $\mathbb{M}$ and explicitly verify our claims. In section \ref{sec:Haupt}  we combine all previous results to derive the genus zero properties of the McKay-Thompson series. Many technical details and proofs are relegated to the appendices.

\section{The setup}
\label{sec:setup}

\noindent In this section, we describe the main properties of the heterotic string compactifications that are the main focus of our paper.

\subsection{Monstrous heterotic string and CHL models}

The models we are interested in are certain compactifications of the heterotic strings to $0+1$ dimensions. We will define one such model for each element $g$ in the Monster group $\mathbb{M}$. When two elements $g$ and $g'$ are conjugated $g'=hgh^{-1}$ for some $h\in \mathbb{M}$, the corresponding models are equivalent and will be identified.

The starting point of our construction is a particular compactification of heterotic string to 1+1 dimensions. The internal CFT  of central charges $(c,\tilde c)=(24,12)$ factorizes as $V^\natural\times \bar V^{s\natural}$, where $V^\natural$ is the famous  Frenkel-Lepowsky-Meurman (FLM) Monster module \cite{FLM} and  $\bar V^{s\natural}$ is the Conway super-moonshine module first discussed by FLM in \cite{FLM} and constructed in \cite{duncan2007super} (see below for more details).\footnote{In our conventions, the anti-holomorphic (right-moving) side of the heterotic string has world-sheet supersymmetry.}  Compactifications of heterotic strings involving the FLM module $V^\natural$ have been considered before \cite{Harvey:1987da,Chaudhuri:1995ee,Green:1997gi,BergmanDistler}; in particular, the compactification on  $V^\natural\times \bar V^{s\natural}$ is discussed in \cite{BergmanDistler}.

\medskip

The FLM module $V^\natural$ is a  holomorphic bosonic conformal field theory, or vertex operator algebra (VOA), with central charge $c=24$. Its partition function is the $SL(2,\ZZ)$-invariant J-function with zero constant term
\be \Tr_{V^\natural}(q^{L_0-1})= J(\tau)=q^{-1}+0+196884q+\ldots\ , \qquad q:=e^{2\pi i \tau}\ .
\ee It is the only known (and, conjecturally, the unique) holomorphic CFT of central charge $c=24$ with no fields of conformal weight $1$ (currents). Its group of symmetries (i.e., linear transformations preserving the OPE, the vacuum, and the stress energy tensor) is isomorphic to the Monster group $\mathbb{M}$. It can be obtained starting from the Leech lattice CFT, i.e. the chiral half of the bosonic non-linear sigma model on the torus $\RR^{24}/\Lambda_{Leech}$, and then taking the $\ZZ_2$ orbifold under the symmetry that inverts the sign of all $24$ torus coordinates. Here, $\Lambda_{Leech}$ is the Leech lattice, the unique $24$-dimensional even unimodular lattice with no vectors of squared length $2$.   

\medskip

The right moving side of heterotic string, in the NS sector, is the (anti-)holomorphic  $\mathcal{N}=1$ superconformal field theory (super VOA) $\bar V^{s\natural}$ with $c=12$ studied in  \cite{duncan2007super}. It can be obtained as a $\ZZ_2$ orbifold of the $\mathcal{N}=1$ SCFT built in terms of the $E_8$ lattice, i.e. the chiral half of the supersymmetric non-linear sigma model with target space the torus $\RR^8/E_8$. The theory $V^{s\natural}$ is characterized as the unique holomorphic SCFT of $c=12$ with no fields of conformal weight $1/2$. Its group of automorphisms $\Aut(V^{s\natural})$ preserving the $\mathcal{N}=1$ superVirasoro algebra is the Conway group $Co_0$, though we will not need this property in the following. The right-moving Ramond sector of the heterotic string is the other unique irreducible module for the superVOA $\bar V^{s\natural}$. (By abuse of language, we will call these two modules the NS and the R sectors of $\bar V^{s\natural}$; one should keep in mind, however, that in the mathematical literature $\bar V^{s\natural}$ denotes only our NS sector.) The Ramond sector has $24$ ground states of conformal weight $1/2$ and positive fermion number.

\medskip

   The internal CFT has no direct geometric interpretation, i.e. it cannot be directly described as a non-linear sigma model on a compact manifold. However, it can be obtained as an asymmetric $\ZZ_2\times \ZZ_2$ orbifold of a compactification on a torus $\TT^8$. In the Narain moduli space parametrizing the geometry and the B-field of the torus $\TT^8$, there is a unique point where the non-linear sigma model factorizes as a product of the holomorphic Leech lattice CFT  and the anti-holomorphic $E_8$ lattice SCFT.  One considers the orbifold of this model by the $\ZZ_2\times \ZZ_2$ symmetries that flip the signs of the $24$ left-moving and, independently, of the $8$ right-moving scalar (super)fields in these theories. 

The compactification of heterotic strings on $\TT^8$ yields a 1+1 dimensional theory with $(8,8)$ space-time supersymmetry. The $\ZZ_2$ orbifold acting on the left-moving (bosonic) side preserves all such supersymmetries, while the $\ZZ_2$ orbifold acting on the right-moving (supersymmetric) sector breaks half of them, down to $(0,8)$. However, including the twisted sector introduces $16$ additional supersymmetries with the same space-time chirality, so that the theory we are considering has $(0,24)$ supersymmetry \cite{BergmanDistler}.

\bigskip

Starting from this heterotic compactification, we will now construct a host of $0+1$ dimensional models (i.e., supersymmetric quantum mechanics) with $24$ supersymmetries, by first compactifying one further space direction on a circle $S^1$ of radius $R$,  and then taking a CHL-like orbifold.  More precisely, we consider the orbifold of heterotic strings on $S^1\times (V^\natural\times \bar V^{s\natural})$ by a $\ZZ_N$ symmetry $(\delta,g)$, where $\delta$ is a shift of $1/N$ of a period along the $S^1$ circle, and $g\in \mathbb{M}$ is a symmetry of the left-moving internal CFT $V^\natural$. This is analogous to the standard CHL construction \cite{Chaudhuri:1995fk,Chaudhuri:1995bf,Chaudhuri:1995dj,Chaudhuri:1995ee}. Many of the models constructed in this way are actually equivalent to each other. In fact, up to equivalence, the CHL models only depend on the conjugacy class of the cyclic subgroup $\langle g\rangle\subset \mathbb{M}$ (although, we will usually denote them simply by the generator $g$). This follows from the fact that $V^\natural$ is invariant under charge conjugation, and that any power $g^a$ of $g$, with $a$ coprime to the order $N$ of $g$, is conjugated with either $g$ or $g^{-1}$ within the Monster group.

\bigskip

The CHL construction outlined above is consistent only for those $g\in \Aut(V^\natural)$ that satisfy the level-matching condition, i.e. such that the conformal weights in the $g$-twisted sector $V^\natural_g$ of $V^\natural$ take value in $\frac{1}{N}\ZZ$, where $N$ is the order of $g$. In general, the conformal weights of a $g$-twisted state takes values in
\be \frac{\E_g}{N\lambda}+\frac{1}{N}\ZZ\ ,
\ee where $\lambda$ is a positive integer (depending on $g$) dividing both $N$ and $24$ (see Appendix \ref{a:TwistAndTwin} for the proof of the latter), and $\E_g \in (\ZZ/\lambda\ZZ)^\times $ is an integer defined modulo $\lambda$ and coprime with $\lambda$. Here, $\lambda$ is also the order of the multiplier system of the McKay-Thompson series $T_g$. Even when $\lambda>1$, a consistent CHL orbifold can be constructed: it is sufficient to take a symmetry $(\delta,g)$ with a shift $\delta$ of order $N\lambda$ rather than $N$. We refer to \cite{Persson:2015jka} and Appendix \ref{a:TwistAndTwin} for more details.

\subsection{The dilaton and other moduli}\label{sec:dilaton}
The spectra of these ``Monstrous CHL'' models will be discussed in some detail in the next sections. However, one striking feature of these models deserves to be stressed: there are no local massless degrees of freedom \cite{BergmanDistler,Green:1997gi,Harvey:1987da}! This is most easily understood in the light-cone quantization, where massless string states correspond to states with conformal weight $1$ (currents) in the internal bosonic CFT. However, as stressed in the last section, $V^\natural$ has no currents; furthermore, no massless states can be introduced in the orbifold by $(\delta,g)$, since the strings in the twisted sectors have non-zero fractional winding along $S^1$, so that they are necessarily massive.\footnote{This argument fails in the limit $R\to 0$. Indeed, we will see in the following sections that, in some CHL models, massless states can appear in this limit.}  Therefore, all physical states in the light-cone quantization must be massive.

In particular, as noticed in \cite{BergmanDistler,Green:1997gi,Harvey:1987da},  there are no moduli and all parameters of the theory seem to be completely fixed, including the string coupling constant $g_s$. This is puzzling, as the $\ZZ_2\times \ZZ_2$ orbifold procedure leading from the compactification on $\TT^8$ to the model we are considering seems to be perfectly consistent for all (small) values of the string coupling constant. It is not clear what kind of mechanism could fix $g_s$ to a specific value. On the other hand, the alternative idea that the coupling constant is a free parameter not related to any string background seems to be at odds with all we know about string theory.

A somehow analogous issue seems to arise for the radius $R$ of the compactification circle. In this case, however, the resolution is quite clear: while the gauge (specifically, diffeomorphism) invariance in two dimensions is more than sufficient to fix the metric and eliminate all \emph{local} degrees of freedom, the length of the geodesic along the circle is gauge invariant and therefore has physical effects. The possibility of a residual global degree of freedom of zero measure that cannot be fixed by a gauge transformation is a well known phenomenon, occurring, for example, in the gauge fixing of string theory at genus higher than zero. Even in this case, however, it is puzzling that there is no physical state in the string theory corresponding to deformations of $R$.

\bigskip

We propose that both these puzzles can be solved by a more careful treatment of the physical states at zero-momentum. Recall that the  light-cone quantization can be shown to be equivalent to BRST only for non-zero momentum states $k^\mu\neq 0$. At zero momentum, the light-cone gauge is not a good gauge choice, and one has to apply a BRST quantization procedure. 

For simplicity, let us consider the bosonic string compactified on $V^\natural\times\bar V^\natural$, where the same issues appear.
In the BRST formalism, the physical states in closed string theory correspond to the semi-relative BRST cohomology, i.e. the BRST cohomology on the complex of states satisfying $(b_0-\bar b_0)|\psi\rangle=0$. At zero momentum, the BRST cohomology includes all states of the form
\be D^{\mu\nu}:=c_1\alpha^{\mu}_{-1}\bar c_1\bar\alpha^\nu_{-1}|0\rangle\ ,
\ee where $\alpha^\mu_n,\bar\alpha^\nu_n$ are the standard bosonic oscillators in the $d$ uncompactified directions and $c_n,\tilde c_n$ are the ghosts, as well as the \emph{ghost dilaton} \cite{Distler:1991au,Bergman:1994qq,Rahman:1995ee,Belopolsky:1995vi}
\be D_g= (c_1 c_{-1}-\bar c_1 \bar c_{-1})|0\rangle\ .
\ee Notice that all $d^2$ states $D^{\mu\nu}$ are physical at zero momentum; the corresponding string background determines the global geometric properties of our two dimensional space-time, such as the radius $R$. 

The ghost dilaton $D_g$ is the BRST variation of $\chi:=(c_0-\bar c_0)|0\rangle$. However, it is not BRST exact in the semi-relative complex, since $\chi$ is not a `legal' state in this complex; that is, $(b_0-\bar b_0)\chi\neq 0$. 
The zero-momentum limit of the physical dilaton field is the following linear combination \cite{Belopolsky:1995vi}
\be D:=\eta_{\mu\nu}D^{\mu\nu}-D_g=D_m-D_g\ ,
\ee of the ghost dilaton and of the so-called matter dilaton
\be D_m:=\eta_{\mu\nu}D^{\mu\nu}\ .
\ee Indeed, this is the linear combination that transforms as a scalar under gauge transformations. Another interesting linear combination is the trace
\be \mathcal{G}:=\eta_{\mu\nu}\mathcal{G}^{\mu\nu}=D_m-\frac{d}{2}D_g\ ,
\ee of the (Einstein frame) graviton
\be \mathcal{G}^{\mu\nu} =D^{\mu\nu}-\frac{1}{2}\eta^{\mu\nu}D_g\ ,
\ee where $d$ is the number of uncompactified space-time dimensions.

In the framework of closed string field theory, it has been shown that, for any number of space-time dimensions $d$,  a change in the ghost dilaton background $D_g$ has the effect of shifting the string coupling constant \cite{Bergman:1994qq,Rahman:1995ee}. On the other hand, a change of the graviton trace background $\mathcal{G}$ only corresponds to a field redefinition and has no observable physical effect \cite{Belopolsky:1995vi}. In the usual case where $d>2$, this implies that the string coupling constant is determined by the background for the dilaton $D$ (the zero-momentum limit of the physical dilaton field) or, equivalently, by the background for the ghost dilaton $D_g$.   In the case we are considering, where there are only $2$ non-compact directions ($d=2$), the dilaton $D$ coincides with the graviton trace $\mathcal{G}$ and therefore has no observable physical effect. However, a  ghost dilaton background still makes perfect sense in our theory and has the effect of shifting the string coupling constant. 
 
 While this reasoning has been derived in the context of bosonic strings, analogous arguments should hold for the Monstrous heterotic CHL models. We conclude that also in these theories, we are free to set the string coupling constant to any particular value, since this corresponds to the choice of the string (ghost) dilaton background, as usual. In particular, in the next sections, we will consider the Monster CHL models in the free theory limit $g_s\to 0$.

\section{The space-time index I: Fock space construction}\label{s:indexI}

\noindent We shall now consider the Fock space construction of the supersymmetric index $Z$. We do this in two steps: first, we study the case without any CHL-orbifold and then then we generalize this to the twisted indices associated with the Monstrous CHL models. 
\subsection{The untwisted case}
Let us consider compactification of the heterotic string on $V^\natural\times \bar V^{s\natural}$ to two space-time dimensions (without CHL orbifold, to start with). We consider a space-time with a flat metric of Lorentzian signature and the topology of a cylinder $S^1\times \RR$, where $\RR$ is the time direction and $S^1$ is a space-like circle of radius $R$. The theory has $(0,24)$ space-time supersymmetries \cite{BergmanDistler} $Q^i$, $i=1,\ldots,24$, with algebra
\be\label{SUSY} \{Q^i,Q^j\}=\delta^{ij} (P^0_R-P^1_R)\ ,
\ee where $(P^0_R,P^1_R)$ are the contributions to the space-time momenta coming from the world-sheet right-moving sector. This algebra has two kinds of supermultiplets: short (BPS) supermultiplets, that are 1-dimensional, and whose states satisfy 
\be\label{BPS} k^0_R=k^1_R\ ,
\ee  which is essentially a BPS condition for the algebra \eqref{SUSY} \footnote{We denote by $k^0_{L,R}$, $k^1_{L,R}$ the eigenvalues of $P^0_{L,R}$, $P^1_{L,R}$.}; and long supermultiplets with $k^0_R>k^1_R$ of dimension $2^{12}$ and containing half fermions and half bosons.

We want to consider a Hilbert space $\Hh$ corresponding to the `second quantization' of this string theory, i.e. including any number of fundamental strings. In this theory we consider the refined supersymmetric index
\be Z(\beta,b,v, R)=\Tr_{\Hh}(e^{-\beta H} e^{2\pi ibW}e^{2\pi ivM} (-1)^{F})\ ,
\ee where $H$ is the space-time Hamiltonian, $F$ is the space-time fermion number, and $W$ and $M$ are the total winding and momentum numbers along the circle $S^1$. We will use lowercase letters $w$ and $m$ to denote the winding and momentum of a single fundamental string. Here, $\beta$ is the inverse temperature, and $b$ and $v$ are chemical potentials conjugate to the quantum numbers $W$ and $M$ ($b$ can be interpreted as a background $B$-field and $v$ as an off-diagonal component of the space-time metric). We impose periodic boundary conditions for the fermions around the circle $S^1$. The space $\Hh$ carries a representation of the supersymmetry algebra \eqref{SUSY} and, by the usual index arguments, the only states contributing to this trace are the ones in short (BPS) supermultiplets. Thus, we can reduce the trace to the BPS subspace $\Hh_{BPS}$.

Let us consider the $1$-particle BPS states, obtained through a light-cone quantization of the string theory. The mass-shell and level-matching conditions read
\be\label{physical} \begin{cases} 0=-\frac{1}{2}(k^0_L)^2+\frac{1}{2}(k^1_L)^2+h_L-1\ ,\\
0=-\frac{1}{2}(k^0_R)^2+\frac{1}{2}(k^1_R)^2+h_R-\frac{1}{2}\ ,\end{cases}
\ee where $(h_L,h_R)$ are the $(L_0,\bar L_0)$-eigenvalues of the state in the internal CFT $V^\natural\times \bar V^{s\natural}$. Since there is no winding around the time direction, we have
\be k^0_L=k^0_R=E\ ,
\ee where $E$ is the eigenvalue of the space-time Hamiltonian $H$. By imposing the BPS condition \eqref{BPS}, and using the relations
\be k^1_L=\frac{1}{\sqrt{2}}\left(\frac{m}{R}-wR\right)\ ,\qquad k^1_R=\frac{1}{\sqrt{2}}\left(\frac{m}{R}+wR\right)
\ee
 we obtain
\be \begin{cases} 0=-\frac{1}{2}(k^1_R)^2+\frac{1}{2}(k^1_L)^2+h_L-1=-mw+h_L-1\ ,\\
0=h_R-\frac{1}{2}\ .\end{cases}
\ee The only states satisfying $h_R=1/2$ are the Ramond ground states in $\bar V^{s\natural}$. If we set
\be J(\tau)=\sum_{n = -1}^\infty c(n) q^{n}=q^{-1}+0+196884q+\ldots\ ,
\ee then for each $w,m\in\ZZ$ there are $24c(mw)$ fermions
carrying winding $w$ and momentum $m$ along $S^1$ and with energy
\be\label{EnBPS} E=k^1_R=\frac{1}{\sqrt{2}}\left(\frac{m}{R}+wR\right)\ .
\ee 
  In a free theory limit, we can think of the second quantized BPS Hilbert space $\Hh_{BPS}$ as a Fock space built in terms of fermionic oscillators corresponding to the $1$-particle BPS states. In particular, states with energy $E>0$ (respectively, $E<0$) are interpreted as creation (respectively, annihilation) operators.  Notice that
\be c(mw)>0\qquad \Rightarrow\qquad  mw=-1\text{ or } mw>0\ ,
\ee so that the condition $E>0$ implies 
\be m,w> 0\quad \text{or}\quad \begin{cases} m=1,w=-1 & \text{if }R<1\\ m=-1,w=1 & \text{if }R>1\ .
\end{cases}
\ee
The ground state is, by definition, the unique (for $R\neq 1$) state in $\Hh$ that is annihilated by all operators with $E<0$. Notice that this definition depends on the radius $R$. When $R=1$ there are additional zero energy fermionic oscillators and the ground state is degenerate. The space $\Hh_{BPS}$ is constructed by acting on the vacuum in all possible ways with creation operators.  Let us focus on the case  $R>1$, for definiteness; the case $R<1$ is analogous. The relation \eqref{EnBPS} generalizes by linearity to the relation 
\be\label{HamBPS} H=\frac{1}{\sqrt{2}}\left(\frac{M}{R}+WR\right)\ ,
\ee between operators on the BPS space $\Hh_{BPS}$,
so that 
\be\label{ZbetabvR}  Z(\beta,b,v, R)=e^{24(-\beta E_0+iv m_0+ib w_0)}\prod_{\substack{w> 0\\m\in \ZZ}}(1-e^{-\frac{\beta}{\sqrt{2}}(\frac{m}{R}+wR)}e^{2\pi i b w}e^{2\pi ivm})^{24c(mw)}\ ,
\ee where we included the possibility of vacuum momentum $24m_0$, winding $24w_0$ and energy $24E_0=\frac{24}{\sqrt{2}}(\frac{m_0}{R}+w_0R)$.  It is useful to introduce the complex parameters
\be T=b+i\frac{\beta R}{2\sqrt{2}\pi}\qquad U=v+i\frac{\beta }{2\sqrt{2}\pi R}\ ,
\ee so that
\be Z(T,U)=e^{24(2\pi i (m_0U+w_0T))}\prod_{\substack{w> 0\\m\in \ZZ}}(1-e^{2\pi i Um}e^{2\pi iTw})^{24c(mw)}\ .
\ee 
In section \ref{s:oneloop}, we will provide a formula for $Z(T,U)$ (or, rather, its absolute value) in terms of a string $1$-loop path integral. In this context, the time direction is Wick-rotated to a Euclidean time compactified on a thermal circle of radius $\beta$, so that the space-time becomes a Euclidean torus $\TT^2$. The complex parameters $U$ and $T$  are then identified with the complex structure and the (complexified) K\"ahler structure moduli of $\TT^2$.

The vacuum winding and momentum will be computed in section \ref{s:Algebras} to be $(w_0,m_0)=(-1,0)$, so that
\be Z(T,U)=\Bigl(e^{-2\pi i T}\prod_{\substack{w> 0\\m\in \ZZ}}(1-e^{2\pi i Um}e^{2\pi iTw})^{c(mw)}\Bigr)^{24}\ .
\ee Apart from the exponent $24$, this is exactly the product formula for Borcherds' Monstrous Lie algebra! This is not an accident: we will show in section \ref{s:Algebras} that the (first quantized) string BPS states are a representation over this algebra.

\bigskip

As stressed above, when the radius $R$ is varied continuously from $R>1$ to $R<1$, the energy of two fermionic operators change sign (an annihilation operator becomes a creation operator and vice-versa), so that the vacuum state changes. Thus, one might expect a discontinuity of $Z(T,U)$ as one crosses the line $R=1$. Furthermore, the infinite product above is expected to converge only for sufficiently large $\beta$. However, in the alternative derivations of the index $Z$ in the following sections, it will be clear that there is no discontinuity as the radius crosses the line $R=1$ and that $Z(T,U)$ is an analytic function of $T$ and $U$ over all the upper half of the complex plane.

\subsection{The twisted case}

There are two `twists' of the previous construction that will be interesting for us. The simplest modification is to consider 
\be Z_{\id,g}(\beta,b,v, R):=\Tr_{\Hh}(g\,e^{-\beta H} e^{2\pi ibw}e^{2\pi ivm} (-1)^{F})
\ee where we insert an element $g\in\MM$ of the Monster group inside the trace. More precisely, the fermionic operators carry an action of the Monster group $\MM$ and this determines a representation of $\MM$ over $\Hh_{BPS}$ preserving $H$, $M$ and $W$. In particular, the vacuum state is invariant under this action, since the only $1$-dimensional representation of $\MM$ is the trivial one. 

The second modification of the index is to consider the `second quantized' BPS space $\Hh_{CHL(g)}$ constructed from the CHL model associated with an element $g\in \MM$, i.e.
\be Z_{g,\id}(\beta,b,v, R):=\Tr_{\Hh_{CHL(g)}}(e^{-\beta H} e^{2\pi ibW}e^{2\pi ivM} (-1)^{F})\ .
\ee We will show that for $g$ of order $N$ (with $\lambda=1$), these twisted indices take the form
\be\label{infiniteprod}
Z_{g,\id}(T,U)=\Bigl(e^{-2\pi i  T}\prod_{\substack{n> 0\\m\in \ZZ}}(1-e^{2\pi i U\frac{m}{N}}e^{2\pi iTn})^{\hat c_{n,m}(\frac{mn}{N})}\Bigr)^{24}\ ,
\ee where the $\hat c_{n, m}(\frac{m n}{N})$ is the dimension of the $e^{2 \pi i m/N}$-eigenspace of the $g^n$-twisted sector $V^\natural_{g^n}$ at level $L_0-1=\frac{m n}{N}$, c.f. equation \eqref{defF}. We will denote these graded spaces of states by
\be \label{eq:Vdef}
V^\natural_{n, m} = \left\lbrace v \in V^\natural_{g^n} \vert g(v) = e^{2 \pi i \frac{m}{N}} \right\rbrace, \quad n, m \in \mathbb{Z}/(N \mathbb{Z}).
\ee Notice that, by definition, $\hat c_{n, m}(\frac{n m}{N})$ are always nonnegative integers. We will also show that similar equations hold in the case $\lambda \neq 1$. The indices $Z_{g,\id}$ will be the main subject of our investigation.

\subsection{Coupling with gravity?}\label{sec:anomaly}

Let us critically reconsider the construction of sections 3.1 and 3.2. We have considered the physical string states arising from the light-cone quantization of the Monstrous CHL models and from these built a `second quantized' Fock space of states, which describes the spectrum of an arbitrary number of free strings. We have taken the strings to propagate in a fixed geometric background, neglecting any backreaction of the strings on the space-time metric or B-field. Consistent with taking a non-dynamical background, we have ignored the zero-momentum string modes that appear in the BRST quantization of string theory, which would be associated to background fluctuations. This decoupling is consistent as long as the string coupling constant is strictly zero, which we have assumed throughout our computations. From a different viewpoint, the limit $g_s\to 0$ pushes the Planck scale $M_{Planck}$ much higher than the string scale $M_{string}$, so that it makes sense to study the theory at energies in  some intermediate region $M_{string}\ll E\ll M_{Planck}$. This is a convenient set-up for studying the symmetries of the spectrum of the CHL models, as we will do in the following sections. 

\bigskip

It is natural to ask if one can turn the string coupling constant on, so as to consider a system of `second quantized' interacting strings coupled to a dynamical background. Unfortunately, this does not seem to lead to a consistent physical model. The basic reason is that, in a theory with dynamical background fields, the number of spacetime-filling strings cannot be arbitrary, but is fixed by the requirement of anomaly (or tadpole) cancellation\footnote{This is completely analogous to the familiar restriction on the number of spacetime-filling D9-branes in the 10-dimensional type IIB superstring.}. In two-dimensional heterotic compactifications, there is a potential 1-point function for the B-field arising from the compactification of the 10-dimensional Green-Schwarz term. This term in the effective action arises from a 1-loop string amplitude and its coefficient can be computed using the techniques in \cite{Lerche:1987qk,Vafa:1995fj} to be
\be
\frac{1}{24} (24J(\tau)E_2(\tau))_{q^0}= \frac{1}{24} \bigl((24q^{-1}+O(q))(1-24 q+O(q^2))\bigr)_{q^0}=-24\ .
\ee Here, $E_2=1-24q+\ldots$ is the Eisenstein series of weight $2$ and the notation $(\cdot )_{q^0}$ denotes the constant term in the Fourier expansion of the modular form. This result can also be understood in terms of local gravitational anomalies for the two-dimensional effective theory: the 1-point function for the B-field is necessary in order to cancel the contributions to the  anomalies from the $24$ chiral gravitini and dilatini. The tadpole makes the background unstable and the theory inconsistent, but there is a standard procedure to cancel it: one has to insert $24$ spacetime-filling fundamental strings, which couple to the B-field and thereby add the required positive term to the one-point function. Equivalently, the contributions to the gravitational anomalies from the degrees of freedom of the spacetime-filling heterotic strings cancel the ones from the gravitini and dilatini. 

The calculation of the tadpole for the B-field also suggests that the correct ground state winding in our second quantized string theory is $24w_0=-24$. A similar analysis shows that there is no tadpole for the off-diagonal component of the metric, which implies the ground state momentum $m_0=0$.

\bigskip

The outcome of this analysis is that the only sector of our second quantized string theory that can be consistently coupled with gravity and a dynamical B-field is the one with $24$ spacetime-filling strings, i.e. the sector with total winding number $W=0$. We stress once again that it is formally correct to consider the full second quantized free strings in a fixed, non-dynamical background. In a sense, the situation is similar to the case of a two- or six-dimensional quantum field theory whose matter content contributes to a non-vanishing gravitational anomaly: the theory is perfectly consistent as long as it is is decoupled from gravity! The results of the following sections provide strong support in favour of this viewpoint.  On the other hand, the inconsistency of the full construction when the string coupling constant is non-vanishing makes our physical interpretation of the moonshine phenomenon not completely satisfactory. We hope we will be able improve this point in future publications.

\section{The space time index II: 1-loop integral}\label{s:oneloop}
\noindent We shall now define and evaluate a one-loop integral in the Monstrous CHL-models. The result of this integral reproduces the same supersymmetric index as was calculated in the 
previous section using completely different methods. The main benefit of the present, path-integral approach is that T-duality becomes manifest. After evaluating the integral we analyze its 
T-duality symmetries in great detail, revealing that the T-duality groups are directly related to the moonshine groups $\Gamma_g$.

\subsection{The 1-loop integral and the GSO projection}
The supersymmetric index $Z_{g,\id}(T,U)$ is a refined partition function at finite (inverse) temperature $\beta$. In general, one expects any such partition function to be given by a suitable Euclidean path-integral with Euclidean time periodically identified with period $\beta$.   In our context, we need to consider our Monstrous CHL-models with the two space-time directions on a Euclidean torus $\TT^2$ with complex modulus $U$ and K\"ahler modulus $T$. The index $Z_{g,\id}(T,U)$ should be obtained by a path-integral
\be Z_{g,\id}(T,U)=e^{-(S_{tree}+S_{1-loop}+\ldots)}\ ,
\ee where $S_{\ell-loop}$ is the string $\ell$-loop contribution.
In general, each loop contribution is weighted by a power $g_s^{2-2\ell}$ of the string coupling constant $g_s$. Since we are considering a free theory $g_s=0$,  it is natural to expect the path-integral should be one-loop exact. The one loop contribution is given by the standard string path integral on a torus
\be S_{1-loop}= \frac 1 2\int_{\F} \frac{d^2\tau}{\tau_2^2}   \Bigl(\Tr_{NS}(q^{L_0-\frac{c}{24}}\bar q^{\bar L_0-\frac{\bar c}{24}} P_{GSO})-\Tr_R(q^{L_0-\frac{c}{24}}\bar q^{\bar L_0-\frac{\bar c}{24}}P_{GSO})\Bigr)\ee
i.e. a trace over the (GSO projected) full space of states, with different signs for the Ramond and Neveu-Schwarz sector to take the space-time fermion number into account. The GSO projection is rather subtle for fermions with $k_R^2=0$. In two dimensions, the massless Dirac equation relates the spin of the state with the sign of its momentum $k_{1,R}$. Therefore, for massless fermions, the GSO projection is implemented by including states with either positive or negative transverse (internal) fermion number, depending on whether $k_{0,R}=k_{1,R}$ or $k_{0,R}=-k_{1,R}$.  In our specific case, there are $24$ internal Ramond ground states with \emph{internal} fermion number $(-1)^{\bar F}=+1$ and no states with $(-1)^{\bar F} =-1$. Therefore, a properly implemented GSO projection should include $24$ states with $k_{0,R}=k_{1,R}$ and no states with $k_{0,R}=-k_{1,R}$.

In practice, it is however very difficult to implement the GSO-condition directly in the 1-loop path integral. We circumvent this by the following trick. Consider instead the naive path integral
\be S_{1-loop}^{\text{naive}}=\frac{1}{2}\int_{\F}\frac{d^2\tau}{\tau_2^2}\Big[ \Tr_{NS}(q^{L_0-1}\bar q^{\bar L_0-\frac{1}{2}}\frac{1-(-1)^{\bar F}}{2})-\Tr_{R}(q^{L_0-1}\bar q^{\bar L_0-\frac{1}{2}}\frac{1+(-1)^{\bar F}}{2})\Big]\ ,
\label{naivepathint}
\ee  
where the trace is taken over the holomorphic and anti-holomorphic `transverse' CFTs and over the winding-moments in the light-cone directions.
The error we introduce in this way corresponds to the contribution of $24$ copies of each massless (i.e., $k_R^2=0$) fermion with the wrong chirality, i.e.
\be  S_{1-loop}^{\text{naive}}=S_{1-loop}^{\text{true}}+\PP(S_{1-loop}^{\text{true}})\ ,
\ee where $\PP$ is the parity transformation in the space direction.
Let us focus on the case $g=\id$ for clarity. We know that $Z(T,U)$ is given by the exponential $\exp (-S_{1-loop}^{\text{true}})$ of the correct 1-loop path integral. Therefore, the exponential of the `naive' 1-loop contribution corresponds to
\be \exp (-S_{1-loop}^{\text{naive}})=Z(T,U)\times \PP(Z(T,U))\ .
\ee 
Parity reversal changes the sign of momentum and winding in the space direction
\be M\to -M\qquad\qquad W\to -W\ ,
\ee while leaving the Hamiltonian $H$ fixed. From the formula
\be Z(R,\beta,b,v)=\Tr_{\mathcal{H}_{BPS}}((-1)^Fe^{-\beta H}e^{2\pi ibW}e^{2\pi ivM})\ ,
\ee we see that the parity transformation corresponds to
\be \PP(Z(R,\beta,b,v))= Z(R,\beta,-b,-v)\ ,
\ee which, in turn, is equivalent to
\be\label{parity} \PP(Z(T,U))= Z(-\bar T,-\bar U)=\overline{Z(T,U)}\ .
\ee
To conclude, if we consider the `naive' 1-loop integral
 then we have
\be e^{-S_{1-loop}^{\text{naive}}}=|Z(T,U)|^2\ .
\ee Analogous results hold for the cases $g\neq \id$. Therefore, the `correct' GSO projection is simply obtained by picking the holomorphic part of $e^{-S_{1-loop}^{\text{naive}}}$. We will drop the superscript `naive' from now on.

\subsection{Evaluating the 1-loop integral}

In this section, we compute the 1-loop path integral \eqref{naivepathint} explicitly. Let us first focus on the unorbifolded case $g=\id$. The trace factorizes into the product of three contributions from $\bar V^{s\natural}$, from $V^\natural$ and from the winding and momenta along $\TT^2$. The contribution of the oscillators along $\TT^2$ and from the ghosts and superghosts cancel each other, as usual. First  note that  $V^{s\natural}$ happens to have the nice property that  \cite{BergmanDistler}
\be  \Tr_{\bar V^{s\natural},NS}(\bar q^{\bar L_0-\frac{\bar c}{24}}\frac{1-(-1)^{\bar F}}{2})- \Tr_{\bar V^{s\natural},R}(\bar q^{\bar L_0-\frac{\bar c}{24}}\frac{1+(-1)^{\bar F}}{2})=-24\ .
\ee
This completely takes care of the trace over $\bar V^{s\natural}$ in \eqref{naivepathint}.

In the untwisted case ($g=\id$), the sum over winding and momenta along $\TT^2$ is the usual   theta function $\Theta_{\Gamma^{2,2}}(T,U,\tau)= \sum_{(k_L, k_R) \in \Gamma^{2, 2}} q^{k_L^2/2}\bar{q}^{k_R^2/2}$, where  $\Gamma^{2, 2}$ (the Narain lattice) is the even unimodular lattice with signature $(2,2)$.
The rest of the left-moving contribution comes from the Monster module $V^\natural$. As worked out carefully in e.g. \cite{DGH}, the theory simply has chiral partition function $J(\tau)$. 

For our more general CHL models, the computation is slightly more complicated. One needs to sum over the $(\delta,g)$-twisted sectors and then project over the $(\delta,g)$-invariant states.
Putting everything together the 1-loop integral \eqref{naivepathint} can now be written in the following explicit form
\be 
S_{1-loop}=-\frac{24}{N}\int_{\mathcal{F}} \frac{d^2\tau}{2\tau_2}\sum_{ r,s=1}^N \Theta^{\Gamma^{2,2}}_{r,s}(T,U,\tau)T_{g^r,g^s}(\tau),
\label{finaloneloop}
\ee
where $r$ labels the different $(\delta,g)$-twisted sectors, $\frac{1}{N}\sum_{s=1}^N$ projects over the $(\delta,g)$-invariant states, and $T_{g^r,g^s}(\tau)$ is the $g^r$-twisted $g^s$-twined partition function of the Monster CFT:
\be T_{g^r,g^s}(\tau)=\Tr_{V^\natural_{g^r}}(g^sq^{L_0-1}).
\ee
 We have also defined the shifted theta series by
\be \Theta^{\Gamma^{2,2}}_{r,s}(T,U,\tau)= \sum_{\lambda\in r\delta+\Gamma^{2,2}}e^{2\pi i s\delta\cdot \lambda} q^{\frac{k_L^2}{2}}\bar q^{\frac{k_R^2}{2}}=\sum_{m_1,m_2,w_2\in \ZZ}\sum_{w_1\in \frac{r}{N}+\ZZ}
e^{-2\pi i \frac{sm_1}{N}} q^{\frac{k_L^2}{2}}\bar q^{\frac{k_R^2}{2}}
\ee 
where
\begin{align}
k_L^2&=\frac{\left|\begin{pmatrix}
T & \frac{1}{N}
\end{pmatrix}\begin{pmatrix}
w_2 & w_1\\ -m_1 & m_2
\end{pmatrix}\begin{pmatrix}
-U \\ N
\end{pmatrix}
\right|^2}{2T_2U_2}\\
k_R^2&=k_L^2-2m_1 w_1-2 m_2 w_2\ .
\end{align}
The lattice $L$ of winding and momenta is the union $L=\bigcup_{r=1}^N (r\delta +\Gamma^{2,2})$, where $r$ labels the different twisted sectors.

Following Harvey-Moore \cite{Harvey:1995fq}, Borcherds \cite{MR1625724} developed a method for calculating general integrals of the form
\be
\Phi(M, F):=\int_{\mathcal{F}} \frac{dx dy}{y}\big( \overline{\Theta}_{M}(z), \, F(z)\big) ,
\label{thetalift}
\ee
where $z=x+iy\in \mathbb{H}$, $M$ is a lattice, $F$ is a weight $k$ vector-valued modular form (valued in the group ring $\mathbb{C}[M^{\vee}/M]$) and $\overline{\Theta}_{M}$ is a weight $-k$ vector-valued 
Siegel theta series for the lattice $M$. The notation $(\, , \, )$ denotes the scalar product 
\be
(e_\gamma, e_{\gamma'})=\delta_{\gamma+\gamma',0}.
\ee
in the vector space freely generated by
$e_\gamma$, $\gamma\in M^{\vee}/M$.
 In terms of the basis $e_\gamma$ a vector-valued modular function $F$ for a congruence subgroup $\Gamma\subset SL(2,\mathbb{Z})$ can be written as
\be
F(z)=\sum_{\gamma\in M^{\vee}/M} F_\gamma(z)e_\gamma,
\ee
where the components $F_\gamma(z)$ are modular functions for $\Gamma$, transforming in the metaplectic representation of (the double cover of) $SL(2,\mathbb{Z})$ on $\mathbb{C}[M^{\vee}/M]$.
In a similar vein, one defines the vector-valued theta series 
as
\be
\Theta_{M}(\tau)=\sum_{\gamma\in M^{\vee}/M} \theta_{M+\gamma} e_\gamma, 
\ee
where $\theta_{M+\gamma}$ is the ordinary ``shifted'' Siegel theta series. 

We now want to relate the general integral $\Phi(M, F)$ to our one-loop integral \eqref{finaloneloop}. Let us focus on the case $\lambda=1$ for simplicity. In our case the lattice $M$ can be identified with (c.f. Eqn. \ref{vector}) 
\be
M\equiv L^\vee \cong N\mathbb{Z}\oplus \mathbb{Z}\oplus \mathbb{Z}\oplus \mathbb{Z}, 
\ee
which is the dual of the winding-momentum lattice $L$. 

To construct the vector-valued modular form $F_g$ we take the discrete Fourier transform of the generalized moonshine functions \cite{Carnahan2014}:
\be
F_{\gamma}(z)=F_{l,k}(z)=\frac{1}{N}\sum_{j\in \mathbb{Z}/N\mathbb{Z}} e^{-2 \pi i \frac{jk}{N}}T_{g^{l}, g^{j}}(z) =: \sum_{n \in \ZZ/N\ZZ}\hat c_{l, k}(n) e^{2 \pi i z n}
\label{defF}
\ee
Notice that the $F_{l, k}$ are the generating functions for the graded dimensions of $V_{l,k}$, c.f. equation \eqref{eq:Vdef}.
Thus
\be
F_g(z)=\sum_{l,k\in \mathbb{Z}/N\mathbb{Z}} F_{l,k}(z)e_{l,k} 
\ee
is a vector-valued modular form of weight $0$. Similarly, the theta function for the lattice $M\equiv L^{\vee}$ is 
\be
\Theta_{L^{\vee}}(\sigma,\tau,z)=\sum_{\gamma\in L/L^{\vee}} \theta_{L^{\vee}+\gamma}e_\gamma,
\ee
and the components are related to our  Narain theta function by:
\be
\theta_{L^{\vee}+\gamma}(\sigma, \tau, z)=\frac{1}{N}\sum_{s=1}^Ne^{\frac{2\pi i st}{N}}\Theta^{\Gamma^{2,2}}_{r,s}(T,U,\tau).
\ee

With this choice of data the integrand in the theta lift $\Phi(L^\vee, F_g)$ can be written out explicitly as 
\bea
\left(\overline{\Theta}_{L^{\vee}}, F_g\right)&=&\sum_{\gamma\in L/L^{\vee}} \sum_{\delta \in L/L^{\vee}} {\theta}_{L^{\vee}+\delta} F_\gamma \delta_{\gamma-\delta,0}
\nonumber \\
&=& \sum_{l,k\in \mathbb{Z}/N\mathbb{Z}} {\theta}_{L^{\vee}+(l,k)} F_{l,k}
\nonumber \\
&=& \frac{1}{N} \sum_{l,j,k\in \mathbb{Z}/N\mathbb{Z}} e^{-2\pi i \frac{jk}{N}} T_{g^{l}, g^{j}} \theta_{L^{\vee}+(l,k)}
\nonumber \\
&=& \frac{1}{N}\sum_{l,j\in \mathbb{Z}/N\mathbb{Z}} \Theta_{l,j}^{\Gamma^{2,2}}(T, U, z) T_{g^{l}, g^{j}}(z),
\label{integrand}
\eea
where we used the relation
\be
\Theta_{l,j}^{\Gamma^{2,2}} =\sum_{k\in \mathbb{Z}/N\mathbb{Z}} e^{-2 \pi i jk/N} \theta_{L^{\vee}+(l,k)}.
\ee
Thus, we conclude that our integral \eqref{finaloneloop} is indeed of the type \eqref{thetalift} and may be evaluated using the methods of Borcherds \cite{MR1625724}. Omitting the details, we find that our one-loop integral \eqref{finaloneloop}, for large enough imaginary parts $T_2,U_2$, evaluates to
\be
S_{1-loop}=- 48 \log\Big|(e^{-2\pi i T}\prod_{\substack{n> 0\\m\in \ZZ}}(1-e^{2\pi i U\frac{m}{N}}e^{2\pi iTn})^{\hat c_{n,m}(\frac{mn}{N})}\Big|.
\label{oneloopresult}
\ee
The argument of the logarithm can be recognized as the absolute value of the infinite product formula \eqref{infiniteprod} for the twisted index $Z_{g, \id}$, provided that the vacuum contribution is $e^{-2\pi i T}$.
Thus we conclude
\be
e^{-S_{1-loop}}=|Z_{g, \id}(T,U)|^2\ ,
\ee as we expected based on physical arguments.

\subsection{T-dualities and automorphisms of lattices}\label{s:Tdualities}

In this section, we study the groups of T-dualities of the Monster CHL models with Euclidean time compactified on a circle. As we will see, the most general group of T-dualities relates a CHL model  at a given point in the moduli space to a (possibly different) CHL model at different values of the moduli. 

We denote by $w_1,m_1,w_2,m_2$ the winding and momenta along the space-like and the Euclidean time circle of $\TT^2$, respectively. The vectors of winding-momenta span a  four dimensional lattice (the Narain lattice) $L$.
With each vector in $L$ are associated the left- and right-moving momenta $(k_{L},k_{R})\in \RR^2\times\RR^2$, depending on the moduli $T,U$. The difference $k_L^2-k_R^2$, however, is a moduli-independent even integer and defines the quadratic form of signature $(2,2)$
\be\label{qform} (m_1,w_1, m_2, w_2)^2:=k_L^2-k_R^2=2m_1w_1+2 m_2  w_2\ ,\ee on the lattice. A necessary condition for T-duality to preserve the OPE, is that the action on the lattice $L$ is an automorphism, i.e. an invertible linear map that preserves the quadratic form. 

Thus, as for ordinary compactifications on $\TT^2$,  the full T-duality group of any CHL model is a discrete subgroup of $O(2,2,\RR)$ and it can be identified with the group of automorphisms of the lattice $L$ (of signature $(2,2)$) of winding-momenta along $\TT^2$. The index $Z_{g,\id}(T,U)$ is expected to be invariant under the subgroup of \emph{self-dualities} of the CHL model, i.e. the group of T-dualities that related two different points in the moduli space of the \emph{same} CHL model.

The group of T-dualities always contains the parity transformation $\PP$ along the space direction, acting as in \eqref{parity} on the index $Z(T,U)$, as well as T-duality $\bfT$ along the Euclidean time circle. The latter acts by
\be \tilde R\leftrightarrow \frac{1}{\tilde R}
\ee on the radius $\tilde R=\frac{\beta}{2\sqrt{2}}$, and more generally by
\be \bfT:U\leftrightarrow  -\frac{1}{T}\ .
\ee
 Thus, $\bfT$ acts as
\be\label{timeTduality} Z_{g,\id}(T,U)\mapsto \bfT(Z_{g,\id}(T,U))= Z_{g,\id}(-\frac{1}{U},-\frac{1}{T})\ ,
\ee on the index and is always a self-duality of any CHL model, so that $Z_{g,\id}(T,U)$ must be invariant (possibly up to a phase) under $\bfT$.

Every other T-duality in $O(2,2,\RR)$ can be obtained by composing $\PP$ and $\bfT$ with T-dualities in the connected component $SO^+(2,2,\RR)$ of $O(2,2,\RR)$ containing the identity.
There is an isomorphism
\be SO^+(2,2,\RR)\cong (SL(2,\RR)\times SL(2,\RR))/(-1,-1)\ .
\ee To make this isomorphism explicit, rewrite a vector in $\RR^{2,2}$ as a $2\times 2$ matrix
\be X=\begin{pmatrix}
w_2 & w_1\\ -m_1 & m_2
\end{pmatrix}
\ee
 so that its norm \eqref{qform} is simply the determinant
 \be \|X\|^2= 2m_1w_1+2 m_2  w_2 =2\det X\ .
 \ee
 Then, there is an obvious action of $SL(2,\RR)\times SL(2,\RR)$ on $X$ preserving its norm, namely 
 \be X\mapsto \gamma_1X\gamma_2\qquad \gamma_1,\gamma_2\in SL(2,\RR)\ ,
 \ee and its clear that the kernel of this action is $(-1,-1)$.
 We denote by $\tilde{SO}^+(L)\subset SL(2,\RR)\times SL(2,\RR)$ the preimage of $SO^+(L)\subset SO^+(2,2,\RR)$ under the quotient map $SL(2,\RR)\times SL(2,\RR)\to SO^+(2,2,\RR)$, so that
 \be SO^+(L)=\tilde{SO}^+(L)/(-1,-1)\ .
 \ee
A T-duality acting on the lattice $L$ by a general automorphism
\be\label{autom} \begin{pmatrix}
w_2 & w_1\\ -m_1 & m_2
\end{pmatrix}\mapsto \begin{pmatrix}
 a & b\\ c & d
\end{pmatrix} \begin{pmatrix}
w_2 & w_1\\ -m_1 & m_2
\end{pmatrix}\begin{pmatrix}
 a' & b'\\ c' & d'
\end{pmatrix}
\ee must also act on the moduli $T$ and $U$ by
\be\label{action} T\mapsto \frac{dT-\frac{c}{N\lambda}}{-N\lambda b T+a}\qquad U\mapsto \frac{d'U+b'N\lambda}{\frac{c'}{N\lambda} U+a'}\ .
\ee so that the left- and right-moving momenta are preserved
\be k_L^2=\frac{\left|\begin{pmatrix}
T & \frac{1}{N\lambda}
\end{pmatrix}\begin{pmatrix}
w_2 & w_1\\ -m_1 & m_2 
\end{pmatrix}\begin{pmatrix}
-U \\ N\lambda
\end{pmatrix}
\right|^2}{2T_2U_2}\ ,\qquad\qquad k_R^2=k_L^2-2m_1 w_1-2 m_2 w_2\ .
\ee 
In the unorbifolded case (i.e. $g=\id$), preserving the norms $k_L^2,k_R^2$ of all left- and right-moving momenta is both necessary and sufficient for the spectrum and the OPE of the theory to be preserved. For general $g$, this condition is not sufficient; the elements of $\tilde{SO}^+(L)$ preserving a given CHL model generate its group $G_g$ of self-dualities.

In the next subsections, we will study the groups of automorphisms $\tilde {SO}^+(L)$ of the lattices $L$  and then discuss the subgroups of self-dualities. We will first consider the simplest case $\lambda=1$ and then extend the analysis to generic $\lambda$. 

\subsubsection{Case $\lambda=1$}

In this section, we describe the group of T-dualities of a Monster CHL model for a symmetry $g$ of order $N$ with trivial multiplier ($\lambda=1$). The case $\lambda>1$ will be considered in the next section.

In the case of a CHL model with respect to a symmetry $g$ of order $N$, the winding and momenta span a lattice $L$ given by
\be\label{vector} (m_1,w_1,m_2, w_2)\in \ZZ\oplus \frac{1}{N}\ZZ\oplus \ZZ\oplus \ZZ\ .
\ee

The following subgroup of $SL(2,\ZZ)\subset SL(2,\RR)$ 
\be \Gamma_0(N):=\left\lbrace\begin{pmatrix}
a & b\\ c & d
\end{pmatrix}\in SL(2,\ZZ)\mid c\equiv 0\mod N\right\rbrace ,
\ee will be important in the following.
The normalizer $\hat \Gamma_0(N)$ of $\Gamma_0(N)$ in $SL(2,\RR)$ is described in \cite{Conway:1979kx}. It consists of the matrices of the form 
\be\label{normal} \frac{1}{\sqrt{e}}\begin{pmatrix}
ae & b/h\\ cN/h & de
\end{pmatrix}\ ,
\ee where $a,b,c,d\in \ZZ$, $h$ is the maximal integer such that $h|24$ and $h^2|N$, $e\in \ZZ_{>0}$ is an exact divisor of $N/h^2$(denoted by $e||\frac{N}{h^2}$), i.e. $e|\frac{N}{h^2}$ and $(e,\frac{N}{eh^2})=1$,\footnote{In this section, we make use of the standard notation $\text{gcd}(a, b)=:(a, b)$.} and
\be ade^2-bc \frac{N}{h^2}=e\ .
\ee Among the elements in $\hat \Gamma_0(N)/\Gamma_0(N)$, an important role is played by the Atkin-Lehner involutions
\be W_e=\frac{1}{\sqrt{e}}\begin{pmatrix}
ae & b\\ cN & de
\end{pmatrix}\ ,
\ee which obey $W_e^2\in \Gamma_0(N)$ and $W_{e_1}W_{e_2}=W_{e_3}$ modulo $\Gamma_0(N)$, where $e_3:=e_1e_2/(e_1,e_2)^2$.

\begin{theorem}\label{th:easy}
The group $\tilde{SO}^+(L)$ of automorphisms of the lattice $L$ is
\begin{multline}\label{automgroup} \tilde{SO}^+(L)=\{\bigl(\frac{1}{\sqrt{e}}\begin{pmatrix}
ae & b\\ cN & de
\end{pmatrix}, \frac{1}{\sqrt{e}}\begin{pmatrix}
a'e & b'\\ c'N & d'e
\end{pmatrix}\bigr)\in SL(2,\RR)\times SL(2,\RR)\\ a,b,c,d,a',b',c',d'\in \ZZ,\ e\in \ZZ_{>0},\ e||N\}\ .
\end{multline} The group is generated by adjoining to the normal subgroup $\Gamma_0(N)\times \Gamma_0(N)\subset \tilde{SO}^+(L)$ the Atkin-Lehner involutions $(W_e,W_e)$, for all $e||N$.
\end{theorem}

\begin{proof} See Appendix \ref{pf:easy}.
\end{proof}

For $N=1$, preserving the norms $k_L^2,k_R^2$ of all left- and right-moving momenta is both necessary and sufficient for the spectrum and the OPE of the theory to be preserved. However, this is not the case for $N>1$. The reason is that the orbifold construction forces the states with winding and momentum $w_1:=\frac n N, \ m_1$ along the spatial circle $S^1$ to be tensored with a state in $V^\natural_{n,m_1}$, the $g=e^{\frac{2\pi i m_1}{N}}$-eigenspace of the $g^{n}$-twisted sector of the internal CFT $V^{\natural}$.
Notice that the subgroup $\Gamma_0(N)\times \Gamma_0(N)$ of $\tilde{SO}^+(L)$ leaves $n$ and $m_1$ fixed modulo $N$. It follows that this subgroup is a genuine T-duality group, establishing an equivalence of the Monster CHL model at two different values of the moduli $T,U$.

The effect of the Atkin-Lehner involutions $(W_e,W_e)\in \tilde{SO}^+(L)$ has been analyzed in \cite{Persson:2015jka}. Let us first introduce some notation. For each holomorphic bosonic VOA $V$ of central charge $c=24$ and symmetry $g\in \Aut(V)$, let us denote by $(V,g)$ the CHL model based on the heterotic compactification on $\TT^2\times(V\times \bar V^{s\natural})$ followed by an orbifold by $(\delta,g)$, where $\delta$ is a shift of the same order as $g$. Then, the transformation $(W_e,W_e)$ establishes an equivalence between the CHL model $(V^\natural,g)$ with moduli $T,U$ and the CHL model $(V',g')$ with moduli $W_e\cdot T,W_e\cdot U$. Here, $V'=V^\natural/\langle g^{N/e}\rangle$ is the orbifold of $V^\natural$ by $g^{N/e}$ and $g'=Qg$ where $Q$ is the quantum symmetry acting by $e^{\frac{2\pi i r}{N/e}}$ on the $(g^{N/e})^r$-twisted sector and the action of $g$ on $V'$ is induced by its action on $V$. Schematically,
\be (V^\natural,g) \xleftrightarrow{(W_e,W_e)} (V'=V^\natural/\langle g^{N/e} \rangle,g'=gQ)\ .
\ee
This result follows immediately by observing that if the action of the Atkin-Lehner involution on the winding momenta is
\be (m_1,w_1, m_2, w_2) \xleftrightarrow{(W_e,W_e)} (m_1',w_1', m_2', w_2')\ ,
\ee then (essentially by definition of $V'$ and $g'$)
\be V'_{n',m_1'}=V^\natural_{n,m_1}\ ,
\ee
where $V'_{n',m_1'}$ is the $g'=e^{\frac{2\pi i m_1'}{N}}$ eigenspace of the ${g'}^{n'}$-twisted sector of the CFT $V'$ and $n'=Nw'$.

From this discussion, it is clear that, for a generic $(\gamma_1,\gamma_2)\in \tilde{SO}^+(L)$, there is a relation between the supersymmetric index relative to the CHL model $(V^\natural,g)$ and one for the CHL model $(V',g')$ by\footnote{Strictly speaking, our argument only implies an identity between the absolute values $|Z^{V^\natural}_{g,\id}|$ and $|Z^{V'}_{g',\id}|$. However, one can check that non-trivial phases only arise when $\lambda>1$.}
\be\label{Tduality} Z^{V^\natural}_{g,\id}(T,U)=Z^{V'}_{g',\id}(\gamma_1\cdot T,\gamma_2\cdot U)\ ,\qquad (\gamma_1,\gamma_2)\in \tilde{SO}^+(L)\ .
\ee It is useful to consider the subgroup $G_g\subseteq \tilde{SO}^+(L)$ of T-dualities such that the corresponding orbifold $V'$ is isomorphic to $V^\natural$ and $g'$ is in the same Monster conjugacy class as $g$. We call $G_g$ the group of self-dualities. For this subgroup, eq.\eqref{Tduality} implies
\be Z^{V^\natural}_{g,\id}(T,U)=Z^{V^\natural}_{g,\id}(\gamma_1\cdot T,\gamma_2\cdot U)\ ,\qquad (\gamma_1,\gamma_2)\in G_g\ .
\ee i.e. the index $Z^{V^\natural}_{g,\id}$ is invariant under the subgroup $G_g$. The group $G_g$ is generated by $\Gamma_0(N)\times\Gamma_0(N)$ as well as the Atkin-Lehner involutions $(W_e,W_e)$ such that $V^{\natural}/\langle g^{N/e}\rangle\cong V^\natural$ and $g'$ is conjugated with $g$.

\subsubsection{Case $\lambda>1$}
Let us consider the case where the level matching condition for the $g$-twisted sector in the Monstrous CFT is not satisfied, i.e. the conformal weights take values in
\be \frac{\E_g}{N\lambda}+\frac{1}{N}\ZZ\ ,
\ee where $\lambda|N$ and $(\E_g,\lambda)=1$. As shown in Appendix \ref{a:TwistAndTwin}, for the Monster CFT $V^\natural$, $\lambda$ is always a divisor of $24$. The CHL model for $\lambda>1$ is constructed by taking a shift of order $N\lambda$ along the space-like circle $S^1$ and then tensoring strings with winding $m$ and momentum $\frac{n}{N\lambda}$ along $S^1$ with states  in the spaces $V^\natural_{n,m}$, defined as in \eqref{Vnm}. One has $V^\natural_{n,m}=0$ unless $m-n\E_g\equiv 0\mod \lambda$. Therefore, the lattice $L$ of  winding-momenta along $\TT^2$ is spanned by
\be \begin{pmatrix}
m_1\\ w_1\\ m_2 \\ w_2
\end{pmatrix}=k\begin{pmatrix}
\lambda\\ 0\\ 0 \\ 0
\end{pmatrix}+n\begin{pmatrix}
\E_g\\ \frac{1}{N\lambda}\\ 0 \\ 0
\end{pmatrix}+\tilde k\begin{pmatrix}
0\\ 0\\ 1 \\ 0
\end{pmatrix}+\tilde n \begin{pmatrix}
0\\ 0\\ 0 \\ 1
\end{pmatrix}\qquad k,n,\tilde k,\tilde n\in \ZZ\ ,
\ee again with quadratic form \eqref{qform}.

It is useful to define the group $\Gamma_0(N|\lambda)\subset SL(2,\RR)$, whose elements are matrices
\be \begin{pmatrix}
a & b/\lambda \\ cN & d
\end{pmatrix}\ .
\ee The group $\Gamma_0(N|\lambda)$ is a subgroup of the normalizer of $\Gamma_0(N\lambda)$ in $SL(2,\RR)$. In particular, as discussed in \cite{Conway:1979kx}, $\Gamma_0(N|\lambda)$ is generated by $\Gamma_0(N\lambda)$, together with
\be \begin{pmatrix}
1 & 1/\lambda \\ 0 & 1
\end{pmatrix}\qquad \text{and}\qquad \begin{pmatrix}
1 & 0 \\ N & 1
\end{pmatrix}\ .
\ee
One can also define the Atkin-Lehner involutions for $\Gamma_0(N|\lambda)$
\be w_e:= \frac{1}{\sqrt{e}}\begin{pmatrix}
ae & b/\lambda\\ cN & de
\end{pmatrix}
\ee where $e||\frac{N}{\lambda}$, that are also in the normalizer of $\Gamma_0(N\lambda)$ in $SL(2,\RR)$. Two transformations
\be \frac{1}{\sqrt{e}}\begin{pmatrix}
ae & b/\lambda\\ cN & de
\end{pmatrix}\qquad \frac{1}{\sqrt{e'}}\begin{pmatrix}
a'e' & b'/\lambda\\ c'N & d'e'
\end{pmatrix}
\ee are in the same $\Gamma_0(N\lambda)$ left (or right) coset if and only if $e=e'$ and there is $\kappa\in \ZZ$, with $(\kappa,\lambda)=1$, such that
\be
 a'\equiv \kappa a,b'\equiv \kappa b,c'\equiv \kappa c,d'\equiv \kappa d\mod \lambda\ .\ee

\begin{theorem}\label{th:notsoeasy}
The group $\tilde{SO}^+(L)$ of automorphisms of the lattice $L$ consists of transformations
\be\label{automgroup2} \bigl(\frac{1}{\sqrt{e}}\begin{pmatrix}
ae & b/\lambda\\ cN & de
\end{pmatrix}, \frac{1}{\sqrt{e}}\begin{pmatrix}
a'e & c'\E_g/\lambda\\ b'\E_g N & d'e
\end{pmatrix}\bigr)\in SL(2,\RR)\times SL(2,\RR)\ee
where  $a,b,c,d,a',b',c',d'\in \ZZ,$  $e\in \ZZ_{>0}$ with $e||\frac{N}{\lambda}$ and
\be
 a'\equiv \kappa a,b'\equiv \kappa b,c'\equiv \kappa c,d'\equiv \kappa d\mod \lambda\ee  for some $\kappa$ coprime to $\lambda$.
The group is generated by adjoining to the normal subgroup $\Gamma_0(N\lambda)\times \Gamma_0(N\lambda)\subset \tilde{SO}^+(L)$ the transformations
\be\label{lambtransf} \left(\begin{pmatrix}
1 & 1/\lambda\\ 0 & 1
\end{pmatrix}, \begin{pmatrix}
1 & 0\\ \E_g N & 1
\end{pmatrix}\right)\ ,\qquad \left( \begin{pmatrix}
1 & 0\\  N & 1
\end{pmatrix},\begin{pmatrix}
1 & \E_g/\lambda\\ 0 & 1
\end{pmatrix}\right)\ , \ee as well as the Atkin-Lehner involutions $(w_e,w_e)$ for all $e||\frac{N}{\lambda}$.
\end{theorem}

\begin{proof} See Appendix \ref{pf:notsoeasy}
\end{proof}

As in the $\lambda=1$ case, the generic T-duality in $\tilde{SO}^+(L)$ is an equivalence between the CHL model $(V^\natural,g)$ and a (possibly) different model $(V',g')$. The index $Z^{V^\natural}_{g,\id}$ is expected to be invariant only under the subgroup $G_g\subset \tilde{SO}^+(L)$ of self-dualities, i.e. where $V'\cong V^\natural$ and $g$ is conjugated with $g'$ within $\Aut(V^\natural)=\MM$. The subgroup $\Gamma_0(N\lambda)\times\Gamma_0(N\lambda)$ fixes the winding and momenta modulo $N\lambda$, so that it must be a (normal) subgroup of $G_g$. 

The effect of the symmetries \eqref{lambtransf} is rather subtle. A direct calculation shows that the action of such dualities on momenta and winding along the spatial circle has the form
\begin{align}
w_1\mapsto& w_1'\equiv w_1+\frac{1}{\lambda}X(m_1,w_1, m_2, w_2)\mod \ZZ\ ,\\
m_1\mapsto& m_1'\equiv m_1 - \E_gNX(m_1,w_1,m_2, w_2)\mod N\lambda\ZZ\ ,
\end{align} where $X:L\to \ZZ$ is an integral-valued linear functional that depends on the particular duality. 
As discussed in Appendix \ref{a:TwistAndTwin}, for all $X\in\ZZ$ there is an isomorphism
\be\label{shift} V^\natural_{n,m}\cong V^\natural_{n+NX,m-\E_g NX}\ ,
\ee  where $V^\natural_{n,m}$ is the $g=e^{\frac{2\pi i m}{N\lambda}}$-eigenspace in the $g^n$-twisted sector, $n,m\in \ZZ/N\lambda\ZZ$. In the CHL orbifold, states with winding $w=\frac{n}{N\lambda}$ and momentum $m$ are tensored with states in $V^\natural_{n,m}$. Thus, the dualities \eqref{lambtransf}, together with \eqref{shift}, are symmetries of the Monster CHL model to itself and the index $|Z_{g,\id}(T,U)|^2$ must be invariant under the corresponding action.

Finally, the Atkin-Lehner involutions $(w_e,w_e)$, where $e||\frac{N}{\lambda}$ establish an equivalence with the model $(V',g')$, where $V'=V^{\natural}/\langle g^{N/e}\rangle$ (notice that $N/e$ is a multiple of $\lambda$, so that the level matching condition is satisfied and the orbifold is consistent) and $g'=g^xQ^y$ is a certain combination of order $N$ of the quantum symmetry $Q$ and $g$ (see \cite{Persson:2015jka} for details). 

To conclude, the group of T-dualities $\tilde{SO}^+(L)$ contains the normal subgroup $G_g$ of self-dualities, generated by $\Gamma_0(N\lambda)\times \Gamma_0(N\lambda)$, the transformations \eqref{lambtransf} and those Atkin-Lehner involutions $(w_e,w_e)$ such that $(V',g')\sim (V^\natural,g)$ up to the non-trivial automorphism of the Monster CHL model. The index $Z^{V^\natural}_{g,\id}(T,U)$ is invariant up to a phase with respect to this group
\be |Z^{V^\natural}_{g,\id}(T,U)|^2=|Z^{V^\natural}_{g,\id}(\gamma_1\cdot T,\gamma_2\cdot U)|^2\qquad (\gamma_1,\gamma_2)\in G_g\ .\ee

\section{The space-time index III: algebras and denominators}\label{s:Algebras}

\noindent In section \ref{s:indexI} we noticed that the (24th root of the) index $Z(T,U)$ \footnote{Occasionally, we will loosely refer to $Z(T, U)$ as the algebra denominator; the reader should be aware that the precise denominator is really its 24th root $Z(T, U)^{1/24}$.} is exactly the denominator of Borcherds' Monstrous Lie algebra. In this section, we will show that this is no coincidence: the Monstrous Lie algebra appears as a spontaneously broken gauge symmetry in the string theory we are considering and the BPS states form a representation for this algebra. An analogous relation exists between each Monstrous CHL model and an infinite dimensional Borcherds-Kac-Moody (BKM) algebra. Using this fact, we will prove that the only states contributing to the index $Z$ are the ones annihilating the winding-momentum $MW$, which is the quadratic Casimir of the algebra. 

\subsection{Monstrous Lie algebras}

The single string BPS states are related to physical string states in the Ramond sector with $k_R^2=0$ and $h_R=\frac{1}{2}$ whose  vertex operator (unintegrated, in the $(-1/2)$--picture) can be written as
\be (c\tilde c \V_\chi e^{-\tilde\phi/2} \tilde S_\alpha^i e^{ikX})(z,\bar z)\ ,
\ee where $c,\tilde c$ are the ghosts, $e^{-\tilde \phi/2}$ is the superghost, $\V_\chi$ is a holomorphic vertex operator corresponding to the state $\chi\in V^\natural$ of conformal weight $h_\chi$, $\tilde S_\alpha^i$, $i=1,\ldots,24$ is a tensor product of one of the $24$ Ramond ground states of the internal SCFT $\bar V^{s\natural}$ (labeled by $i$) and a Ramond ground state of positive chirality in the space-time directions. Furthermore, the space-time momentum $k$ must satisfy the level-matching conditions \eqref{physical}. We denote by
\be |\chi,i,\alpha,k\rangle \ ,
\ee the corresponding state.

These states have no supersymmetric partners: the reason is that the corresponding states in the NS sector are either zero or BRST exact. Explicitly, space-time supersymmetry acts on the BPS states by (we omit the holomorphic part $c \V_\chi  (z)$ of the vertex operator since it plays no role in this computation)
\begin{align} &\oint dz (e^{-\tilde\phi/2} \tilde S^i_\alpha)(\bar z)\, (\tilde c  e^{-\tilde\phi/2} \tilde S^j_\beta e^{ikX})(0)=\delta^{ij}(\Gamma^0\Gamma^\mu)_{\alpha\beta}(\tilde c   \tilde\psi_\mu e^{ikX})(0)\\
&=\delta^{ij}(\tilde c  (\tilde\psi_0-\tilde\psi_1)  e^{ikX})(0)\propto k_{R,\mu}(\tilde c   \tilde\psi^\mu e^{ikX}) (0)\ ,
\end{align} where $\tilde \psi^\mu$ is the weight $1/2$ field related to the current $\bar\partial X^\mu$ by world-sheet supersymmetry. We used the fact that the supersymmetry charge and the BPS state are fermions of the same chirality and that $k_R^0=k_R^1$ by the BPS condition.
Restoring the dependence on the holomorphic part, we conclude that the supersymmetric partners of the BPS states
can be written as
\be\label{null} k_{R,\mu}\tilde \psi^\mu_{-1/2} |\chi,k\rangle\ .
\ee
 and correspond to the (unintegrated, in the $(-1)$--picture) vertex operator
\be\label{nullvert}  (c\tilde c \V_\chi e^{-\tilde\phi} k_{R,\mu}\tilde \psi^\mu e^{ikX}) (z,\bar z)\ .
\ee
  The space-time momentum $k$ satisfies the same conditions \eqref{physical} as above. This state 
  is BRST exact 
\be k_{R,\mu}\tilde \psi^\mu_{-1/2} |\chi,k\rangle=\tilde G_{-1/2}|\chi,k\rangle \ee
where $\tilde G_{-1/2}$ is the world-sheet supersymmetry operator. 
More precisely, the vertex operator \eqref{nullvert} is the BRST variation
\be  c\tilde c \V_\chi e^{-\tilde\phi} k_{R,\mu}\tilde \psi^\mu e^{ikX}=\{\Q_R^{BRST},  c \V_\chi e^{-\tilde\phi} e^{ikX}\}\ ,\qquad \{\Q_L^{BRST},  (c \V_\chi e^{-\tilde\phi} e^{ikX})\}=0 \ ,
\ee where $\Q_L^{BRST}$ and $\Q_R^{BRST}$ are the left- and right-moving components of the BRST operator.  In general, massless ($k^2=0$) BRST exact (null) states in string theory correspond to gauge invariances. The zero momentum limit of such states give rise to global gauge invariance and charge conservation laws. In the present context of compactification to $0+1$ dimensions, the very definition of the mass as $k^2$ is somehow problematic. Notice, however, that the BRST exact states we are considering always carry non-zero momentum and,   in this sense, they are similar to \emph{massive} null states. In analogy with massive null states, the BRST exact states \eqref{nullvert} are expected to generate a spontaneously broken gauge symmetry, which is restored in the limit of tensionless string $\alpha'\to \infty$, where all these states become massless. 

Even if spontaneously broken, we expect the gauge symmetry generators to form an algebra. To see this explicitly,  it is useful to write the vertex operator \eqref{nullvert} in its integrated form and $0$-picture
\be 
\W_\chi:=\int d^2 z\; (\V_\chi k_{R,\mu}(\bar \partial X^\mu)  e^{ikX}) (z,\bar z)=\int d^2 z\;\bar \partial (\V_\chi  e^{ikX}) (z,\bar z)\ .
\ee (For a vector with generic polarization $\epsilon_\mu$, the $0$-picture vertex operator  also contains a term proportional to $\epsilon_{\mu}k_{R,\nu}\tilde \psi^\nu\tilde \psi^\mu$; this term vanishes when $\epsilon_\mu=k_{R,\mu}$ for symmetry reasons). 

By formally inserting this null vertex operator inside some string amplitude, we obtain 
\be 0=\langle \W_{\chi} \V_1\V_2\ldots \rangle = \lim_{\epsilon\searrow 0}\sum_i \langle \V_1\V_2 \ldots \oint_{\gamma_{i,\epsilon}} dz (\V_\chi  e^{ikX})(z,\bar z) \V_i\ldots \rangle\ ,
\ee where $\gamma_{i,\epsilon}$ is a small circle of radius $\epsilon$ around the insertion point of the vertex operator $\V_i$.
Let us consider the case where $\V_i$ is the vertex operator of a BPS state $|\hat\chi,\alpha,\hat k\rangle$ inserted at $z=0$, for some momentum $\hat k$ and state $\hat\chi\in V^\natural$. Then,
\begin{multline}\label{Monstaction} \lim_{\epsilon\searrow 0}\oint_{\gamma_{0,\epsilon}} dz (\V_\chi  e^{ikX})(z,\bar z) |\hat\chi,\alpha,\hat k\rangle\\=\lim_{\epsilon\searrow 0}\oint_{\gamma_{0,\epsilon}} dz (\V_\chi  e^{ik_LX_L})(z) \, \bar z^{k_R\cdot \hat k_R} \exp\Bigl(\sum_{n>0}k_{R,\mu}\frac{\bar\alpha^\mu_{-n}}{n} \bar z^{n}\Bigr)|\hat\chi,\alpha,\hat k\rangle\ ,
\end{multline} where we have written the dependence on $\bar z$ explicitly. Now, since both $k_R^0=k_R^1$ and $\hat k_R^0=\hat k_R^1$, we conclude that $k_R$ and $\hat k_R$ are proportional to each other. Therefore, 
\be k_R\cdot\hat k_R\propto  k_R^2=0\ .
\ee Let us consider the mode expansion
\be (\V_\chi  e^{ik_LX_L})(z)=\sum_{n\in \ZZ} (\V_\chi  e^{ik_LX_L})_{-n} z^{n-1}\qquad \exp\Bigl(\sum_{n>0}k_{R,\mu}\frac{\bar\alpha^\mu_{-n}}{n} \bar z^{n}\Bigr)=\sum_{m\ge 0} (e^{ik_RX_r})_{-m} \bar z^{m}\ ,
\ee where we notice that the (anti-)holomorphic fields $(\V_\chi  e^{ik_LX_L})(z)$ and $\exp\Bigl(\sum_{n>0}k_{R,\mu}\frac{\bar\alpha^\mu_{-n}}{n} \bar z^{n}\Bigr)$ have total conformal weights\footnote{In the remainder of the paper, for ease of notation, we will simply denote $m_1, w_1$ by $m, w$ unless otherwise noted.} $-mw+h_\chi=1$ and $k_R^2=0$, respectively. By replacing this expansion in \eqref{Monstaction}, we obtain
 \begin{align}
 &\lim_{\epsilon\searrow 0}\int_0^{2\pi} dt \sum_{\substack{m\ge 0\\n\in\ZZ}} \epsilon^{m+n} e^{it(n-m)} (\V_\chi  e^{ik_LX_L})_{-n}(e^{ik_RX_r})_{-m}|\hat\chi,\alpha,\hat k\rangle\\
 &=\lim_{\epsilon\searrow 0}\sum_{m\ge 0} \epsilon^{2m} (\V_\chi  e^{ik_LX_L})_{-m}(e^{ik_RX_r})_{-m}|\hat\chi,\alpha,\hat k\rangle
 \\
 &=(\V_\chi  e^{ik_LX_L})_0|\hat\chi,\alpha,\hat k\rangle\ .
\end{align}  
Thus, the action is given by the zero-modes of the holomorphic currents $(\V_\chi  e^{ik_LX_L})(z)$.  These zero-modes form an algebra, with commutation relations given by the usual contour argument, and the space of single particle BPS states is a module over this algebra. It is easy to see that this is exactly the Monster Lie algebra constructed by Borcherds \cite{Borcherds}. Indeed, the starting point of Borcherds' definition is the vertex operator algebra given by the product of $V^{\natural}$ and the lattice VOA based on the unimodular lattice $\Gamma^{1,1}$ of signature $(1,1)$. Borcherds then essentially takes the cohomology with respect to a suitable `BRST' operator, whose class representatives have total conformal weight $1$. Borcherds' BRST operator corresponds to the \emph{left-moving} component $\Q_L^{BRST}$ of the full BRST operator in the string model. In more detail, the level-matching and the right-moving mass-shell conditions force the left-moving momentum $k_L$ to take values in a lattice isomorphic to $\Gamma^{1,1}$ and the left-moving mass-shell condition is exactly Borcherds' physical state condition.

To summarize, the algebra of gauge symmetries corresponding to the null states \eqref{null} is Borcherds' Monster Lie algebra and the space of BPS (first quantized) string states is a module over this algebra. More precisely, for each generator of the Monster Lie algebra, there are $24$ physical BPS single string states, that form $24$ copies of the adjoint representation. In section \ref{s:indexI}, these BPS string states were the starting point for the construction of the second quantized BPS Fock space $\Hh_{BPS}$ and of the index $Z(T,U)$.  As we will explain in the next subsections, $Z(T,U)^{1/24}$ is a denominator for the Monster Lie algebra and one can apply the Generalized Kac-Moody generalization of the standard Weyl-Kac denominator formula to evaluate it.

This construction admits a direct generalization to the other Monster CHL models. In all cases, the superpartners of the BPS states are BRST variations of some fields of conformal weights $(1,0)$ and one obtains an infinite dimensional Lie algebra generated by the zero modes of the corresponding holomorphic currents. These are exactly the algebras considered in section 4 of \cite{Carnahan2014} as part of the proof of the Generalized Monstrous Moonshine conjecture \cite{Norton1987}. In all cases, the first quantized BPS string states form $24$ copies of the adjoint representation of the corresponding CHL algebra.

\subsection{The index as an algebra denominator}\label{ss:algebra}

In the last section, we showed that each CHL model contains an infinite dimensional Lie algebra $\g$ and that each generator of this algebra is associated with $24$ fermionic BPS states. These fermionic BPS states are the ones involved in the construction of the BPS Fock space in section \ref{s:indexI} and in the definition of the index $Z_{g, \id}(T,U)$. In this section, we will show that each index $Z_{g, \id}(T,U)$ is essentially (the 24th root of) the Weyl denominator of the corresponding Lie algebra $\g$; this interpretation will allow us to derive explicit formulae for the index in terms of the McKay-Thompson series. This section may be viewed as a physical reinterpretation of the Lie algebra homology computations of \cite{GarlandLepowsky,Jurisich1996,Jurisich1998,Jurisich2004}.

\bigskip

Let us first fix some notation. Let $\g$ be a (semi-simple/(generalized) Kac-Moody) real Lie algebra with a decomposition
\be \g=\g^-\oplus \h\oplus \g^+\ ,
\ee with respect to a Cartan subalgebra $\h$ and subalgebras $\g^\pm$ corresponding to positive/negative roots. The algebra has a grading $\g=\bigoplus_{\gamma\in \Gamma} \g^\gamma$ with weight space $\Gamma\subset \h^*$ (or possibly in some extended $(\h^e)^*$) and such that each $\g^\gamma$ is finite dimensional. We denote by
$f^a_{\phantom{a}bc}$ the (real) structure constants of $\g$
\be [b,c]=\sum_{a\in \g}f^a_{\phantom{a}bc}a\qquad b,c\in \g\ .
\ee

\medskip

In the case where $\g$ is the algebra of a Monster CHL model, the Cartan subalgebra is generated by the zero-modes of the left-moving currents $\partial X^\mu(z)$, $\mu=0,1$, so that the weight space $\h^*$ is two dimensional. The grading is given by the set of winding-momenta $(n,m)$ or, equivalently, by the left-moving momenta $k^\mu_L$, $\mu=0,1$, which are linear combinations of $n$ and $m$, depending on the radius $R$. Remember, we have defined $n \in \ZZ$ implicitly via $w:= {n \over N}$ (taking $\lambda=1$ for now). The zero component $k^0_L$ (the energy) will be often denoted by $E$. The positive energy condition $E>0$ then selects the subalgebra $\g^-$ of negative  roots. The graded components $\g^\gamma$, with $\gamma=(n,m)$ have dimension $\hat c_{n,m}(\frac{mn}{N})$, where $\hat c_{n,m}(k)$ is the coefficient of $q^k$ in the Fourier expansion of $F_{n,m}(\tau)$ \eqref{defF} or, equivalently, the dimension of the space $V^\natural_{n,m}$ at level $L_0-1=k$.

\medskip

The algebras we are considering can be endowed with a non-degenerate symmetric $\g$-invariant bilinear form $\kappa$ such that, for any $a\in \g^\gamma$, $\kappa(a,b)=0$ unless $b\in \g^{-\gamma}$. Physically, $\kappa(a,b)$ is the coefficient of the two point function of the holomorphic vertex operators associated with $a,b\in \g$. This is non-zero only when the two vertex operators have opposite (left-moving) momentum. On the Cartan subalgebra, $\kappa$ is the usual space-time metric with Lorentzian signature. In particular, for any $k^\mu_L,\hat k^\nu_L\in \h^*$, one has $\kappa(k^\mu_L,\hat k^\nu_L)=\eta_{\mu\nu}k^\mu_L\hat k^{\nu}_L$. Finally, we require a (Cartan) involution $\theta$ acting by multiplication by $-1$ on $\h$ and mapping $\theta \g^{\pm}\to \g^{\mp}$ and such that $-\kappa(a,\theta(a))>0$ for all nonzero $a\in \g^\gamma$, $\gamma\neq (0,0)$.  This involution simply flips the sign of the momentum in the vertex operator $(V_\chi e^{ik_LX_L})\to (V_\chi e^{-ik_LX_L})$.\footnote{We are considering real vertex operators $V_\chi$ of the Monster VOA. In the complexified case, the involution $\theta$ is anti-linear and $-\kappa(a,\theta(b))$ is hermitian.}

\medskip

For each element $a\in \g$ of the algebra, there are $24$ Ramond string states satisfying the BPS condition $(k_R)^2=0$. For simplicity, we will build a BPS space $\Hh$ just from one of these $24$ sets of BPS states. The full BPS space considered in section \ref{s:indexI}, therefore, will be the tensor product of $24$ identical copies of the space $\Hh$ and the corresponding index $Z(T,U)$ the $24$th power of the index of a single copy. 

\medskip

Thus, with each $a\in \g$, we associate a fermionic mode $\eta_a$, corresponding to a physical Ramond string state in the $(-1/2)$-picture. The two point function of any pair of such $(-1/2)$-picture states vanishes, suggesting the anticommutation relations
\be \{\eta_a,\eta_b\}=0\qquad a,b\in \g.
\ee 
Let us consider a ground state $|0\rangle$ such that
\be \eta_a|0\rangle=0\qquad a\in \g^+\oplus \h\ ,
\ee and let $\Hh$ be the Fock space constructed by acting on $|0\rangle$ by $\eta_a$, $a\in \g^-$. In fact, the operators $\eta_a$, $a\in \g^+\oplus \h$ annihilate every state in $\Hh$, so we can simply set them to zero and consider only $\eta_a$, $a\in \g^-$.
 The Fock space $\Hh$ is isomorphic to $\bigwedge \g^-$ and inherits the grading from $\g$, so that $\Hh=\bigoplus_{\gamma\in\Gamma} \Hh_{\gamma}$.\footnote{If $\g$ is an ordinary Kac-Moody algebra, $\bigwedge \g^-$ is an irreducible highest weight $\g$-module, with highest weight the Weyl vector. This can be easily proved using the Weyl-Kac denominator formula. To the best of our knowledge, the analogous statement for Generalized Kac Moody algebras has not  been established. In particular, when $\g$ is the Monster Lie algebra, $\bigwedge \g^-$ might be a reducible $\g$-module. We thank Richard Borcherds and Scott Carnahan for correspondence about this point.} 
 The denominator of the algebra is simply an index in this space
\be\label{WeylDen} Z(z):=\sum_{\gamma\in \Gamma}\\Tr_{\Hh_\gamma}(e^{\kappa(z,\gamma+\rho)}(-1)^F)\ ,\qquad z\in \h^*\ ,
\ee where $\rho\in \h^*$ is the Weyl vector with defining property
\be\label{Weyldef} \kappa(\rho,\gamma)=-\frac{1}{2}\kappa(\gamma,\gamma)\ ,
\ee for all \emph{simple} roots $\gamma$. As we will see, the Weyl vector exists in all the algebras relevant for our construction.

The space of Ramond string states in the $(-1/2)$-picture is dual to the space of Ramond states in the $(-3/2)$-picture. Associated with the latter states, we introduce fermionic operators
$\frac{\partial}{\partial \eta_a}$, $a\in \g^-$ satisfying
\be \frac{\partial}{\partial \eta_a}|0\rangle=0\qquad \left\lbrace\frac{\partial}{\partial \eta_a},\frac{\partial}{\partial \eta_b}\right\rbrace=0\qquad  \left\lbrace\frac{\partial}{\partial \eta_a},\eta_b\right\rbrace=\delta^a_b\quad 
a,b\in \g^-\ .\ee

Finally, we can endow $\Hh$ with the structure of a Hilbert space by defining the adjoint operator as
\be \eta_a^\dag:= -\kappa_{a \theta(b)}\frac{\partial}{\partial \eta_b}\ .
\ee
where we use the shorthand $\kappa_{a b}:= \kappa(a, b)$.
\medskip

The definition of the energy depends on the compactification radius $R$, it can happen that, as we vary $R$, the energy of some of the algebra generators change sign. From the point of view of the algebra, this corresponds to a Weyl transformation, leading to a different choice of the subalgebra $\mathfrak{g}^-$. Apparently, the Hilbert space $\Hh$ depends on this choice. However, notice that two Hilbert spaces $\Hh$ and $\Hh'$ related by a Weyl transformation are actually isomorphic. If $a_1,\ldots, a_r$ are the algebra elements changing sign under a Weyl transformation, then this isomorphism maps the new ground state $|0'\rangle\in\Hh'$ to $\eta_{a_1}\ldots\eta_{a_r}|0\rangle\in \Hh$ and exchanges \be\label{WeylExc} \eta_{a_i}\leftrightarrow \kappa_{a_ib} \frac{\partial}{\partial \eta'_{b}}\qquad \frac{\partial}{\partial \eta_{a_i}}\leftrightarrow  \kappa^{a_ib}\eta'_{b}\ , \qquad i=1,\ldots,r\ .\ee

Let us define the left-moving momenta operators $P_\mu$, $\mu=0,1$, as
\be\label{WeylInv} P_{\mu}:=\Bigl(\frac{1}{2}\sum_{a,b\in \g^-} f^a_{\phantom{a}\mu b}(\eta_a\frac{\partial}{\partial \eta_b}-\frac{\partial}{\partial \eta_b}\eta_a) \Bigr)_{reg}=\sum_{a,b\in \g^-} f^a_{\phantom{a}\mu b}\eta_a\frac{\partial}{\partial \eta_b}-\frac{1}{2}\Bigl(\sum_{a\in \g^-} f^a_{\phantom{a}\mu a}  \Bigr)_{reg}\ ,
\ee and the winding and momentum number operators 
\be M=\frac{R}{\sqrt{2}} (P^0+P^1)\qquad W=\frac{1}{\sqrt{2}R} (P^0-P^1)\ .
\ee Here, $f^a_{\phantom{a}\mu b}$ are the structure constants for the Cartan generator $\partial X_\mu$ and the subscript $reg$ denotes a suitable regularization procedure. It is easy to check that the definition \eqref{WeylInv} is independent of the choice of the subalgebra $\g^-$, provided one identifies the corresponding operators as in \eqref{WeylExc}. In particular, $M$ and $W$ are the appropriate winding and momenta in order to define an index $Z(T,U)$ that is continuous in the moduli, as expected from the 1-loop integral definition of the previous section.

The vector 
\be\label{noconstant} \rho_\mu:=-\frac{1}{2}\Bigl(\sum_{a\in \g^-} f^a_{\phantom{a}\mu a}  \Bigr)_{reg}\ ,
\ee is a normal ordering constant that determines the winding and momentum $(n_0,m_0)$ of the vacuum state $|0\rangle$. If $\g$ was a finite dimensional Lie algebra, $\rho$ would correspond to the Weyl vector. We will argue that, in order to match with the 1-loop expression, this must be true also in the case of the CHL Lie algebras. Notice that each graded component $\Hh_\gamma$, $\gamma\in \Gamma$ is an eigenspace for $P^\mu$ with a \emph{shifted} eigenvalue
\be\label{PmuonH} P^\mu|\phi\rangle=(\gamma^\mu+\rho^\mu)|\phi\rangle\ ,\qquad \phi\in \Hh_\gamma\ .
\ee It is now clear that the index $Z(T,U)^{1/24}$ exactly corresponds to the Weyl-Kac-Borcherds denominator \eqref{WeylDen}, with $\gamma+\rho=(n+ n_0,m + m_0)$ and $z=(T,U)$.

\bigskip

 The final ingredients needed in order to explicitly compute the index $Z(T,U)$ are the `boundary operator' 
\be d:=\frac{1}{2}\sum_{a,b,c\in \g^-} f^a_{\phantom{a}bc}\eta_a\frac{\partial}{\partial \eta_b}\frac{\partial}{\partial \eta_c}\ ,
\ee and its adjoint
\be d^\dag:=\frac{1}{2}\sum_{a,b,c\in \g^+} f^a_{\phantom{a}bc} \eta^b\eta^c \frac{\partial}{\partial \eta^a}\ .
\ee Here,  the indices are lowered and raised using $\kappa_{ab}$ and its inverse $\kappa^{ab}$
\be \eta^a:=\kappa^{ab}\eta_b\ ,\qquad \qquad \frac{\partial}{\partial \eta^a}:=\kappa_{ab}\frac{\partial}{\partial \eta_b}\ ,
\ee and we used the algebra automorphism
\be f^a_{\phantom{a}bc}=f^{\theta(a)}_{\phantom{\theta(a)}\theta(b)\theta(c)}\ .
\ee If we identify $\Hh$ with the $\bigwedge \g^-:=\oplus_i\bigwedge^i\g^- $, then $d\equiv \oplus_i d_i$ represents the standard boundary operator $d_i:\bigwedge^i\g^-\to\bigwedge^{i-1}\g^-$ acting by
\be d_i(a_1\wedge\ldots\wedge a_i)=\sum_{1\le r<s\le i} (-1)^{r+s} [a_r,a_s]\wedge a_1\wedge\ldots \wedge\check{a}_r\wedge \ldots\wedge  \check{a}_s\wedge \ldots\wedge a_i\ .
\ee The  operators $d$ and $d^\dag$ are nilpotent 
\be d^2=0=(d^\dag)^2\ ,
\ee
as follows by direct calculation and by the Jacobi identity
\be \sum_{c\in \g^{\pm}} (f^a_{\phantom{a}bc}f^c_{\phantom{c}de}+f^a_{\phantom{a}ec}f^c_{\phantom{c}bd}
+f^a_{\phantom{a}dc}f^c_{\phantom{c}eb})=0\ ,\qquad a,b,d,e\in \g^{\pm}
\ee
for the algebras $\g^-$ and $\g^+$. Furthermore, they preserve the grading of $\Hh$
\be d(\Hh_\gamma)\subseteq \Hh_\gamma\qquad d^\dag (\Hh_\gamma)\subseteq \Hh_\gamma\ ,
\ee since $f^a_{\phantom{a}bc}\neq 0$ implies $\gamma_a=\gamma_b+\gamma_c$ and each $\eta_a$ (respectively, $\frac{\partial}{\partial \eta_a}$) raise (resp., lower) the weight by $\gamma_a$.
The homology groups $H_i(\g)=\ker d_i/{\rm Im}d_{i+1}$ on the nilpotent operator  $d$ represent the standard homology of the Lie algebra $\g$. 

We can decompose $\Hh$ into irreducible representations of the algebra generated by $d,d^\dag$. It is clear that only the $1$-dimensional representations contribute to the index $Z$, since the higher dimensional representations contain the same number of fermions and bosons, all with the same weight. States belonging to a $1$-dimensional representation are annihilated by both $d$ and $d^\dag$, so that they are contained in the kernel of $\{d,d^\dag\}$ (in fact, $\ker\{d,d^\dag\}=\ker d\cap \ker d^\dag$, using the standard argument that
\be \langle\psi| \{d,d^\dag\}|\psi\rangle=\|d|\psi\rangle\|^2+\|d^\dag|\psi\rangle\|^2\ge 0\ ,
\ee and the right-hand side vanishes only when $d|\psi\rangle=0=d^\dag|\psi\rangle$.
). 

It is a standard theorem (see for example \cite{GarlandLepowsky}) that the anticommutator $\{d,d^\dag\}$ is given by the quadratic Casimir of the algebra. When the algebra admits a Weyl vector $\rho\in \h^*$ satisfying \eqref{Weyldef}, the Casimir has the form
\be\label{Casimiro} \{d,d^\dag\}_{\rvert \Hh_{\gamma}}=-\frac{1}{2}\bigl(\kappa(\gamma+\rho,\gamma + \rho)-\kappa(\rho,\rho)\bigr)\ .\ee As we will see, all the algebras we are interested in have a Weyl vector $\rho\in \h^*$ which coincides with the normal ordering constant \eqref{noconstant} and is null $\kappa(\rho,\rho)=0$. Therefore, using \eqref{PmuonH}, we conclude that the only states contributing to the index $Z_{g,\id}(T,U)$ will be the ones annihilated by $P_\mu P^\mu=2MW$.

\subsection{Explicit formulas for the indices}\label{s:explformulas}

In this section, we will use the information of section \ref{ss:algebra}  to obtain a third alternative formula for the indices $Z_{g, \id}(T,U)$.

 \subsubsection{The index for the Monster Lie algebra}
 
For the Monster Lie algebra, the infinite product formula \eqref{infiniteprod} obeys the famous identity (independently proved in the 1980's by Koike, Norton, and Zagier)
\be Z(T,U)^{1/24}= J(T)-J(U)\ ,
\ee where the vacuum contribution is given by $e^{-2\pi iT}$, as required by the 1-loop formula, corresponding to a normal ordering constant $(n_0,m_0)=(-1,0)$. Borcherds \cite{Borcherds} used this identity to deduce that the simple roots of the algebra are the ones of the form $(1,m)$, $m\in \ZZ$, with multiplicity $c(m)$. For any algebra, the one particle states corresponding to simple roots are always in the kernel of $\{d,d^\dag\}$. In this case, their contribution the index is 
\be -\sum_{m\in \ZZ} c(m) e^{2\pi i Um}= -J(U) \ .\ee It is clear then that the vector $\rho=(-1,0)$ is indeed the Weyl vector, since it satisfies the defining property \eqref{Weyldef} for all simple roots. As anticipated, the Weyl vector coincides with the vacuum winding-momentum $(n_0,m_0)$ and is null. Therefore, all non-vanishing contributions to $Z(T,U)$ come from states satisfying
 \be 0=\kappa(\gamma,\gamma+2\rho)=2 M W = 2 m_{tot}(n_{tot}-1)\ ,
 \ee where we have defined $m_{tot}, n_{tot}$ as the sum of momenta and windings, respectively, of multiparticle states. In particular, the term $J(T)-e^{-2\pi i T}$ comes entirely from two particle states.
 
Using  this description of the algebra, it is easy to compute a similar formula for the twined indices \cite{Borcherds}
\be\label{Zidg} Z_{\id,g}(T,U)^{1/24}=T_{\id,g}(T)-T_{\id,g}(U)\ ,
\ee for any $g\in \MM$. 

Of course, it would be ideal to derive the Koike-Norton-Zagier formula directly in our heterotic  construction from the outset. The main impediment to this is an independent proof of the existence of a Weyl vector for the Monster Lie algebra, i.e. to derive the vacuum winding and momentum. We leave this point to future work.
 
\subsubsection{CHL models}

Let us consider now the CHL case, for $g$ of order $N$ (in this section, we will restrict to $\lambda=1$, for simplicity). The root lattice $\Gamma$ is given by
 \be \gamma= (k^0_L,k^1_L)=\frac{1}{\sqrt{2}}\left(\frac{m}{R}+\frac{nR}{N},\frac{m}{R}-\frac{nR}{N}\right)\ ,\qquad m,n\in \ZZ\ ,
 \ee so that
 \be \kappa(\gamma,\gamma')=\frac{mn'+m'n}{N}\ ,
 \ee and the root multiplicity is
 \be {\rm mult}(n,m)= \hat c_{n,m}\left(\frac{nm}{N}\right)\ ,
 \ee where $\hat c_{n,m}(l)$ is the Fourier coefficient of $F_{n,m}$, the (discrete Fourier transform of the) generalized McKay-Thompson series (see eq.\eqref{defF}).
Physically, they represent the dimensions of $V^\natural_{n,m}$ (the $g=e^{\frac{2\pi im}{N}}$ eigenspace  in the $g^n$-twisted sector) at level $L_0-1= l$. 
 
 The positive roots satisfy
 \be \frac{m}{R}+\frac{n R}{N}>0\ .
 \ee
This condition, and therefore the Fock space used in the definition of $Z$, changes discontinuously whenever $R$ crosses a `wall of marginal stability', i.e. when the energy of  some single particle state changes sign. On the other hand, the path-integral description suggests that the index $Z$ itself should be a continuous function of $R$. Mathematically, the choice of $R$ determines which Weyl chamber is the fundamental one. Continuity of $Z$ corresponds to the statement that the denominator of the algebra is invariant under Weyl transformations, although this invariance is not manifest in its product formula description. Thus, we expect to be able to obtain the same expression for the index for different values of $R$. We will consider, in particular, the regimes $R\gg 0$ and $1/R\gg 0$.

\subsubsection{The $R\gg 0$ regime}

  For $R$ sufficiently large, a necessary and sufficient condition for single particle states to have positive energy, i.e.
\be  m>-n\frac{R^2}{N}\ ,  \ee
   is to have positive winding $n>0$. This follows from the following observations:
 \begin{enumerate}
 \item There are no states with zero winding $n=0$. This can be seen by noticing that $F_{0,m}(\tau)= \frac{1}{N}\sum_{k}e^{-\frac{2\pi imk}{N}} T_{\id,g^k}(\tau)$ has vanishing constant term, so that the multiplicity $\hat c_{0,m}(0)=0$. Note that, in general, there can be states with zero momentum $m=0$. This happens if some of the twisted-twining McKay-Thompson series $T_{g^r,g^s}(\tau)$ have constant terms.
 \item For $n<0$ and $R\gg 0$, the positive energy condition implies $m\gg 0$ and therefore $\frac{mn}{N}\ll 0$. On the other hand, all Fourier coefficients $\hat c_{n,m}(\frac{mn}{N})$ vanish when $\frac{mn}{N}$ is a sufficiently large negative number, so these states have zero multiplicity (in fact, by unitarity, we expect $\hat c_{n,m}(l)=0$ for $l<-1$).
 \end{enumerate}
The second quantized index is
\be\label{indexProd} Z_{g, \id}(T,U)=\Bigl(e^{-2\pi i T}\prod_{\substack{n> 0\\m\in \ZZ}}(1-e^{2\pi i U\frac{m}{N}}e^{2\pi iTn})^{\hat c_{n,m}(\frac{mn}{N})}\Bigr)^{24}\ ,
\ee
where $T,U$ are defined as (note the appearance of the factors $N$ with respect to the Monstrous Lie algebra)
\be T=b+i\frac{\beta R}{2\sqrt{2}\pi N}\qquad U=v+i\frac{\beta N}{2\sqrt{2}\pi R}\ .
\ee In this product formula, we set the constant orderings $(n_0,m_0)$ to $(-1,0)$ as for the (unorbifolded) Monster Lie algebra. We will justify this choice below.
 
 \medskip
 
Using the expression \eqref{Zidg}, it is easy to derive an alternative formula for $Z_{g,\id}(T,U)$. Indeed, in the 1-loop picture, $Z_{\id,g}$ can be simply computed by requiring the fields (and the strings) to transform by a $g$ transformation as one moves around the Euclidean time circle. This is exactly the same as taking a CHL orbifold by $(\delta,g)$ where the shift $\delta$ is taken along the Euclidean time direction, rather than along the space-like circle. This means that $Z_{g,\id}$ and $Z_{\id,g}$ are simply related by a rotation of the Euclidean $\TT^2$ torus that exchanges the space with Euclidean time. Such a rotation acts by $U\leftrightarrow -\frac{1}{U}$, $T\leftrightarrow T$ on the moduli, so that
\be\label{indexForm} Z_{g,\id}(T,U)^{1/24}=Z_{\id,g}(T,-\frac{1}{U})^{1/24}=T_{\id,g}(T)-T_{g,\id}(U)\ .
\ee Strictly speaking, since the derivation is based on the 1-loop expression for the index, we might expect a non-trivial phase to arise. This phase can be excluded by considering the index $Z_{g,\id}(T,U)$ in the large radius, low temperature limit $R,\beta\to \infty$, $\beta/R$ fixed (i.e. $T\to i\infty$, with $U$ fixed). In this limit, we expect the CHL model to be indistinguishable from the unorbifolded case, since all twisted states become infinitely massive and the effect of the orbifold projection is negligible when the momentum is approximately continuous. By requiring $Z_{g,\id}(T,U)^{1/24}\sim Z(T,U)^{1/24}\sim e^{-2\pi i T}$ in this limit, we find that \eqref{indexForm} must hold with no additional phase.  This reasoning also shows that the appropriate vacuum winding number is $n_0=-1$. The identity $Z_{g,\id}(T,U)=Z_{\id,g}(T,-\frac{1}{U})$ can also be derived directly by manipulating the infinite product formulas of both sides (see, for example, the proofs of Lemma 3.12 and Proposition 3.13 in \cite{Carnahan2014}).

We can use \eqref{indexForm} to describe the set of simple roots of the algebra in the regime $R\gg 0$. Since all positive roots have $n\ge 1$, then all roots with $n=1$ are simple. Single particle states corresponding to simple roots of the form $(1,m)$ give a non-vanishing contribution
\be -\sum_{m\in \ZZ} \hat c_{1,m}\left(\frac{m}{N}\right) e^{2\pi i U \frac{m}{N}}=-\sum_{m=1}^N F_{1,m}(U)=-T_{g,\id}(U)\ ,\ee
to the index. Here, we have to set $m_0=0$ in order to match with \eqref{indexForm}, so that the vacuum winding-momentum is $(n_0,m_0)=(-1,0)$, as anticipated in \eqref{indexProd}. In general, a single particle state corresponding to a simple root $(n,m)$ contributes
\be -e^{2\pi i T(n-1)+mU}\ ,
\ee to the index. By \eqref{indexForm}, the contributions of all simple roots that are \emph{not} of the form $(1,m)$ must be included in $T_{\id,g}(T)$, i.e. they must depend only on $T$. This means that all simple roots are either of the form $(1,m)$ or $(n,0)$. It follows that the vector $\rho=(-1,0)\in \h^*$ satisfies \eqref{Weyldef} for any simple root and is therefore a Weyl vector for the algebra. Thus, any CHL algebra has a (null) Weyl vector and it coincides with the normal ordering constants $(n_0,m_0)$ defining the vacuum winding and momentum.

The only states contributing to $(Z_{g, \id})^{1/24}$ are the ones satisfying
\be  0= \kappa(\gamma+\rho,\gamma + \rho)= 2 {M W \over N} =\frac{2}{N}m_{tot}(n_{tot}-1)\ .
\ee 
Indeed, $(Z_{g, \id})^{1/24}$ is the sum of a function depending only on $T$ (the contribution from states with $m_{tot}=0$) and a function depending only on $U$ (from states with $n_{tot}=1$).

\subsubsection{The $1/R\gg 0$ regime}

The positive energy condition
\be n>-\frac{Nm}{R^2}
\ee implies, for $1/R\gg 0$,
\be m>0,n\in \ZZ\qquad \text{or}\qquad m=0,n>0\ .
\ee In particular,  one can argue that there are no positive energy states with negative momentum $m<0$. The argument is analogous to the one used in the  $R\gg 0$ regime, just with the winding and momenta exchanged: if $m<0$, the positive energy condition requires $n\gg 0$, but the multiplicity $\hat c_{n,m}(\frac{mn}{N})$ vanishes for $mn\ll 0$.  

In fact, the two regimes $R\gg 0$ and $1/R\gg 0$ are similar, upon exchanging winding and momentum. The main qualitative difference is that, while there are no states with $n=0$ (because $\hat c_{0,m}(0)=0$ for all $m$), we cannot exclude the presence of states with zero momentum $m=0$. This phenomenon is related to the fact that, while the $R\to \infty $ limit corresponds to a two dimensional heterotic model compactified on $V^\natural$, the $R\to 0$ limit might correspond to a two dimensional heterotic model either on a $c=24$ VOA with currents (if there are states with $m=0$) or on one without currents (if there are no states with $m=0$). It is convenient to distinguish these two cases.

\medskip

\paragraph{\textit{Without zero momentum states.}}
Let us assume first that there are no states with $m=0$, i.e. that $\hat c_{n, 0}(0)=0$ for all $n$. Then, positive energy single particle states have necessarily $m>0$, so that the single-particle states with $m=1$ are simple roots for any winding $n$ and the Weyl vector $\rho$ corresponds to $(n_0,m_0)=(0,-1)$. 

The (single- or multi-particle) states contributing to the index $Z^{1/24}$ satisfy
\be 0=\kappa(\gamma+ 2\rho,\gamma)=\frac{2MW}{N}=\frac{2}{N}(m_{tot}-1)n_{tot}\ ,
\ee so that $Z^{1/24}$ is the sum of a function of $T$ (states with $m_{tot}=1$) and a function of $U$ (states with $n_{tot}=0$). 

The states with $n_{tot}=0$ consist of the vacuum, contributing $-e^{-\frac{2\pi i U}{N}}$, and multi-particle states contributing $-e^{\frac{2\pi i Um_{tot}}{N}}$, $m_{tot}>0$,  -- no single-particle states with $n=0$ exists. Notice that we have to assume that the vacuum in the $1/R\gg 0$ case has negative fermion number, in order to match with the $R\gg 0$ analysis.\footnote{The action of the fermion number $(-1)^F$ on the vacuum state $|0\rangle$ is a matter of convention. When $R$ crosses a domain wall, the fermion number of the vacua $|0\rangle$ and $|0'\rangle$ differ by a factor $(-1)^n$, where $n$ is the number of single particle fermionic states whose energy changes sign. In our case, since we have $24$ identical copies of each fermion, $n$ is always a multiple of $24$, so that all vacua have the same fermion number (that we conventionally fix to be positive). However, in the calculations, we often focus on the Fock space and index (that we denote by $Z^{1/24}$) built from a single copy of each fermion. In this case we might have vacua with different fermion numbers. By convention we fix the fermion number of the vacuum for $R\gg0$ to be positive.
}  In order to match with the $R\gg 0$ regime, all such contributions should sum up to $-T_{g,\id}(U)$. This implies that $T_{g,\id}(\tau)$ must be of the form 
\be\label{cuspa} T_{g,\id}(\tau)=q^{-\frac{1}{N}}+O(q^{1 \over N})\ ,
\ee with the polar term $q^{-\frac{1}{N}}$ coming from the vacuum contribution and the other terms coming from multi-particle states. 

The only states with $m_{tot}=1$ are single-particle states and their contribution to the index $(Z_{g, \id}(T, U))^{1/24}$ is
\be \sum_{n\in \ZZ} e^{2\pi iTn} \hat c_{n,1}\left(\frac{n}{N}\right)\ .
\ee By comparing with the $R\gg 0$ case, we find the identity
\be\label{stranona} \sum_{n\in \ZZ} e^{2\pi iTn} \hat c_{n,1}\left(\frac{n}{N}\right)=T_{\id,g}(T)=e^{-\frac{2\pi iT}{N}}+0+\ldots
\ee

Eq.\eqref{stranona} implies that, whenever $\hat c_{n,0}(0)=0$ for all $n$, the algebra  in the $1/R\gg 0$ regime has always one simple root $(-1,1)$ with multiplicity $\hat c_{-1,1}(-\frac{1}{N})=1$. Similarly, by \eqref{cuspa}, for $R\gg 0$, the algebra has always one simple root $(1,-1)$, again with multiplicity $\hat c_{1,-1}(-\frac{1}{N})=1$.
Roots with $\kappa(\gamma,\gamma)=\frac{2mn}{N}<0$ are called \emph{real} roots.  Notice that the only oscillators $\eta_a$ whose energy can (and do) change sign as we vary $R$ are the ones corresponding to real roots. Further implications of the existence of real roots in the algebra will be discussed in section \ref{s:Weyl}.

\medskip

\paragraph{\textit{With zero momentum states.}}

Let us consider the case where $\hat c_{n,0}(0)\neq 0$ for some $n>0$. Let $\hat n>0$ be the smallest winding such that $\hat c_{\hat n,0}(0)\neq 0$ and let $\tilde{n}>0$ the smallest winding such that $(\tilde{n},1)$ has non-zero multiplicity ${\rm mult}(\tilde{n},1)=\hat c_{\tilde n,1}\left(\frac{\tilde{n}}{N}\right)\neq 0$. Then, $(\hat n,0)$ and $(\tilde{n},1)$ are necessarily simple roots, so that the Weyl vector must be $(n_0,m_0)=(-\tilde{n},0)$. The ground states contributes $e^{-2\pi i \tilde{n} T}$ to the index $Z_{g,\id}(T,U)^{1/24}$ and comparison with the $R\gg 0$ regime shows that $\tilde{n}= 1$, so that $\hat c_{1,1}\left(\frac{1}{N}\right)\neq 0$ and $(n_0,m_0)=(-1,0)$. The condition $\kappa(\rho,\gamma)=-\kappa(\gamma,\gamma)/2$ implies that, as in the $R\gg 0$ regime, the simple roots are of the form $(1,m)$ or $(n,0)$. The states contributing to $(Z_{g,\id})^{1/24}$ satisfy
\be  \kappa(\gamma + 2 \rho, \gamma)=\frac{2}{N}MW=\frac{2}{N}m_{tot}( n_{tot}-1)=0\ ,
\ee so that states with $m_{tot}=0$ give a function of $T$ and states with $n_{tot}=1$ give a function of $U$. There is also a constant term from states with $(n_{tot},m_{tot})=(1,0)$, that are necessarily single-particle states; it is convenient to include this contribution in the function of $U$. Since the positive energy single-particle states have $m\ge 0$, all  contributions with $m_{tot}=0$ come from states built (in all possible ways) from single particles of zero momentum. These contributions  are given by
\be e^{-2\pi i T}\prod_{n>0} (1-e^{2\pi i Tn})^{\hat c_{n,0}(0)}+\hat c_{1,0}(0)\ ,
\ee where we subtract the contribution $-\hat c_{1,0}(0)$ from the (fermionic) states with $(n,m)= (1,0)$. By comparing with the $R\gg 0$ regime, we find that
\be\label{piustrana} T_{\id,g}(T)=e^{-2\pi i T}\prod_{n=1}^\infty (1-e^{2\pi i Tn})^{\hat c_{n,0}(0)}+\hat c_{1,0}(0)\ .
\ee This formula can be written as a product of $\eta$-functions (see appendix \ref{pf:etaproduct})
\be\label{etaproduct} T_{\id,g}(T)=\prod_{\ell|N}\eta(\ell T)^{\alpha(\ell)}+\alpha(1)\ ,
\ee where
\be \alpha(\ell)=\sum_{d|\ell}\hat c_{d,0}(0)\mu(\ell/d)\ ,
\ee with $\mu$ the M\"obius function. Using M\"obius inversion formula, it is easy to verify that the modular weight of this product is zero
\be \sum_{\ell|N}\alpha(\ell)=\hat c_{N,0}(0)=0\ .
\ee  The exponents $\alpha(\ell)$ are also related to the number of currents of the orbifold CFT $V^\natural/\langle g\rangle$, which is given by
\begin{equation}
\begin{split}\label{nrcurrents} \sum_{n=1}^N \hat c_{n,0}(0)=\sum_{d|N}\sum_{\substack{n=1\\(n,N)=d}}^N \hat c_{d,0}(0)=
\sum_{d|N} \hat c_{d,0}(0)\varphi\left(\frac{N}{d}\right)=
\sum_{d|N}\sum_{\ell|d} \alpha(\ell)\varphi\left(\frac{N}{d}\right)\\=\sum_{\ell|N}\alpha(\ell)\sum_{k|\frac{N}{\ell}} \varphi\left(\frac{N}{k\ell}\right)=\sum_{\ell|N}\alpha(\ell)\frac{N}{\ell}\ ,\end{split}
\end{equation}
where $\varphi(n)$ is the Euler totient function counting the numbers coprime to $n$ in $\ZZ/n\ZZ$.

Let us consider the states with $n_{tot}=1$. Each such state gives a contribution of the form $e^{\frac{2\pi iUm_{tot}}{N}}$, $m_{tot}\ge 0$; in particular, there are no poles in the limit $U\to i\infty$.  By comparison with the $R\gg0$ regime, these contributions should sum up to $-T_{g,\id}(U)$; this implies that $T_{g,\id}(\tau)$ is of the form
\be\label{nopole} T_{g,\id}(\tau)=const+O(q^{1/N})\ .
\ee This can be verified explicitly using \eqref{etaproduct}, since
\be\label{Tge} T_{g,\id}(\tau)=T_{\id,g}(-\frac{1}{\tau})=\prod_{\ell|N}\Bigl(\ell^{-1/2}\eta(\frac{\tau}{\ell})\Bigr)^{\alpha(\ell)}+\alpha(1)=\alpha(1)+O(q^{\sum_{\ell|N}\frac{\alpha(\ell)}{24\ell}})
\ee and noticing that 
\be\label{orderfirst} \sum_{\ell|N}\frac{\alpha(\ell)}{24\ell}=\frac{1}{24N}\sum_{\ell|N}\frac{\alpha(\ell)N}{\ell}>0
\ee since the right-hand side is, up to the $24N$ factor, the number of currents \eqref{nrcurrents} of the orbifold $V^\natural/\langle g\rangle$. Actually, the number of currents \eqref{nrcurrents} can be evaluated exactly. As we argued above, whenever there are zero momentum states, one has $\hat c_{1,1}(\frac{1}{N})\neq 0$, i.e. the coefficient of $q^{\frac{1}{N}}$ in $T_{g,\id}(\tau)$ must be non-zero. By \eqref{Tge}, this implies that the order \eqref{orderfirst} of the first non-constant coefficient is at most $1/N$, so that the number of currents  \eqref{nrcurrents} must be at most $24$ \be 0<\sum_{\ell|N}\frac{\alpha(\ell)N}{\ell}\le 24\ .\ee
The only VOA satisfying this bound is the Leech lattice VOA, which has exactly $24$ currents. We conclude that, whenever the orbifold $V^\natural/\langle g\rangle$ is consistent and has currents, it must be the Leech lattice VOA.

By \eqref{nopole}, $T_{g,g^k}(\tau)$ and $F_{1,r}(\tau)$ have no poles as $\tau\to i\infty$, for all $r,k\in \ZZ$, so that $\hat c_{1,r}(l)=0$ for $l<0$. Using these observations, it is easy to check that all contributions to $-T_{g,\id}(U)$ come from single-particle states with $n=1$, $m\ge 0$
\be -\sum_{m\ge 0} \hat c_{1,m}\left(\frac{m}{N}\right)e^{\frac{2\pi i m U}{N}}=-\sum_{m\in \ZZ} \hat c_{1,m}\left(\frac{m}{N}\right)e^{\frac{2\pi i m U}{N}}=-\sum_{r=1}^NF_{1,r}(U) =-T_{g,\id}(U)\ .
\ee
The absence of a polar term in $T_{g,\id}(U)$ implies that there are no simple real roots with $\kappa(\gamma,\gamma)=\frac{mn}{N}<0$. As a consequence, all simple (and therefore all positive) roots have $m,n\ge 0$.  This means that there is no oscillator $\eta_a$ for which the energy can change sign as we vary $R$. This is compatible with the observation that the Weyl vector $\rho=(-1,0)$ is the same in the $R\gg 0$ and $1/R\gg 0$ regime.

\subsubsection{The $\lambda \neq 1$ case}\label{sec:lambdaneq1}

Much of the analysis of the previous subsections carries over directly to the $\lambda \neq 1$ case if one simply makes the replacement $N \mapsto N \lambda$. In this subsection we will spell out the less automatic aspects of the generalization. In this case, the root lattice is given by
\be 
\gamma =\left(k^0_L, k^1_L \right) = {1 \over \sqrt{2}}\left({k \lambda + n \mathcal{E}_g \over R} + {n R \over N \lambda}, {k \lambda + n \mathcal{E}_g \over R} - {n R \over N \lambda} \right), \ n, k \in \mathbb{Z}
\ee
and, to emphasize the similitude with the previous subsections, we will use the notation $w := {n \over N \lambda}$ and $m:= k \lambda + n \mathcal{E}_g$.

The argument in the $R\gg 0$ regime directly applies to this case, after making the aforementioned substitution. In particular, the Weyl vector is $(n_0,m_0)=(-1,0)$ and the index is
\be Z_{g, \id}(T,U)=\Bigl(e^{-2\pi i T}\prod_{\substack{n> 0\\m\in \ZZ}}(1-e^{2\pi i U\frac{m}{N\lambda}}e^{2\pi iTn})^{\hat c_{n,m}(\frac{mn}{N\lambda})}\Bigr)^{24}=(T_{\id,g}(T)-T_{g,\id}(U))^{24}\ .
\ee Recall that $\hat c_{n,m}(l)=0$ unless $m- \E_g n\equiv 0\mod \lambda$. The infinite product formula then implies that
\be\label{unaphase} Z_{g, \id}\left(T-\frac{\E_g}{\lambda},U+N\right)^{1/24}=e^{\frac{2\pi i\E_g}{\lambda}}Z_{g,\id}(T,U)^{1/24}\ ,
\ee where the phase $e^{\frac{2\pi i\E_g}{\lambda}}$ comes from the vacuum contribution.

The ${1 \over R} \gg 0$ regime is slightly more involved, but the case without zero momentum states is very similar to its $\lambda=1$ counterpart. Namely, one computes as before
\be T_{g,\id}(U)=e^{-\frac{2\pi i U}{N\lambda}}+0+O(e^{\frac{2\pi i U}{N\lambda}})\ ,
\ee
and 
\begin{align}\label{lambdaneq1}
\sum_{n \in \mathbb{Z}}e^{2 \pi i T n}\hat{c}_{n, 1}\left(\frac{n}{N \lambda} \right) 
= T_{\id, g}(T)\ .
\end{align}
Notice that the coefficients $\hat{c}_{n, 1}\left(\frac{n}{N \lambda} \right)$ in this sum can be non-zero only for $n\E_g\equiv 1\mod\lambda$.   For theories without currents we can show that in fact the $\mathcal{E}_g \equiv -1 \textrm{ mod }\lambda$ condition always holds: we know the right hand side of \ref{lambdaneq1} starts with a pole of the form $e^{- 2 \pi i T}$ so the left hand side must have $\hat{c}_{-1, 1}\left( {-1 \over N \lambda}\right) \neq 0$. Then, the values $m=1, n=-1$ must satisfy $m\equiv n\E_g\mod \lambda$, and this enforces the condition $\mathcal{E}_g \equiv -1 \textrm{ mod }\lambda$.

\bigskip
On the other hand, assume that $\hat{c}_{n, 0}(0) \neq 0$, for some $n>0$ (to satisfy the positive energy condition). In this case, since $(\lambda, \mathcal{E}_g)=1$, we must have $\lambda \vert n$, so we can write $n = \bar{n} \lambda$ for some $\bar{n} \in \mathbb{N}$. This condition is necessary but not sufficient to have currents; in particular, it is violated when $n \equiv 0 \textrm{ mod } N$, so we allow for the possibility that $\hat{c}_{n, 0}(0)=0$ even if $\lambda\vert n$. The $T$-dependent part of the index again reduces to 
\be 
T_{\id,g}(T)=e^{- 2 \pi i T}\prod_{\bar n>0}\left(1 - e^{2 \pi i T  \bar n \lambda} \right)^{\hat{c}_{\bar n \lambda, 0}(0)}=\prod_{\ell|N}\eta(\ell T)^{\alpha(\ell)}
\ee
where now we have the condition $ \lambda\vert \ell$. The proof of \eqref{etaproduct} can be repeated without any essential modification. 
Turning to the states with $n_{tot}=1$, we can repeat the remainder of the argument after performing the $N \mapsto N\lambda$ substitution and find that $\hat{c}_{1, m}(l)= 0, l<0$ and the $-T_{g, \id}(U)$ contribution comes from the states
\be 
-\sum_{m \geq 0} \hat{c}_{1, m}\left(\frac{m}{N \lambda} \right) e^{\frac{2 \pi i m U }{ N \lambda}} = - \sum_{r=1}^{N \lambda}F_{1, r}(U) = -T_{g, \id}(U).
\ee

\section{Examples} \label{sec:examples}

\noindent In this subsection we compute some low-order, representative examples of the index to illustrate the properties that were discussed abstractly in the previous subsection. We include cases with and without currents, and cases where $\lambda \neq 1$.
\subsection{Elements of order $2$} 

We warm-up with the simplest possible examples. There are two classes of involutions in the Monster, class 2A and 2B. Let us start with the CHL model for class 2A. The McKay-Thompson series are
\be T_{\id,2A}(\tau)=\frac{\eta(\tau)^{24}}{\eta(2\tau)^{24}}+2^{12}\frac{\eta(2\tau)^{24}}{\eta(\tau)^{24}} + 24=\frac{1}{q}+4372q+96256q^2+\ldots
\ee
\be T_{2A,\id}(\tau)=T_{\id,2A}(\tau/2)=\frac{1}{q^{1/2}}+4372q^{1/2}+96256q+\ldots
\ee
so that
\begin{align} F_{0,0}(\tau)&=\frac{1}{q}+0+100628q+\ldots & F_{0,1}(\tau)&=96256q+\ldots\\
F_{1,0}(\tau)&=96256q+\ldots & F_{1,1}(\tau)&=\frac{1}{q^{1/2}}+4372q^{1/2}+\ldots
\end{align}
Consider the $R\gg 0$ regime. Let us compute the contribution to the index $(Z_{g, \id}(T,U))^{1/24}$ from states with $m_{total}=0$. The first step is to find the positive real roots, which correspond to $m<0$, $n>0$ and have multiplicity $\hat c_{n,m}(\frac{m n}{2})$. The only non-vanishing polar coefficients are $\hat c_{0,0}(-1)=1$ and $\hat c_{1,1}(-1/2)=1$. Now, it is impossible to have $\frac{mn}{2}=-1$ with $m,n\equiv 0\mod 2$, so that there is no root corresponding to $\hat c_{0,0}(-1)$. The only solution to $\frac{mn}{2}=-1/2$ with $m,n\equiv 1\mod 2$ is $(m,n)=(-1,1)$ and this root has multiplicity $\hat c_{1,1}(-1/2)=1$ (as always the case for real roots). There are no roots with $m=0$, since $\hat c_{n,0}(0)=0$ for all $n$. Thus, the only way to get $m_{tot}=0$ is, apart from the vacuum, a $2$-particle state of the form $(m_{tot},n_{tot})=(-1,1)+(1,n)$. There are $\hat c_{1,1}(-1/2)\hat c_{1,n}(\frac{n}{2})=\hat c_{1,n}(\frac{n}{2})$ such states. Therefore,
\begin{align} (Z(T,U))^{1/24}=&-T_{2A,\id}(U)+e^{-2\pi i T} +\sum_{n>0}c_{1,n}(\frac{n}{2})e^{2\pi i Tn}=-T_{2A,\id}(U)+F_{1,0}(2T)+F_{1,1}(2T)\\
=&-T_{2A,\id}(U)+T_{2A,\id}(2T)=-T_{2A,\id}(U)+T_{\id,2A}(T)\ .
\end{align}
Note that
\be T_{\id,2A}(\tau)=F_{0,1}(2\tau)+F_{1,1}(2\tau)\ .
\ee This is a special case of the general identity $T_{\id,g}(T)=\sum_{l=1}^N F_{l,1}(NT)$ that can be proved from \eqref{stranona}.

\medskip

Let us consider the CHL model for class 2B. The McKay-Thompson series are
\be T_{\id,2B}(\tau)=\frac{\eta(\tau)^{24}}{\eta(2\tau)^{24}}+ 24=\frac{1}{q}+276q-2048q^2+\ldots
\ee
\be T_{2B,\id}(\tau)=24+4096q^{1/2}+98304 q+\ldots
\ee so that
\begin{align} F_{0,0}(\tau)&=\frac{1}{q}+0+98580q+\ldots & F_{0,1}(\tau)&=98304q+\ldots\\
F_{1,0}(\tau)&=24+98304q+\ldots & F_{1,1}(\tau)&=4096q^{1/2}+\ldots
\end{align}
The only non-vanishing polar term is $\hat c_{0,0}(-1)=1$, but there is no $m,n$ such that $\frac{m n}{2}=-1$ with $m,n\equiv 0\mod 2$ and $m<0$, so there is no real positive root. The only way to obtain states with $m_{tot}=0$ is to consider multiparticle states formed from single particles with $m=0$. The latter appear for odd $n$ and have multiplicity $\hat c_{1,0}(0)=24$. The contribution from such states is given by
\begin{align} (Z_{g, \id}(T,U))^{1/24}=&-T_{2B,\id}(U)+e^{-2\pi i T}\prod_{\substack{n>0\\\substack{n\text{ odd}}}}(1-e^{2\pi iTn})^{24} + 24\\
=&-T_{2B,\id}(U)+e^{-2\pi i T}\prod_{n>0}\frac{(1-e^{2\pi iTn})^{24}}{(1-e^{2\pi iT(2n)})^{24}} + 24\\=&
-T_{2B,\id}(U)+\frac{\eta(T)^{24}}{\eta(2T)^{24}} + 24=-T_{2B,\id}(U)+T_{\id,2B}(T)\ .
\end{align}
In particular, we have the identity
\be T_{\id,2B}(T)=e^{-2\pi i T}\prod_{n=1}^{\infty}(1-e^{2\pi iT n})^{\hat c_{n,0}(0)} + \hat c_{1,0}(0)\ ,
\ee in agreement with \eqref{piustrana}.

\subsection{Elements of order $3$}

Again we will start with the simplest case, which is the conjugacy class $3A$. The expression for the McKay-Thompson series is
\be
T_{\id, 3A} = 12 + \left(\frac{\eta(q)}{\eta(q^3)}\right)^{12} + 729  \left(\frac{\eta(q^3)}{\eta(q)}\right)^{12} = \frac{1}{q}+783 q+8672 q^2+65367 q^3+O\left(q^4\right)
\ee
Similarly to the previous subsection, we have
\be 
T_{3A, \id}(\tau) = T_{\id, 3A}(\tau/3) = {1 \over q^{1/3}} + 783 q^{1/3} + 8672 q^{2/3} + O(q)
\ee
As before, we compute the functions $F_{n, m}$, for which the first few terms are
\begin{align*}
F_{0, 0}(\tau) &= \frac{1}{q}+66150 q+7170368 q^2+O\left(q^3\right) & F_{0, 1}(\tau) &= 65367 q+7161696 q^2+O\left(q^3\right) \\
F_{0, 2}(\tau) &= 65367 q+7161696 q^2+O\left(q^3\right) & F_{1, 0}(\tau) &= 65367 q+O\left(q^2\right) \\
F_{1, 1}(\tau) &= 783 q^{1/3 }+371508 q^{4/3}+O\left(q^{7/3}\right) & F_{1, 2}(\tau) &= \frac{1}{q^{1/3}}+8672 q^{2/3}+1741787 q^{5/3} + O\left(q^{8/3}\right)  \\
F_{2, 0}(\tau) &= 65367 q+O\left(q^2\right) & F_{2, 1}(\tau) &=\frac{1}{q^{1/3}}+8672 q^{2/3}+1741787 q^{5/3}+O\left(q^2\right) \\
F_{2, 2}(\tau) &= 783 q^{1/3}+371508 q^{4/3}+O\left(q^2\right) \\
\end{align*}

To write down the index, we again focus on the regime $R\gg 0$, and so choose states with $m<0, n>0$. Looking at the polar terms, we have the multiplicity $\hat{c}_{2, 1}(-1/3)=1$ of a positive real root with $m = -1 \equiv 2 \textrm{ mod } 3, n = 1 \equiv \textrm{ mod } 3$. Essentially identically to $2A$, we have two-particle states with $m_{tot}=0$ of the form $(-1, 1) + (1, n)$ with multiplicity $\hat{c}_{2, 1}(-1/3)\hat{c}_{1, n}(n/3) = \hat{c}_{1, n}(n/3)$. 
Plugging all this in, we have
\begin{align*}
Z_{g, \id}(T, U)^{1/24} &= e^{- 2 \pi i T}\left(\prod_{n>0, m\in \mathbb{Z}} (1- e^{2\pi i U {m \over 3}}e^{2\pi i T n}) \right)^{\hat{c}_{n, m}({m n \over 3})} \\
& = -T_{3A, \id}(U) + e^{-2 \pi i T} + \sum_{n>0}\hat{c}_{n, 1}(n/3)e^{2 \pi i T n} \\
& = -T_{3A, \id}(U) + F_{0, 1}(3T) + F_{1, 1}(3T) + F_{2, 1}(3T) \\
&= T_{\id, 3A}(T) - T_{3A, \id}(U)
\end{align*}
\bigskip
In the case of 3B, the McKay-Thompson series is 
\be 
T_{\id, 3B} = \left(\frac{\eta(q)}{\eta(q^3)} \right)^{12} + 12 = \frac{1}{q}+54 q-76 q^2-243 q^3+1188 q^4 +O\left(q^{5}\right)
\ee
with
\be 
T_{3B, \id}= 12+729 q^{1/3}+8748 q^{2/3}+65610 q+370332 q^{4/3}+1743039 q^{5/3}+O\left(q^2\right)
\ee
and
\begin{align*}
F_{0, 0}(\tau) &= \frac{1}{q}+65664 q+7164536 q^2 + O(q^3) & F_{0, 1}(\tau) &= 65610 q+O\left(q^2\right) \\
 F_{0, 2}(\tau) &= 65610 q+O\left(q^2\right) & F_{1, 0}(\tau) &= 12+65610 q+O\left(q^2\right) \\
F_{1, 1}(\tau) &= 729 q^{1/3}+370332 q^{4/3}+O\left(q^2\right) &  F_{1, 2}(\tau) &=8748 q^{2/3}+1743039 q^{5/3}+O\left(q^2\right)  \\
F_{2, 0}(\tau) &= 12 + 65610 q+O\left(q^2\right) & F_{2, 1}(\tau) &=  8748 q^{2/3}+1743039 q^{5/3}+O\left(q^2\right) \\
F_{2, 2}(\tau) &= 729 q^{1/3}+370332 q^{4/3}+O\left(q^2\right).
\end{align*}
As in the 2B case, there are no polar terms other than the one with multiplicity $1=\hat{c}_{0, 0}(-1)$, which does not have a solution for $m, n$ that satisfies our conditions. Therefore, there are no positive real roots here either. In order to get states with $m_{tot}=0$ we again consider combinations of single particle states with $m=0$. 
These have the form
\be 
e^{-2 \pi i T}\prod_{n=1}^{\infty}(1- e^{2 \pi i Tn})^{\hat{c}_{n, 0}(0)} + \hat{c}_{1, 0}(0)
\ee
Inspecting $F_{1, 0}$ we see that $\hat{c}_{1, 0}(0) = 12$ and $\hat{c}_{2, 0}(0)=12$ as well, so we are only excluding states of the form $n \equiv 0 \textrm{ mod } 3$. Thus we can re-write this contribution as
\be 
\left({\eta(q) \over \eta(q^3)}\right)^{12} + 12 = T_{\id, 3B}(T)
\ee
and the total index is
\be 
(Z_{g, \id}(T, U))^{1/24} = T_{\id, 3B}(T) - T_{3B, \id}(U)
\ee

\bigskip

Finally, we write the McKay-Thompson series for 3C. This example has $\lambda \neq 1$ and $N=\lambda=3$:
\be 
T_{\id, 3C}=\left(\frac{\eta(q^3)}{\eta(q^6)}\right)^8 + 2^8 \left(\frac{\eta(q^6)}{\eta(q^3)}\right)^{16} = \frac{1}{q}+248 q^2+4124 q^5+O\left(q^{8}\right)
\ee
Applying the S-transformation we get
\be 
T_{3C, \id} = \frac{1}{q^{1/9}}+248 q^{2/9}+4124 q^{5/9}+34752 q^{8/9}+O\left(q^{11/9}\right)
\ee
In this case, we have essentially 81 $F_{n, m}$ sectors but most of them will be zero. 
For compactness of notation, let us define the following shorthand:
\bea
\alpha(\tau) &= 248 q^{2/9} + 213126 q^{11/9} + \ldots \\
\beta(\tau) &= 4124 q^{5/9} + 1057504 q^{14/9} + \ldots \\
\gamma(\tau) &= {1 \over q^{1/9}} + 34752 q^{8/9} + 4530744 q^{17/9} + \ldots \\
\delta(\tau) &= 65628 q + 7164504 q^2 + \ldots \\
\kappa(\tau) &= {1 \over q} + 65628 q + 7164752 q^2 + \ldots
\eea 
The nonzero $F_{n, m}$ in this notation are:
\begin{align}\nonumber
F_{0, 0}(\tau) &= \kappa(\tau) & F_{0, 3}(\tau) &= \delta(\tau) & F_{0, 6}(\tau) &= \delta(\tau) \\ 
F_{1, 2}(\tau) &= \alpha(\tau)  & F_{1, 5}(\tau) &= \beta(\tau)  & F_{1, 8}(\tau) &= \gamma(\tau)   \\ \nonumber
F_{2, 1}(\tau) &= \alpha(\tau) & F_{2, 4}(\tau) &= \gamma(\tau) & F_{2, 7}(\tau) &= \beta(\tau) \\ \nonumber
F_{3, 0}(\tau) &= \delta(\tau) & F_{3, 3}(\tau) &= \kappa(\tau) & F_{3, 6}(\tau) &= \delta(\tau) \\ \nonumber
F_{4, 2}(\tau) &= \gamma(\tau) & F_{4, 5}(\tau) &= \alpha(\tau) & F_{4, 8}(\tau) &= \beta(\tau) \\ \nonumber
F_{5, 1}(\tau) &= \beta(\tau) & F_{5, 4}(\tau) &= \alpha(\tau) & F_{5, 7}(\tau) &= \gamma(\tau) \\ \nonumber
F_{6, 0}(\tau) &= \delta(\tau) & F_{6, 3}(\tau) &= \delta(\tau)  & F_{6, 6}(\tau) &= \kappa(\tau) \\ \nonumber
F_{7, 2}(\tau) &= \beta(\tau) & F_{7, 5}(\tau) &= \gamma(\tau) & F_{7, 8}(\tau) &= \alpha(\tau) \\\nonumber
F_{8, 1}(\tau) &= \gamma(\tau) & F_{8, 4}(\tau) &= \beta(\tau) & F_{8, 7}(\tau) &= \alpha(\tau) \nonumber
\end{align}
We will focus on the $R\gg 0$ regime again and so let $m<0, n>0$. There are no solutions for ${mn \over 9} = -1$ when $m, n$ are both congruent (mod $9$) to $0, 3, 6$. We can find a solution for ${m n \over 9} = -{1 \over 9}$ subject to our constraints, though: $m = -1 \equiv 8 \textrm{ mod } 9, n = 1 \equiv 1 \textrm{ mod } 9 $. This example then proceeds in the same way as a Fricke-invariant case of order $N\lambda = 9$.  
The two particle states with $m_{tot}=0$ can again be written as $(-1, 1) + (1, n)$ with multiplicity $\hat{c}_{n, 1}(n/9)$. 
Putting it together we again have the $U$-dependent piece
\be 
-\sum_{m \in \mathbb{Z}}\hat{c}_{1, m}\left({m \over 9}\right) e^{2 \pi i U {m \over 9}} = - \sum_{m=1}^{9}F_{1, m}(U) = -T_{3C, \id}(U)
\ee
and combining with the $T$-dependent piece:
\begin{align*}
Z_{g, \id}(T, U)^{1/24} &= e^{- 2 \pi i T}\left(\prod_{n>0, m\in \mathbb{Z}} (1- e^{2\pi i U {m \over 9}}e^{2\pi i T n}) \right)^{\hat{c}_{n, m}({m n \over 9})} \\
& = -T_{3C, \id}(U) + e^{-2 \pi i T} + \sum_{n>0}\hat{c}_{n, 1}(n/9)e^{2 \pi i T n} \\
& = -T_{3C, \id}(U) + F_{2, 1}(9T) + F_{5, 1}(9T) + F_{8, 1}(9T) \\
&= T_{\id, 3C}(T) - T_{3C, \id}(U)
\end{align*}
\subsection{Some elements of order $4$}
We also study the McKay-Thompson series corresponding to the conjugacy class $4D$. It is given by the expression
\be 
T_{\id, 4D} = \left({\eta(q^2) \over \eta(q^4)} \right)^{12} = {1 \over q} -12 q + 66 q^3 - 232 q^5 + O(q^7)
\ee
while
\be 
T_{4D, \id} = 64 q^{1/8}+768 q^{3/8}+4992 q^{5/8}+24064 q^{7/8}+O\left(q^{9/8}\right)
\ee
with $N=4, \lambda=2$. Somewhat similarly to the previous case, there are naively 64 sectors, but most will vanish by the arguments in Section \ref{sec:lambdaneq1}. For compactness, we will only display the minimum number of independent, nonvanishing $F_{n,m}(\tau)$. The others can be obtained by various symmetries relating the Hilbert spaces of the different CHL models. In particular (c.f. Appendix \ref{a:TwistAndTwin}), we have the relation $V_{n, m}\stackrel{\cong}{\longrightarrow} V_{n+N,m-\E_g N}$, which in this subsection becomes $V_{n, m}\stackrel{\cong}{\longrightarrow}V_{n+4,m \pm 4}$\footnote{In principle, this sign differs depending on whether we are looking at a case with or without currents, but the distinction is immaterial in these examples.}.
 
We have: 
\begin{align}
F_{0, 0}(\tau) &= \frac{1}{q}+49284 q + 5372928 q^2 + O(q^3) & F_{0, 2}(\tau)&= 49152 q + 5373952 q^2 + O(q^3) \\ \nonumber
F_{0, 4}(\tau) &= 49296 q + 5372928 q^2 + O(q^3) & F_{0, 6}(\tau)&= 49152 q + 5373952 q^2 + O(q^3) \\ \nonumber
F_{1, 1}(\tau) &= 64 q^{1/8} + O(q^{9/8})  &   F_{1, 3}(\tau) &= 768 q^{3/8} + O(q^{11/8}) \\ \nonumber
F_{1, 5}(\tau) &= 4992 q^{5/8} + O(q^{12/8}) & F_{1, 7}(\tau)&= 24064 q^{7/8} + O(q^{15/8}) \\ \nonumber
F_{2, 0}(\tau) &= 12 + 49152 q + O(q^2) & F_{2, 2}(\tau)&= 2016 q^{1/2} + O(q^{3/2}) \\  \nonumber
F_{2, 4}(\tau) &= 12 + 49152 q + O(q^2) & F_{2, 6}(\tau)&= 2080 q^{1/2} + O(q^{3/2}) \\  \nonumber
F_{3, 1}(\tau) &= 768 q^{3/8} + O(q^{11/8}) &   F_{3, 3}(\tau) &= 64 q^{1/8} + O(q^{9/8}) \\ \nonumber
F_{3, 5}(\tau) &= 24064 q^{7/8} + O(q^{15/8}) & F_{3, 7}(\tau)&= 4992 q^{5/8} + O(q^{12/8}) \\ \nonumber
\end{align}
In this case let's focus on the ${1 \over R} \gg 0$ regime. In this case we require $m>0, n \in \mathbb{Z}$ or $m=0, n>0$. Note that we \textit{do} have zero momentum states in this example. In particular, we have $\hat{c}_{2, 0}(0) = 12 = \hat{c}_{6, 0}(0)$. $\hat{c}_{1, 1}({1 \over 8}) = 64$, much as before. 
In this case we have, for the $T$-dependent piece 
\be 
e^{-2 \pi i T} \prod_{n>0}(1 - e^{2 \pi i T n})^{\hat{c}_{n, 0}(0)} = e^{-2 \pi i T} \prod_{n = 1, 3, \ldots}(1 - e^{2 \pi i T (2 n)})^{12} = \left({\eta(2 T) \over \eta(4 T)} \right)^{12} = T_{\id, 4D}(T)
\ee
The $U$-dependent term comes from the states with $n = 1$ and sum to 
$-\sum_{m \geq 0}\hat{c}_{1, m}(m/8) e^{2 \pi i m U/8} = -T_{4D, \id}(U)$. 

\bigskip
As a final example, consider the McKay-Thompson series for $4B$,
\be 
T_{\id, 4B} = \left({\eta(q^2) \over \eta(q^4)} \right)^{12} + 64\left({\eta(q^4) \over \eta(q^2)} \right)^{12} = {1 \over q} + 52 q+834 q^3+4760 q^5+24703 q^7 + O(q^9)
\ee
with 
\be 
T_{4B, \id} = {1 \over q^{1/8}}+52 q^{1/8}+834 q^{3/8}+4760 q^{5/8}+24703 q^{7/8}+94980 q^{9/8} + O(q^{11/8})
\ee
and
\begin{align}
F_{0, 0}(\tau) &= \frac{1}{q}+50340 q + 5397504 q^2 + O(q^3) & F_{0, 2}(\tau)&= 48128 q + 5349376 q^2 + O(q^3) \\ \nonumber
F_{0, 4}(\tau) &= 50288 q + 5397504 q^2 + O(q^3) & F_{0, 6}(\tau)&= 48128 q + 5349376 q^2 + O(q^3) \\ \nonumber
F_{1, 1}(\tau) &= 52 q^{1/8} + 94980 q^{9/8} + O(q^{17/8})  &   F_{1, 3}(\tau) &= 834 q^{3/8} + O(q^{11/8}) \\ \nonumber
F_{1, 5}(\tau) &= 4760 q^{5/8} + O(q^{12/8}) & F_{1, 7}(\tau)&= {1 \over q^{1/8}} + 24703 q^{7/8} + O(q^{15/8}) \\ \nonumber
F_{2, 0}(\tau) &= 48128 q + O(q^2) & F_{2, 2}(\tau)&= {1 \over q^{1/2}} + 2160 q^{1/2}+ O(q^{3/2}) \\  \nonumber
F_{2, 4}(\tau) &= 48128 q + O(q^2) & F_{2, 6}(\tau)&= 2212 q^{1/2} + O(q^{3/2}) \\  \nonumber
F_{3, 1}(\tau) &= 834 q^{3/8} + O(q^{11/8}) &   F_{3, 3}(\tau) &= 52 q^{1/8} + 94980 q^{9/8} + O(q^{17/8}) \\ \nonumber
F_{3, 5}(\tau) &={1 \over q^{1/8}} + 24703 q^{7/8} + O(q^{15/8}) & F_{3, 7}(\tau)&=  4760 q^{5/8} + O(q^{12/8}) \\ \nonumber
\end{align}
In this case we can compute directly in the $R\gg 0$ regime again. There is one positive root of multiplicity $\hat{c}_{1,7}(-1/8)=1$ with $n=1, m = -1 \equiv 7 \textrm{ mod } 8$; the other poles do not satisfy the congruence conditions in the $R\gg 0$ regime. If we build up the contribution from the multi-particle states as before we get
\be 
\sum_{n}F_{n, 1}(8 T) = F_{1, 1}(8 T) + F_{3, 1}(8T) + F_{5, 1}(8 T) + F_{7, 1}(8 T) = T_{\id, 4B}(T)
\ee
and again we get the same kind of contribution from the $U$-dependent side.

\section{Genus zero and Hauptmodul properties}\label{sec:Haupt}

\noindent In this section, we describe the properties of the McKay Thompson series $T_{\id,g}$ that can be deduced from physics arguments, and in particular from the properties of the index $Z_{g,\id}(T,U)$. Starting from these properties, we will then prove that each of these functions must be the Hauptmodul for a genus zero group.

\subsection{Space-like T-duality and Weyl reflections}\label{s:Weyl}

The analysis of section \ref{s:explformulas} shows that there are two classes of algebras that can emerge from the Monstrous CHL models. 

\medskip

The first class corresponds to the case where there the algebra has some roots with zero momentum, i.e. $\hat c_{n,0}(0)\neq 0$ for some $n$. Equivalently, this is the case where the orbifold $V^\natural/\langle g^\lambda\rangle$ has currents.  As argued in section \ref{s:explformulas}, both in the $R\gg 0$ and in the $1/R\gg0$ regimes, the positive roots are characterized by $n>0,m\ge 0$, the simple roots  are of the form $(n=1,m)$ or $(n,m=0)$ and the Weyl vector is $(n_0,m_0)=(-1,0)$. In particular there are no positive roots with $mn<0$ (\emph{real} roots), i.e. no oscillator $\eta_a$ whose energy can change its sign as we vary the radius $R$.

\medskip

The second class arises when $\hat c_{n,0}(0)=0$ for all windings $n$, i.e. when there are no roots with zero momentum. Equivalently, this is the case where the orbifold $V^{\natural}/\langle g^\lambda\rangle$ has no currents. In this case, at $R\gg0$, the positive roots are characterized by $n>0$, the simple roots are the ones with $n=1$ and the Weyl vector is $(n_0,m_0)=(-1,0)$.  At $1/R\gg 0$, the positive roots have $m>0$, the simple roots have $m=1$ and the Weyl vector is $(n_0,m_0)=(0,-1)$. There is a single pair $\pm \gamma= \pm(1,-1)$ of real roots,  both with multiplicity $1$, and we denote by $a\in \g_{(1,-1)}$ and $\theta(a)\in \g_{(-1,1)}$ the corresponding generators. In particular, $(1,-1)$ is positive at $R\gg 0$ and $(-1,1)$ is positive at $1/R\gg0$; the remaining positive roots are characterized by $m,n>0$ in both regimes.\footnote{To be precise, in section \ref{s:explformulas}, we proved that $\gamma$ is the only \emph{simple} real root and that it has multiplicity one. It is easy to show that there are no other \emph{positive} (not only simple) real roots. Indeed, if $a$ is the generator of $\g_{(1,-1)}$, the commutator of $a$ with any positive root with $m>0$ either has still $m>0$ or vanishes (remember that there are no roots with $m=0$ in this algebra). Since all positive roots are obtained by commutators of simple roots, we conclude.} This means that, as we cross the critical value of $R=\sqrt{N\lambda}$, the energy of the oscillator $\eta_a$ changes sign, so that the energy of the excited state $\eta_a|0\rangle$ gets lower than the energy of  $|0\rangle$ and becomes the new ground state. From the point of view of the algebra, this phase transition corresponds to a change of the subalgebra $\g^-$ (corresponding to the subset of positive definite oscillators), i.e. to a change of the Weyl chamber.

\bigskip

It is known that, for each real root $\gamma$ of a BKM algebra, there is a Weyl reflection $r_\gamma$, i.e. an automorphism of the algebra such that $r_{\gamma}(\gamma)=-\gamma$ and, more generally,
\be r_{\gamma}(\gamma')=\gamma'-2\gamma \frac{\kappa(\gamma,\gamma')}{\kappa(\gamma,\gamma)}\ .
\ee  The two Weyl chambers relative to $R>\sqrt{N\lambda}$ and $R<\sqrt{N\lambda}$ are related by the Weyl reflection $r_\gamma$ corresponding to the real root $\gamma=(1,-1)$. This Weyl reflection acts by
\be r_{\gamma}(n,m)=(n,m)-2(1,-1)\frac{m-n}{-2}=(m,n)\ .
\ee Since this is an automorphism of the algebra, the multiplicities ${\rm mult}(n,m)$ and ${\rm mult}(r_\gamma(n,m))$ must be the same, i.e.
\be \hat c_{n,m}\left(\frac{mn}{N\lambda}\right)=\hat c_{m,n}\left(\frac{mn}{N\lambda}\right)\ .
\ee This implies the following symmetry for the index
\begin{align}
(Z_{g, \id}(T,U))^{1/24}&=(e^{-2\pi i T}-e^{-2\pi i \frac{U}{N\lambda}})\prod_{m,n> 0}(1-e^{2\pi i U\frac{m}{N\lambda}}e^{2\pi iTn})^{\hat c_{n,m}(\frac{mn}{N\lambda})}\\
&=-(e^{-2\pi i \frac{U}{N\lambda}}-e^{-2\pi i T})\prod_{m,n> 0}(1-e^{2\pi i U\frac{m}{N\lambda}}e^{2\pi iTn})^{\hat c_{m,n}(\frac{mn}{N\lambda})}\\
&=-\left(Z_{g, \id}\left(\frac{U}{N\lambda},N\lambda T\right)\right)^{1/24}\ ,
\end{align} where we have separated the contribution of the ground state and the real positive root from the contribution of the imaginary positive roots.
Therefore, the Weyl transformation exchanges the winding and momenta along the space-like circle and transforms $T\rightarrow \frac{U}{N\lambda}$ and $U\leftrightarrow N\lambda T$, which (for $b=v=0$) corresponds to $R\leftrightarrow \frac{N\lambda}{R}$. 

The physical interpretation is clear: the Weyl reflection corresponds to T-duality along the space-like circle. T-duality along a single direction is not in the component $\tilde{SO}^+(L)$ of the T-duality group connected to the identity. However, we can compose it with T-duality \eqref{timeTduality} along the Euclidean time direction, which is always a self-duality for all Monster CHL models. The composition of the two T-dualities gives
\be Z_{g,\id}(T,U)=Z_{g,\id}\left(-\frac{1}{N\lambda T},-\frac{N\lambda}{U}\right)\ .
\ee It is easy to identify this T-duality  with the Fricke involution $(W_{N\lambda},W_{N\lambda})\in \tilde{SO}^+(L)$. 

The same reasoning as for the Fricke involution in section \ref{s:Tdualities} shows that T-duality along the space-like circle exchanges the CHL model relative to $(V^\natural,g)$ at the radius $R$ with the CHL model relative to the orbifold $(V',g')$ at the radius $N\lambda/R$, where $V'=V^\natural/\langle g\rangle$ and $g'$ is the quantum symmetry $Q$. The fact that the Weyl reflection is an automorphism of the associated BKM algebra implies that the BKM algebras based on these two models are isomorphic. This suggests that also the underlying CFTs are isomorphic, so that this T-duality is really a self-duality, i.e. an equivalence of the \emph{same} CHL model at different values of the moduli.\footnote{This implication is far from trivial: in general, it is not known whether isomorphic BKM algebras can only arise from isomorphic CFTs. We thank Scott Carnahan for clarifying this point to us.} This is indeed the case, as we will now show. The generators $a$ and $\theta(a)$, corresponding to the real roots $(1,-1)$ and $(-1,1)$, and the Cartan generator $H_a=\frac12 [a,\theta(a)]$ form a $\mathfrak{sl}_2$ subalgebra of the CHL algebra. At the self-dual radius $R=R_{sd}:=\sqrt{N\lambda}$, the corresponding 1-particle BPS states have all zero energy, which by the BPS condition implies $k^\mu_R=0$.  This means that the CFT $(V^\natural\times S^1)/(\delta,g)$ at the self-dual radius contains three purely holomorphic currents with winding and momenta $(-1,1)$, $(0,0)$ and $(1,-1)$, respectively; the zero energy BPS states are formed by tensoring these three currents with one of the Ramond ground states of $\bar V^{s\natural}$.  The zero modes of these three holomorphic currents generate a $SU(2)$ group which is a symmetry of $(V^\natural\times S^1)/(\delta,g)$ at the self-dual radius. This $SU(2)$ symmetry contains a $\ZZ_2$ subgroup acting by $k^1_L\to -k^1_L$ and fixing $k^0_L,k^0_R,k^1_R$ and can be identified with space-like T-duality. The fact that T-duality at the self-dual radius becomes part of a continuous symmetry group  is a familiar phenomenon in string theory. However, its occurrence in CHL orbifolds depends on the existence of the holomorphic currents generating this continuous group.  Any CFT whose radius $R$ is infinitesimally close to $R_{sd}$ can be obtained as a conformal perturbation of the one at the self-dual radius. The model at radius $R=R_{sd}(1+\epsilon+O(\epsilon^2))$ and the model at radius $1/R=R_{sd}(1-\epsilon+O(\epsilon^2))$ are then equivalent, since the deformations $\pm \epsilon R_{sd}$ are related by the $\ZZ_2$ subgroup of the $SU(2)$ symmetry. Thus, for $R$ in a neighbourhood of $R_{sd}$, the CFT $(V^\natural\times S^1)/(\delta,g)$ at radius $R$ is equivalent to the same CFT at radius $1/R$. But the latter is also equivalent to the CFT $(V'\times S^1)/(\delta,g')$. Following the chain of equivalences, we conclude that the CFTs based on $(V^\natural,g)$ and the one based on $(V',g')$ are equivalent at the same radius $R$. This can only happen if $V'\cong V^\natural$ and this isomorphism can be chosen so that $g$ and $g'$ are the same symmetry.

As a consequence, whenever the orbifold $V^\natural/\langle g^\lambda\rangle$ has no currents, both space-like T-duality and the Fricke involution $(W_{N\lambda},W_{N\lambda})\in \tilde{SO}^+(L)$ are self-dualities of the model. In particular, T-duality along the space-like circle alone corresponds to the Weyl reflection with respect to the (unique) positive real root of the BKM algebra. 

\bigskip

One of the most striking consequences of this construction is that whenever the orbifold $V'=V^\natural/\langle g\rangle$ is consistent (i.e. $\lambda=1$) and has no currents, then it is isomorphic to $V^\natural$.   Furthermore, the quantum symmetry $Q$ is in the same conjugacy class as $g$. This results extends to the case $\lambda>1$: if the orbifold $V'=V^\natural/\langle g^\lambda\rangle$ has no currents, then it is isomorphic to $V^\natural$, and the symmetry $gQ$ on the orbifold $V'$ is mapped, via this isomorphism, to an element in the same class as $g$.

More generally, one can show that for every $h\in \Aut(V^{\natural})$ that commutes with $g$ and fixes  the generator $a$ of the root $(1,-1)$, the induced symmetry $h'\in \Aut(V')$ is in the same Monster class as $h$. The idea is that, in this case, $h$ commutes with the $SU(2)$ symmetry at the self-dual radius containing the space-like T-duality.  Therefore, the isomorphism $V^\natural\cong V'$ induced by T-duality must map $h'\in \Aut(V')$ to $h\in \Aut(V^{\natural})$, so that these two symmetries must be in the same Monster class.

In Appendix \ref{pf:WeylCompatible}, we will prove that whenever $g$ and $h$ commute and have coprime orders, then $h$ fixes $a$
\be\label{WeylCompatible} \gcd(o(g),o(h))=1\qquad \Rightarrow\qquad h(a)=a\ ,
\ee so that the argument above applies\footnote{In the remainder of this section, we will use the shorthand $(a, b):=\text{gcd}(a, b)$.}.
 Now, consider some $g\in \Aut(V^\natural)$ of order $N$ and multiplier $\lambda$, which is the product $g=h_1h_2$ of symmetries $h_1,h_2$ of coprime orders $N_1,N_2$ and multipliers $\lambda_1=(N_1,\lambda)$ and $\lambda_2=(N_2,\lambda)$. Then, $e:=\frac{N_1}{\lambda_1}$ is an exact divisor of $\frac{N}{\lambda}$ and every exact divisor can be obtained in this way. Suppose that $V'=V^\natural/\langle h_1^{\lambda_1}\rangle$ has no currents. Then, $V'\cong V^\natural$ and the symmetry $g'=(Qh_1)h_2$ of $V'$ is in the same Monster class as $h_1h_2$. But this is exactly the condition for the CHL model based on $g$ to be self-dual under the Atkin-Lehner involution $(w_e,w_e)$ (notice that $\langle g^{N/e}\rangle=\langle h_1^{\lambda_1}\rangle$ ). We conclude that the CHL model based on $g$ is self-dual under $(w_e,w_e)$ if and only if $V^\natural/\langle g^{N/e}\rangle$ has no currents.

\subsection{Decompactification limits}\label{s:decompactify}

In this section we will analyze the behaviour of the index $Z_{g,\id}(T,U)$ at the boundary of the moduli space $\HH\times \HH$. We will focus on the region of the moduli space where $T_2,U_2\gg 1$, i.e. the low temperature regime $\beta\gg 1$. This condition assures that the leading contribution to the index is given by the vacuum state. When the temperature grows, one expects a Hagedorn-type phase transition. In particular, the infinite product formula for $Z_{g,\id}(T,U)$ diverges above the Hagedorn temperature, signalling that the Fock space description of the index cannot be trusted in this regime.  High temperature regions are most naturally described in dual pictures, for example using T-duality along the Euclidean time (thermal) circle. 

The first limit we consider is
\be\label{decompa} T\to i\infty \qquad U\text{ fixed},\quad U_2\gg 1\ ,
\ee  corresponding to
\be \beta,R\to \infty\qquad \frac{\beta}{R}\text{ fixed}.
\ee This can be interpreted as a decompactification limit where the volume of the Euclidean $\TT^2$ torus becomes infinite. For $R\gg 0$, the ground state has energy
\be\label{vacuumenergy} E_0=-\frac{R}{\sqrt{2}}\ ,
\ee and its contribution $e^{-\beta E_0}$ to the index diverges as $R\to \infty$
\be Z_{g,\id}(T,U)^{1/24}=T_{\id,g}(T)-T_{g,\id}(U)\sim e^{2\pi T_2} \to \infty\ ,\qquad R\to \infty\ .
\ee 

More generally, we are interested in studying the possible divergences of the index $Z_{g,\id}(T,U)$ in the limits
\be\label{Atkindecompa} W_e\cdot T\to i\infty\ , \qquad\qquad W_e\cdot U\text{ fixed},\quad  (W_e\cdot U)_2\gg 1\ ,
\ee where $W_e$ is some Atkin-Lehner involution for $\Gamma_0(N\lambda)$, whose action on $T,U$ is as in \eqref{action}. Eq.\eqref{Atkindecompa} can be interpreted as a decompactification limit for the T-dual CHL model $(V',g')$ related to $(V^\natural,g)$ by the duality $(W_e,W_e)$. 
The index $Z_{g,\id}(T,U)$ diverges in the limit \eqref{Atkindecompa} if and only if the McKay-Thompson series $T_{g}$ has a pole (i.e., is \emph{unbounded}) at the cusp $W_e\cdot i\infty$. 

Let us focus first on the Fricke involution $W_{N\lambda}$, for which \eqref{Atkindecompa} becomes
\be\label{Frickedecompa} T\to 0\ , \qquad\qquad  U\text{ fixed},\quad  U_2\ll 1\ .
\ee This is a  high temperature limit, whose direct analysis is quite complicated.  Therefore, it is useful to perform a T-duality $\bfT:T\leftrightarrow -\frac{1}{U}$ along the Euclidean time circle, so that \eqref{Frickedecompa} becomes
\be\label{decompa2} U\to i\infty\ ,\qquad T\text{ fixed},\quad T_2\gg 1\ ,
\ee which is a small radius, low temperature limit
\be \beta\to \infty,\ R\to 0,\qquad \beta R\text{ fixed}\ .
\ee Since the T-duality $\bfT$ along the Euclidean time is always a self-duality of any CHL model, the two limits \eqref{Frickedecompa} and \eqref{decompa2} are equivalent.

We can smoothly vary the moduli from the large radius limit \eqref{decompa} to the small radius limit \eqref{decompa2} while keeping the temperature low $\beta\gg 0$, so that the main contribution to the index is always given by the ground state of the model. It is clear that if the state with energy \eqref{vacuumenergy} remains the ground state of the theory as we shrink the radius $R$ all the way to zero, then the index $Z_{g,\id}(T,U)$ is necessarily bounded in the limit \eqref{decompa2}. As discussed in section \ref{s:Algebras}, this is the behaviour expected in the case where the dual theory $V'=V/\langle g^\lambda\rangle$ has currents.

Therefore, the only possibility for the index $Z_{g,\id}(T,U)$ to diverge in the small radius limit \eqref{decompa2} is that the model undergoes a phase transition of the kind described in section \ref{s:Weyl}. Recall that this happens if and only if the dual model $V'$ has no currents. Suppose such a phase transition occurs for a given CHL model. This means that there exists some excited state, whose energy gets lower than \eqref{vacuumenergy} as the radius $R$ crosses a critical (self-dual) value $R_{sd}$; such excited state becomes the new ground state for $R<R_{sd}$. The low temperature region $T_2,U_2>1$ is divided into two different phases, separated by a critical manifold \be T=\frac{U}{N\lambda}\ ,\ee of codimension one in the moduli space. The critical manifold can be characterized as the locus of the moduli space that is fixed under T-duality $T\leftrightarrow \frac{U}{N\lambda}$ along the  space-like circle. As discussed in section \ref{s:Weyl}, at the critical manifold, the model has an enhanced gauge symmetry containing, in particular, this space-like T-duality, which is therefore a self-duality of the model. The space-like T-duality exchanges the limits \eqref{decompa} and \eqref{decompa2}, which are therefore equivalent. In section \ref{s:Weyl}, we also argued that this phase transition implies that the orbifold theory $V'=V^\natural/\langle g^\lambda\rangle$ is isomorphic to $V^\natural$ and that the CHL model is self-dual under the Fricke involution $(W_{N\lambda},W_{N\lambda})$. This argument provides a neat physical understanding of the relationship between Fricke invariance and unboundedness of a McKay-Thompson series. The index $Z_g(T,U)$ is divergent in the limit \eqref{decompa2} if and only if there is a phase transition at the critical manifold $T=\frac{U}{N\lambda}$, and the latter occurs if and only if the Fricke involution is a self-duality.

\bigskip

Similar arguments apply to the limits  \eqref{Atkindecompa} for more general Atkin-Lehner involutions $W_e=\frac{1}{\sqrt{e}}\left(\begin{smallmatrix}
ae & b\\ N\lambda c & de
\end{smallmatrix} \right)$, where $e||N\lambda$. Let us sketch the basic steps of the reasoning. By composing the Atkin-Lehner duality and the T-duality $\bfT$ along the Euclidean time circle, we obtain a T-duality along a circle $S$ of the Euclidean torus $\TT^2$. For a general $W_e$, the circle $S$ is not aligned along the space-like direction. The fixed locus for the T-duality along $S$ is a critical submanifold of codimension $1$ in the moduli space where phase transitions can possibly occur. 
In general, the critical manifold does not intersect the low temperature region $T_2,U_2\gg 1$, so that its physical interpretation is not clear. For this reason, it is useful to introduce a different duality frame, where the Euclidean torus $\TT^2$ is rotated in such a way that the circle $S$ is the space-like direction. The rotation acts by
\be T\to T'=T\qquad \qquad U\to U'=\frac{-c U+ae}{-dU+b\frac{N}{e}}\ ,
\ee on the moduli, and the critical manifold is now given by the equation
\be\label{AtkinCritical} T'=\frac{U'}{e}\ .
\ee 
Since the rotation mixes the space and the Euclidean time direction, the physics in the rotated frame looks quite different than in the original one. In particular,  the critical manifold \eqref{AtkinCritical} now intersects the low temperature region $T_2',U_2'\gg 1$ in the rotated frame.   Furthermore, the CHL orbifold in the rotated frame involves shifts both in the space and in the time direction. The net result is that the index $Z_{g,\id}(T,U)$ is interpreted in the rotated Frame as a `twisted-twined index' 
\be Z_{g,\id}(T,U)=Z_{h_1,h_2}(T',U')\ .
\ee Here, $Z_{h_1,h_2}$ is the index `twined' by the symmetry $h_2=g^{e}$ in the CHL model $(V^\natural, h_1)$, with $h_1=g^{N\lambda/e}$. 
The limit \eqref{decompa} is interpreted as a large radius, low temperature limit both in the original and in the rotated frame. The limit \eqref{Atkindecompa} is equivalent (upon taking T-duality $\bfT$ in the original Euclidean time direction) to $U'\to i\infty$, with $T'$ fixed, which is interpreted as a low temperature, small radius limit in the rotated frame. The index $Z_{h_1,h_2}(T',U')$ can have a pole in this small radius limit if and only if  the  CHL model $(V^\natural,h_1)$ undergoes a phase transition at the critical submanifold \eqref{AtkinCritical}. As argued in section \ref{s:Weyl}, the occurrence of this phase transition implies the self-duality of the CHL model  $(V^\natural,h_1)$ under the associated space-like T-duality \be\label{spaceTduality} T'\leftrightarrow \frac{U'}{e} ,\ee and also self-duality of the CHL model $(V^\natural,g)$ under the Atkin-Lehner involution $(W_e,W_e)$. We conclude that the index $Z_{g,\id}(T,U)$ diverges in the  limit \eqref{Atkindecompa} if and only if the CHL model is self-dual under $(W_e,W_e)$; in this case the limits \eqref{Atkindecompa} and \eqref{decompa} are physically equivalent.

\subsection{The McKay-Thompson series are Hauptmoduln}\label{s:McHaupt}
The Monstrous moonshine conjecture claims that the McKay-Thompson series $T_{\id,g}$ are modular under a group $\Gamma_g\subset SL(2,\RR)$ of genus zero in the normalizer of $\Gamma_0(N)$ and that they are Hauptmoduln for such group, i.e. degree $1$ holomorphic maps from $\overline{\HH /\Gamma}_g$ to the Riemann sphere $\overline \CC$.  In this section, we will derive this conjecture using the properties of the index $Z_{g,\id}(T,U)$. It is useful to work with the \emph{eigengroup} $\Gamma_g'$ under which $T_{\id,g}$ is invariant up to a $24$-th root of unity, rather than the \emph{fixing group} $\Gamma_g$ under which $T_{\id,g}$ is exactly invariant.

\bigskip

The group of self-dualities $G_g$ of the CHL model $(V^\natural,g)$ contains $\Gamma_0(N\lambda)\times \Gamma_0(N\lambda)$ as well as the transformations \eqref{lambtransf}. Since $|Z_{g,\id}(T,U)|^2$ is invariant under $G_g$, the series $T_{\id,g}$ must be invariant (up to a phase) under the projection $\text{proj}_1(G_g)$ of this group onto the first factor of $SL(2,\RR)\times SL(2,\RR)$. This projection contains the group $\Gamma_0(N|\lambda)$, so that $T_{\id,g}$ must be modular, up to a multiplier, under this group. Modularity of $T_{\id,g}$ under $\Gamma_0(N)$ (up to a multiplier) also follows by standard CFT arguments -- in fact, this was one of our starting points in the construction of the CHL orbifolds. In particular, $T_{\id,g}$ is invariant under $\Gamma_0(N\lambda)$ and transforms with a multiplier of order $\lambda$ under $\Gamma_0(N)$ (see appendix \ref{a:TwistAndTwin}).  The group $\Gamma_0(N|\lambda)$ is generated by $\Gamma_0(N)$ together with the transformation $\tau\mapsto \tau+\frac{1}{\lambda}$. Using Eq.\eqref{unaphase} in section \ref{s:explformulas}, one can prove that 
\be T_{\id,g}\left(\tau+\frac{1}{\lambda}\right)=e^{-\frac{2\pi i }{\lambda}}T_{\id,g}(\tau)\ .
\ee This implies that $T_{\id,g}$ is a modular form for $\Gamma_0(N|\lambda)$ with multiplier of order $\lambda$, as expected.

\bigskip

Let us consider the limits of the McKay-Thompson series $T_{\id,g}$ at the boundary of the quotient space $\overline{\HH/\Gamma_0(N|\lambda)}$. This boundary consists of a finite number of points, corresponding to the orbits of $\QQ\cup\{\infty\}$  (cusps)  under $\Gamma_0(N|\lambda)$. We say that a McKay-Thompson series is \emph{bounded} (respectively, \emph{unbounded}) at a certain cusp if the limit of $T_{\id,g}(\tau)$ at the cusp is finite (resp., infinite).
 It is obvious by construction that $T_{\id,g}$ has a single pole at the cusp at $\infty$. Furthermore, at the cusp $0$, $T_{\id,g}(\tau)$ is either bounded  or it has a single pole with coefficient $1$. More precisely,
\be\label{zerocusp} T_{g,\id}(\tau)=\begin{cases} O(q^0) & \text{if }V^\natural/\langle g^\lambda\rangle\text{ contains currents}\\
q^{-\frac{1}{N\lambda}}+0+O(q^{\frac{1}{N\lambda}}) & \text{otherwise} \ .\end{cases}
\ee

\bigskip

Let us consider now the limit of $T_{\id,g}$ at the other cusps, different from $0$ and $\infty$. Consider first the case $\lambda=1$. Each cusp for $\Gamma_0(N)$  has a representative of the form $\frac{a}{c}$, where $a,c\in \ZZ_{>0}$ with $c|N$ and $(a,c)=1$ (of course, it can happen that two rational numbers of this form are equivalent under $\Gamma_0(N)$). The group $\Gamma_0(N|\lambda)$ is conjugated to $\Gamma_0(N/\lambda)$ by $\left(\begin{smallmatrix}
\lambda & 0\\ 0 & 1
\end{smallmatrix}
\right)$, so that each cusp of $\Gamma_0(N|\lambda)$ has a representative of the form $\frac{a}{\lambda c}$, where $a,c\in \ZZ_{>0}$, with $c|\frac{N}{\lambda}$ and $(a,c)=1$.
In particular, all cusps of the form $\frac{a}{N\lambda}$, with $(a,N\lambda)=1$, are equivalent to $\infty$ and all cusps of the form $\frac{a}{\lambda}$ are equivalent to $0$; thus, another useful set of representatives is given by $0,\infty$ and rationals of the form $\frac{a}{\lambda c}$ as above, with the restriction $1<c<N/\lambda$. 

Some of the cusps $\cc=\frac{a}{c\lambda}$ are related to $\infty$ by some Atkin-Lehner involution for $\Gamma_0(N|\lambda)$, i.e. $\cc=w_\cc\cdot \infty$. Such an involution exists if and only if $c$ is an exact divisor of $N/\lambda$. In particular, $0$ is always related to $\infty$ by the Fricke involution.

Using the properties above, we can prove the following:

\begin{lemma}\label{t:lecuspe} If $T_{\id,g}$ is unbounded at a cusp $\cc$ for $\Gamma_0(N|\lambda)$, then there exists an Atkin-Lehner involution $w_\cc$ for $\Gamma_0(N|\lambda)$ such that $\cc=w_\cc\cdot \infty$.
\end{lemma}
See Appendix \ref{pf:lecuspe} for the proof.

\bigskip

We are now ready to prove that the fixing groups $\Gamma_g$ are genus zero and that the McKay-Thompson series are the corresponding Hauptmoduln. The main ingredient is the result of section \ref{s:Weyl} that the Atkin-Lehner dualities $(w_e,w_e)\in \tilde{SO}^+(L)$, $e||\frac{N}{\lambda}$ are self-dualities whenever the orbifold $V^\natural/\langle g^{N/e}\rangle$ has no currents. 

\begin{theorem}\label{t:mainth} Let $g\in \MM$ have order $N$ and $T_{\id,g}$ be modular under $\Gamma_0(N|\lambda)$. Then the eigengroup $\Gamma_g'$ of $T_{\id,g}$ is the projection $\mathrm{proj}_1(G_g)$ to the first $SL(2,\RR)$ factor of the group $G_g\subset SL(2,\RR)\times SL(2,\RR)$ of self-dualities of the corresponding CHL model. Furthermore, the fixing group $\Gamma_g$ has genus zero and $T_{\id,g}$ is a Hauptmodul.
\end{theorem}

\begin{proof} It is clear that $\text{proj}_1(G_g)$  is contained in the eigengroup $\Gamma'_g$, since $G_g$ leaves the index $Z_{g,\id}(T,U)$ invariant up to a phase. Furthermore, we argued above that $\text{proj}_1(G_g)$ contains $\Gamma_0(N|\lambda)$ as a normal subgroup. Since $\Gamma_g'$ is generated by $\Gamma_0(N|\lambda)$ and some Atkin-Lehner involutions $w_e$, $e||\frac{N}{\lambda}$, we only have to prove that, whenever $w_e\in \Gamma_g'$ then $(w_e,w_e)\in G_g$. To this aim, notice that if $w_e\in \Gamma_g'$, then the cusp $w_e\cdot \infty$ is unbounded. This happens only if the orbifold $V^\natural/\langle g^{N/e}\rangle$ has no currents. By the argument in section \ref{s:Weyl}, this implies that the model is self-dual under $(w_e,w_e)$. Therefore, $w_e\in \Gamma_g'$ implies $w_e\in\text{proj}_1(G_g)$, so that $\Gamma_g'=\text{proj}_1(G_g)$.

As for the genus zero properties, by lemma \ref{t:lecuspe} the only possibly unbounded cusps for $T_{\id,g}$, modulo $\Gamma_0(N|\lambda)$, are at $w_e\cdot \infty$ for some $e||\frac{N}{\lambda}$. This cusp is unbounded if and only if the orbifold $V^\natural/\langle g^{N/e}\rangle$ has no currents, and in this case the group of self-dualities $G_g$ contains  $(w_e,w_e)$ and the eigengroup $\Gamma'_g$ contains $w_e$. Therefore, all unbounded cusps for $T_{\id,g}$ are related to $\infty$ by some $\gamma'$ in the modular eigengroup $\Gamma_g'$. It is easy to see that every such unbounded cusp $\cc$ must be related to $\infty$ also by an element $\gamma$ of the fixing group $\Gamma_g$. Indeed, suppose that $\cc = \gamma'\cdot \infty$, $\gamma'\in \Gamma'_g$, and  that the series $T_{\id,g}$ is invariant under $\gamma'$ up to a phase $e^{\frac{2\pi i  r}{\lambda}}$. Using the fact that the element $\left(\begin{smallmatrix}
1 & 1/\lambda\\ 0 & 1
\end{smallmatrix}  \right)\in \Gamma_g'$ fixes $\infty$ and has multiplier $e^{-\frac{2\pi i }{\lambda}}$, it is clear that $\gamma:=\gamma'\circ  \left(\begin{smallmatrix}
1 & r/\lambda\\ 0 & 1
\end{smallmatrix}  \right)$ has trivial multiplier and still maps $\infty$ to $\cc$. Since $T_{\id,g}$ has only a single pole at $\infty$ modulo $\Gamma_g$, then it must be one-to-one as a holomorphic function from $\overline{\HH/\Gamma}_g$ to the Riemann sphere $\overline\CC$. Therefore, $\overline{\HH/\Gamma}_g$ must have the topology of a sphere and $T_{\id,g}$ is a Hauptmodul.
\end{proof}

\section{Outlook}
\noindent The results of this paper point to many interesting directions for future research. One of the main points of the paper was to provide a physical derivation of the genus zero property of Monstrous moonshine in the context of the spacetime properties of certain heterotic CHL-models. Another physical interpretation of genus zero was proposed previously by Duncan and Frenkel \cite{MR2905139}. Inspired by an earlier conjecture by Witten \cite{Witten:2007kt}, Duncan and Frenkel proposed that there exists a class of twisted chiral quantum gravity theories in $AdS_3$, whose partition functions are given by the McKay-Thompson series $T_g(\tau)$. If true, this implies that all the McKay-Thompson series have a Rademacher sum, interpreted physically as a sum over black hole states in the gravitational theory. Thus, the Duncan-Frenkel conjecture implies that the genus zero property of moonshine can rephrased as the statement that each McKay-Thompson series $T_g$ coincides with the Rademacher sum attached to the invariance group $\Gamma_g$. Given that our analysis identifies the product formula for the McKay-Thompson series with the supersymmetric index counting BPS-states, it is natural to expect that there is a relation between our results and those of Duncan-Frenkel. In particular, we would expect that whenever our twisted partition function $Z_g(T,U)$ is \emph{unbounded}, then the associated McKay-Thompson series $T_g(\tau)$ is Rademacher summable. From this point of view it would also be interesting to write the one-loop integrals $S_{1-loop}$ as explicit BPS-state sums, using the formalism developed in \cite{Angelantonj:2011br,Angelantonj:2012gw,Angelantonj:2013eja}.

We have shown that the space of BPS-states in our heterotic CHL-models form a module for the Monstrous Lie algebras $\mathfrak{m}_g$. This differs from the original proposal of Harvey and Moore, whose starting point was the observation that for string theory compactified on a manifold $X$, there is a product on the space of BPS-states itself,
\be
\mathcal{H}_{\mathrm{BPS}}\otimes \mathcal{H}_{\mathrm{BPS}}\to \mathcal{H}_{\mathrm{BPS}},
\label{BPSproduct}
\ee
which is a realization of the fact that two BPS-states can combine into a bound state which is also BPS. The space of BPS-states is graded by the charge lattice $\Gamma$ (essentially the integer cohomology $H^{*}(X; \mathbb{Z})$):
\be
\mathcal{H}_{\mathrm{BPS}}=\bigoplus_{\gamma\in \Gamma}\mathcal{H}_{\mathrm{BPS}}(\gamma).
\ee
For  states $\mathcal{B}_1, \mathcal{B}_2$ of charges $\gamma_1, \gamma_2$ their product $\mathcal{B}_1\otimes \mathcal{B}_2 \mapsto \mathcal{B}_3$
yields a  bound state $\mathcal{B}_3$ of charge $\gamma_3=\gamma_1+\gamma_2$. Harvey and  Moore argued that the product (\ref{BPSproduct}) on the space of BPS-states is defined via the \emph{correspondence conjecture}, which asserts that the three BPS-states $(\mathcal{B}_1, \mathcal{B}_2, \mathcal{B}_3)$ must fit into an
exact sequence $0\to \mathcal{B}_1\to \mathcal{B}_3\to \mathcal{B}_2\to 0$
which means that the bound state $\mathcal{B}_3$ should be viewed as an \emph{extension} of $\mathcal{B}_1$ and $\mathcal{B}_2$. The product on $\mathcal{H}_{\mathrm{BPS}}(\gamma)$ should  reflect this property and 
a natural candidate is therefore
\be
\mathcal{B}_1\otimes \mathcal{B}_2=\sum_{\mathcal{B}_3} |\{ 0\to \mathcal{B}_1\to \mathcal{B}_3\to \mathcal{B}_2\to 0\}| \mathcal{B}_3,
\ee
where we consider BPS-states as cohomology classes $\mathcal{B}_i\in H_{L^2}^{*}(\mathcal{M}(\gamma))$, where $\mathcal{M}(\gamma)$ is the moduli space of (semi-)stable sheaves on $X$, and the structure constants are given by the dimension of 
the space of extensions. This product is the characteristic feature of a \emph{Hall algebra}, as was first noted in \cite{Fiol:2000wx}. It would be very interesting to understand the precise relation between our analysis and the BPS-algebra of Harvey and Moore. Is it possible to endow the Monstrous Lie algebras $\mathfrak{m}_g$ with a Hall algebra structure?

Carnahan has recently proven \cite{2008arXiv0812.3440C,Carnahan2014,2012arXiv1208.6254C} Norton's generalized moonshine conjecture, showing in particular that all generalized McKay-Thompson series $T_{g,h}(\tau)$ are Hauptmoduln for genus zero subgroups $\Gamma_{g,h}\subset SL(2,\mathbb{R})$. To be precise, this may be considered as a ``weak version'' of generalized moonshine, since there are possible modular phases of $T_{g,h}$ which are not specified by Carnahan's theorem. It was conjectured by Gannon that these phases can be completely captured by a 3-cocycle $\alpha$, determining a class in $H^3(\mathbb{M}, U(1))$. However, very little is known about this cohomology group so at present this seems out of reach. The  generalized version of Mathieu moonshine for $M_{24}$ \cite{Eguchi:2010ej,Cheng:2010pq,Gaberdiel:2010ch,Gaberdiel:2010ca,Eguchi:2010fg,Gannon:2012ck} was established in \cite{Gaberdiel:2013nya,Gaberdiel:2013nya,Persson:2013xpa}, where an explicit cocycle $\alpha\in H^3(M_{24}, U(1))$ was constructed and shown to reproduce all modular phases of the twisted twining genera $\phi_{g,h}(\tau, z)$. Thus, this may be viewed as a ``strong'' form of generalized moonshine for $M_{24}$. It would be very interesting to investigate whether the Monstrous CHL-models constructed in the present paper can be used to shed light on the analogous strong form of generalized moonshine for the Monster group $\mathbb{M}$.

It would be illuminating to attach a better physical interpretation to the Lie algebra homology operators $d, d^{\dagger}$, considered in \cite{GarlandLepowsky} and employed in \ref{ss:algebra}. In particular, we notice that their definition is reminiscent of (part of) the usual definition of the standard worldsheet BRST operator, if one makes the substitutions $b \mapsto \eta, c \mapsto {\partial \over \partial \eta}$. Of course, the $\eta_a, {\partial \over \partial \eta_a}$ have the statistics of ordinary fermions, while the BRST operator is comprised of ghost fields. (In fact, \cite{GarlandLepowsky} introduce a second, equivalent, complex with operators $D, D^{\dagger}$ that match the usual worldsheet BRST operator and its conjugate after the aforementioned substitution; the arguments of \ref{ss:algebra} carry through with these operators in the same way, but with a larger technical burden.) If we take this correspondence seriously, an exciting, but currently quite speculative, possibility is that $d, d^{\dagger}$ (or $D, D^{\dagger}$) are indeed BRST-like operators for the spontaneously broken gauge symmetry and reduce to ordinary BRST operators in the tensionless limit, when the gauge symmetry is restored. Regardless, it would be interesting to elucidate the physical importance of these operators more fully.

Finally, an outstanding open question is to determine whether or not our Monstrous heterotic model has a IIA dual and, if so, what the dual theory is. The absence of currents in this model makes this a difficult question to approach. A perhaps more manageable, yet related, question is to investigate analogous constructions with $c=24$ CFTs with currents on the left, arising from compactification on the Niemeier lattices. If this latter class of theories admit IIA duals whose geometry includes $K3$ surfaces, it is conceivable that our methods may help shed light on the proposed relationships between umbral moonshine \cite{Cheng:2012tq, Cheng:2013wca} (including $M_{24}$ moonshine) and $K3$ geometry, by string-string duality  \cite{Cheng:2013kpa, Datta:2015hza, Harrison:2013bya,Cheng:2014zpa, Paquette:2014rma, cheng2015landau}.

\section*{Acknowledgements}

\noindent We are grateful to Richard Borcherds, Scott Carnahan, Miranda Cheng, Matthias Gaberdiel, Terry Gannon, Shamit Kachru and Michael Tuite for helpful discussions and correspondence. We also thank Matthias Gaberdiel for sharing some related unpublished notes on heterotic compactifications on $V^\natural$. N.M.P. is supported by a National Science Foundation Graduate Research Fellowship. N.M.P. and R.V. thank the University of Amsterdam for hospitality during the development of this work, and N.M.P. gratefully acknowledges the Delta Institute for Theoretical Physics for additional support. D.P. thanks Stanford University for hospitality while this work was initiated, and all of us thank the organizers and participants of the workshop on ``(Mock) Modularity, Moonshine and String Theory'' at Perimeter Institute, and the LMS-EPSRC Symposium on ``New Moonshines, Mock Modular Forms and String Theory'', at Durham University.

\section*{Appendices}
\appendix

\noindent In these appendices, we will present several technical proofs referenced in the main text. In addition, Appendix \ref{a:TwistAndTwin} contains details about the twisted twining genera and their multiplier systems. Throughout these appendices, as in several parts of the main text, we will employ the notation $(a, b):= \text{gcd}(a, b)$.

\section{Proofs}

\subsection{Proof of Theorem \ref{th:easy}}\label{pf:easy}

Let us consider the subgroup of automorphisms of $L$ given by elements of the form $(\gamma,1)\in SL(2,\RR)\times SL(2,\RR)$. It is easy to see that the image of the action
\be X\mapsto \begin{pmatrix}
a & b\\ c & d
\end{pmatrix}X\ ,\qquad X\in L\ ,
\ee is in $L$ for all $X\in L$ if and only if $\left(\begin{smallmatrix}
a & b\\ c & d
\end{smallmatrix}\right)\in \Gamma_0(N)$. By an analogous reasoning, we find that an element $(1,\gamma_2)$ is an automorphism of $L$ if and only if $\gamma_2\in \Gamma_0(N)$. Therefore,
\be \Gamma_0(N)\times \Gamma_0(N)\subseteq \tilde{SO}^+(L)\ .
\ee Furthermore, if $(\gamma_1,\gamma_2)\in \tilde{SO}^+(L)$, then for any $\gamma\in\Gamma_0(N)$ also $(\gamma_1\gamma\gamma_1^{-1},1)$ is in $\tilde{SO}^+(L)$, so that $\gamma_1\gamma\gamma_1^{-1}\in \Gamma_0(N)$. Analogously, $\gamma_2\gamma\gamma_2^{-1}\in \Gamma_0(N)$. Therefore, both $\gamma_1$ and $\gamma_2$ are contained in the normalizer $\hat\Gamma_0(N)$ of $\Gamma_0(N)$ in $SL(2,\RR)$ and
\be \Gamma_0(N)\times \Gamma_0(N)\subseteq \tilde{SO}^+(L)\subseteq \hat\Gamma_0(N)\times \hat\Gamma_0(N)\ .
\ee 
Let us consider the action of a generic element of $\hat\Gamma_0(N)\times \hat\Gamma_0(N)$ on $L$
\be X\mapsto \frac{1}{\sqrt{ee'}}\begin{pmatrix}
ae & b/h\\ cN/h & de
\end{pmatrix} X\begin{pmatrix}
a'e' & b'/h\\ c'N/h & d'e'
\end{pmatrix}\ ,\qquad X\in L\ ,
\ee where, as above, $a,b,c,d,a',b',c',d'\in \ZZ$, $e,e'\in \ZZ_{>0}$, $e|| \frac{N}{h^2}$, $e'||\frac{N}{h^2}$ and
\be\label{detcond} ade-bc \frac{N}{eh^2}=1\qquad a'd'e'-b'c' \frac{N}{e'h^2}=1\ .
\ee
A direct calculation shows that the image of this action is in $L$ for all $X\in L$ if and only if $ee'$ is a perfect square 
\be ee'=z^2\qquad z\in \ZZ_{>0}
\ee and the following congruences hold
\be\label{congru} \frac{aec'}{hz}\in \ZZ\qquad \frac{deb' }{hz}\in \ZZ\qquad \frac{ba'e'}{hz}\in \ZZ\qquad \frac{cd'e'}{hz}\in \ZZ\ .
\ee Notice that, since both $e$ and $e'$ divide $N/h^2$, then $ee'|\frac{N^2}{h^4}$, so that also $z$ divides $N/h^2$.

From these relations, we can now restrict the form of the matrices that comprise $\tilde{SO}^+(L)$ to be exactly that of the matrices in Theorem \ref{th:easy}. Let $p$ be a prime divisor of $hz$, and $p^r$, with $r>0$, be the maximal power for which $p^r|hz$. Let us consider two cases.
\begin{enumerate}
\item Assume first that $p|c'$. Then, for the second of \eqref{detcond} to be satisfied, it is necessary that $a'$, $d'$, and $e'$ be all coprime to $p$. By \eqref{congru}, this implies that $p^r|b$ and $p^r|c$. But then, by the first of \eqref{detcond},  $a$, $d$, $e$ must be all coprime to $p$ (and since both $e$ and $e'$ are coprime to $p$, then also $z$ is and $p^r|h$). Thus, by \eqref{congru}, $p^r|b'$ and $p^r|c'$. Using an analogous reasoning, we conclude that if any of $b$, $c$, $b'$, or $c'$ is divisible by $p$, then all of them must be divisible by $p^r$, and $a,d,e,a',d',e'$, and $z$ are coprime to $p$. 
\item Now, suppose that $p$ does not divide any of $b$, $c$, $b'$, or $c'$. Then, by \eqref{congru}, $p^r$ divides $ae$, $de$, $a'e'$, and $d'e'$. By \eqref{detcond}, both $N/(h^2 e)$ and $N/(h^2 e')$ must be coprime to $p$. Let $p^s$ be the maximal power of $p$ dividing $N/h^2$. Then, $p$ is a prime factor both of $e$ and $e'$ (and therefore also $z$) and with the same power $s$. By \eqref{congru}, $a,d,a',d'$ are divisible by $p^{r-s}$, which is the maximal power of $p$ dividing $h$. 
\end{enumerate}
From this analysis, we deduce that if $p$ is a prime factor of $e$ (and therefore of $z$ and $hz$), then it also divides $e'$ with the same power. We conclude that
\be e=e'=z\ .
\ee Furthermore, $a,d,a',d'$ have the same greatest common divisor with $h$ (let us call it $k$)
\be (a,h)=(a',h)=(d,h)=(d',h)=:k\ ,
\ee 
and analogously 
\be (b,h)=(b',h)=(c,h)=(c',h)=h/k\ ,
\ee  and $(k,h/k)=1$.
We conclude that the elements in $\tilde{SO}^+(L)$ have the form
\be \left(\frac{1}{\sqrt{e}}\begin{pmatrix}
\alpha k e & \beta/k\\ \gamma N/k & \delta k e
\end{pmatrix},\ \frac{1}{\sqrt{e}}\begin{pmatrix}
\alpha' k e & \beta'/k\\ \gamma' N/k & \delta' k e
\end{pmatrix}\right)
\ee
where $\alpha,\beta,\gamma,\delta,\alpha',\beta',\gamma',\delta'\in \ZZ$, $k\in \ZZ_{>0}$ with $k||h$, $e\in \ZZ_{>0}$ with $e||\frac{N}{h^2}$ and
\be \alpha \delta k^2 e-\beta\gamma \frac{N}{ek^2}=1\qquad \alpha' \delta' k^2 e-\beta'\gamma' \frac{N}{ek^2}=1\ .
\ee Upon defining $\epsilon:=k^2e$, we finally obtain
\be \left(
\frac{1}{\sqrt{\epsilon }}\begin{pmatrix}
\alpha \epsilon & \beta\\ \gamma N & \delta \epsilon
\end{pmatrix},\  \frac{1}{\sqrt{\epsilon }}\begin{pmatrix}
\alpha' \epsilon & \beta'\\ \gamma' N & \delta' \epsilon
\end{pmatrix}\right)\ ,
\ee 
with
\be \alpha \delta \epsilon-\beta\gamma \frac{N}{\epsilon}=1\qquad \alpha' \delta' \epsilon-\beta'\gamma' \frac{N}{\epsilon}=1\ .
\ee Conversely, every $\epsilon||N$ can be written as $\epsilon=ek^2$, for some $k||h$ and $e||\frac{N}{h^2}$, so that \eqref{automgroup} follows. The fact that $\tilde{SO}^+(L)$ is generated by $\Gamma_0(N)\times \Gamma_0(N)$ and the Atkin-Lehner involutions follows from \cite{Conway:1979kx}.

\subsection{Proof of Theorem \ref{th:notsoeasy}}\label{pf:notsoeasy}

It is easy to verify by a direct calculation that the normal subgroup of $\tilde{SO}^+(L)$ generated by elements of the form $(1,\gamma)$ and elements of the form $(\gamma,1)$ is $\Gamma_0(N\lambda)\times \Gamma_0(N\lambda)$. This means that $\tilde{SO}^+(L)$ is a subgroup of the normalizer of $\Gamma_0(N\lambda)\times \Gamma_0(N\lambda)$ in $SL(2,\RR)\times SL(2,\RR)$:
\be \Gamma_0(N\lambda)\times \Gamma_0(N\lambda)\subseteq \tilde{SO}^+(L)\subseteq \hat\Gamma_0(N\lambda)\times \hat\Gamma_0(N\lambda)\ .
\ee The elements of $\hat \Gamma_0(N\lambda)$ are given in \eqref{normal}, where $h$ is the maximal integer such that $h^2$ is a divisor of $N\lambda$ and $h|24$ (in particular,  $\lambda|h$).
Let us consider the action of a generic element of $\hat\Gamma_0(N\lambda)\times \hat\Gamma_0(N\lambda)$ on $L$
\be X\mapsto \frac{1}{\sqrt{ee'}}\begin{pmatrix}
ae & b/h\\ cN\lambda /h & de
\end{pmatrix} X\begin{pmatrix}
a'e' & b'/h\\ c' N\lambda/h & d'e'
\end{pmatrix}\ ,\qquad X\in L\ ,
\ee where, as above, $a,b,c,d,a',b',c',d'\in \ZZ$, $e,e'\in \ZZ_{>0}$, $e|| \frac{N\lambda}{h^2}$, $e'||\frac{N\lambda}{h^2}$ and
\be\label{detcond2} ade-bc \frac{N\lambda}{eh^2}=1\qquad a'd'e'-b'c' \frac{N\lambda}{e'h^2}=1\ .
\ee   A direct calculation shows that the image of this action is in $L$ for every $X\in L$ if and only if $ee'$ is a square
\be ee'=z^2\qquad z\in \ZZ_{>0}
\ee and the following congruences are satisfied
\begin{align}\label{first}
\frac{aec'-ba'e'\E_g}{hz}&\in \ZZ\ ,  &
\frac{cd'e'-deb' \E_g}{hz}&\in \ZZ\ ,\\
\label{second} \frac{ba'e'}{\frac{h}{\lambda}z}&\in \ZZ\ , &
 \frac{deb'}{\frac{h}{\lambda}z}&\in \ZZ\ , 
\end{align}
\be\label{last} \frac{\E_g z(ad'-da')+\frac{N\lambda}{h^2z}(cc'-bb'\E_g^2)}{\lambda}\in \ZZ\ .
\ee Eqs.\eqref{first} and \eqref{second} imply the weaker conditions
\be\label{reduced} aec'\in \frac{hz}{\lambda}\ZZ\ ,  \qquad
cd'e'\in \frac{hz}{\lambda}\ZZ\ ,\qquad
ba'e'\in \frac{hz}{\lambda}\ZZ\ ,\qquad
deb'\in \frac{hz}{\lambda}\ZZ\ .
\ee
Let $p$ be a prime factor of $\frac{hz}{\lambda}$ and $p^r$, with $r>0$, the maximal power for which $p^r|\frac{hz}{\lambda}$. Then, using a reasoning analogous to the proof of Theorem \ref{th:easy}, \eqref{reduced} and \eqref{detcond2} imply that either $p^r$ divides $b$, $b'$, $c$, $c'$ and is coprime to $a$, $a'$, $d$, $d'$, $e$, $e'$; or $p^r$ divides $ae$, $de$, $a'e'$, $d'e'$ and is coprime to $b$, $b'$, $c$, $c'$, $\frac{N\lambda}{eh^2}$, and $\frac{N\lambda}{e'h^2}$. In particular, if $p^s$ is the maximal power of $p$ dividing $\frac{N\lambda}{h^2}$ and if $p$ divides $e$, then $p^s$ is an exact divisor of both $e$ and $e'$. It follows that
\be e=e'=z\ .
\ee Furthermore,
\be \left(a,\frac{h}{\lambda}\right)=\left(a',\frac{h}{\lambda}\right)=\left(d,\frac{h}{\lambda}\right)=\left(d',\frac{h}{\lambda}\right)=:k\ ,
\ee 
and
\be \left(b,\frac{h}{\lambda}\right)=\left(b',\frac{h}{\lambda}\right)=\left(c,\frac{h}{\lambda}\right)= \left(c',\frac{h}{\lambda}\right)=\frac{h}{\lambda k}\ ,
\ee with $(k,\frac{h}{\lambda k})=1$. Thus, upon defining 
\be a=\alpha k\quad d=\delta k \quad a'=\alpha' k\quad d'=\delta' k
\ee
\be b=\beta\frac{h}{\lambda k}\quad c=\gamma\frac{h}{\lambda k}\quad
b'=\beta'\frac{h}{\lambda k}\quad c'=\gamma'\frac{h}{\lambda k}
\ee and $\epsilon:=ek^2=e'k^2$ we obtain
\be \left(\frac{1}{\sqrt{e}}\begin{pmatrix}
ae & \frac{b}{h}\\ c\frac{N\lambda}{h} & de
\end{pmatrix}, \frac{1}{\sqrt{e'}} \begin{pmatrix}
a'e' & \frac{b'}{h}\\ c' \frac{N\lambda}{h} & d'e'\end{pmatrix}\right)=\left(\frac{1}{\sqrt{\epsilon}}\begin{pmatrix}
\alpha \epsilon & \frac{\beta}{\lambda }\\ \gamma N & \delta \epsilon
\end{pmatrix} ,\frac{1}{\sqrt{\epsilon}}\begin{pmatrix}
\alpha' \epsilon & \frac{\beta'}{\lambda}\\ \gamma' N & \delta' \epsilon
\end{pmatrix}\right)\ ,
\ee where 
 $\alpha,\beta,\gamma,\delta,\alpha',\beta',\gamma',\delta'\in \ZZ$,  $\epsilon\in \ZZ_{>0}$ with $\epsilon||\frac{N}{\lambda}$ (since $e|| \frac{N\lambda}{h^2}$ and $k||\frac{h}{\lambda}$, then $ek^2||\frac{N\lambda}{h^2}\frac{h^2}{\lambda^2}$)
 and
\be\label{detcond3} \alpha \delta \epsilon-\beta\gamma \frac{N}{\epsilon\lambda}=1\qquad \alpha' \delta' \epsilon-\beta'\gamma' \frac{N}{\epsilon\lambda}=1\ .
\ee The congruence conditions \eqref{first}--\eqref{last} become
\be\label{newcong} \frac{\alpha \gamma'-\beta \alpha'\E_g}{\lambda}\in \ZZ\ ,\qquad \frac{\gamma\delta'-\delta\beta'\E_g}{\lambda}\in \ZZ\ ,
\ee and
\be\label{newcong2} \frac{\epsilon\E_g(\alpha\delta'-\delta\alpha')+\frac{N}{\epsilon\lambda}(\gamma\gamma'-\beta\beta'\E_g^2)}{\lambda}\in \ZZ\ .
\ee Eqs.\eqref{detcond3} and \eqref{newcong}, together with $(\E_g,\lambda)=1$, imply
\be\label{samelcm} (\alpha,\lambda)=(\alpha',\lambda)\ ,\qquad
(\delta,\lambda)=(\delta',\lambda)\ ,\qquad
(\beta,\lambda)=(\gamma',\lambda)\ ,\qquad
(\gamma,\lambda)=(\beta',\lambda)\ .
\ee Set $f:= (\alpha,\lambda)=(\alpha',\lambda)$. The first equation in \eqref{samelcm} implies that
\be\label{prop} \alpha'\equiv \kappa_1\alpha\mod \lambda
\ee for some   integer $\kappa_1$; furthermore, eq.\eqref{newcong} determines $\kappa_1$ modulo $\frac{\lambda}{f}$. By \eqref{prop} and the first of \eqref{newcong}, we obtain
\be \gamma'\equiv \beta\E_g\kappa_1 \mod \frac{\lambda}{f}\ ,
\ee so that there is an integer $y$ such that
\be \gamma'\equiv \beta\E_g\kappa_1+y\frac{\lambda}{f}\mod \lambda\ .
\ee Since $\kappa_1$ is only defined modulo $\lambda/f$, we are free to shift $\kappa_1\to \kappa_1+x\frac{\lambda}{f}$ for some integer $x$, and obtain
\be \gamma'\equiv \beta\E_g\kappa_1+\frac{\lambda}{f}(y+x\beta\E_g)\mod \lambda\ .
\ee Let us choose $x$ such that
\be y+x\beta\E_g\equiv 0\mod f\ .
\ee Such an integer $x$ always exists, because $(\beta,\alpha)=1$ by \eqref{detcond3} and $(\E_g,\lambda)=1$ by construction, so that $\beta\E_g$ is coprime to $f=(a,\lambda)$ and is therefore invertible mod $f$. We conclude that there is an integer $\kappa_1$ satisfying \eqref{prop} and such that
\be\label{prop2} \gamma'\equiv \beta\E_g\kappa_1\mod \lambda\ .
\ee Eqs.\eqref{prop} and \eqref{prop2} determine $\kappa_1$ modulo $\lambda$. By a similar reasoning, one can show that there exists $\kappa_2$ such that
\be\label{prop3} \delta'\equiv\kappa_2\delta,\quad \E_g\beta'\equiv \kappa_2\gamma\mod \lambda\ .
\ee By \eqref{detcond3},\eqref{prop},\eqref{prop2},\eqref{prop3}, we obtain
\be 1\equiv \alpha' \delta' \epsilon-\beta'\gamma' \frac{N}{\epsilon\lambda}\equiv \kappa_1\kappa_2(\alpha \delta \epsilon-\beta\gamma \frac{N}{\epsilon\lambda})\equiv \kappa_1\kappa_2\mod \lambda\ .
\ee Finally, notice that, for every divisor $\lambda$ of $24$,
\be\label{cool24} \kappa_1\kappa_2\equiv 1\mod \lambda \qquad \Leftrightarrow \qquad \Bigl((\kappa_1,\lambda)=1 \text{ and }\kappa_1\equiv\kappa_2\mod \lambda\Bigr)\ .
\ee We conclude that eqs.\eqref{detcond3} and \eqref{newcong}  imply that there exists $\kappa\in \ZZ$, with $(\kappa,\lambda)=1$, such that
\be\label{bigprop} \alpha'\equiv \kappa \alpha\ , \quad
\beta'\equiv \kappa\E_g \gamma\ , \quad
\gamma'\equiv \kappa \E_g\beta\ , \quad
\delta'\equiv \kappa \delta\mod \lambda\ ,
\ee where in the second equality we used that, by \eqref{cool24}, $(\E_g,\lambda)=1$ implies $\E_g^2\equiv 1\mod \lambda$. Vice versa, \eqref{bigprop} obviously implies \eqref{newcong} and \eqref{newcong2}. We conclude that the elements of $\tilde{SO}^+(L)$ are as described in the statement of the theorem. The description of the generators of $\tilde{SO}^+(L)$ follows from \cite{Conway:1979kx}.

\subsection{Proof of Eq.\eqref{etaproduct}}\label{pf:etaproduct}
The coefficient $\hat c_{n,0}(0)$ depends only on the greatest common divisor $d:=(n,N)$, i.e.
\be \hat c_{n,0}(0)=\hat c_{d,0}(0)\ .
\ee Indeed, each $g^n$ is conjugated to either $g^d$ or to its inverse within the Monster group, so that the $g^n$-twisted sector is isomorphic to the $g^d$-twisted sector (or its dual). In particular, the dimensions of the $g$-invariant subspaces at level $L_0-1=0$ are the same. Thus,
\begin{align} &\prod_{n>0} (1-q^n)^{\hat c_{n,0}(0)}=\prod_{d|N}\prod_{\substack{n>0\\(n,N)=d}} (1-q^n)^{\hat c_{d,0}(0)}=
\prod_{d|N}\prod_{\substack{r>0\\(r,N/d)=1}} (1-q^{rd})^{\hat c_{d,0}(0)}\\
&=\prod_{d|N}\prod_{r=1}^\infty (1-q^{rd})^{\hat c_{d,0}(0)\sum_{i|(r,\frac{N}{d})}\mu(i) }=\prod_{d|N}\prod_{i|\frac{N}{d}}\prod_{k=1}^\infty (1-q^{kid})^{\hat c_{d,0}(0)\mu(i) }\\
&=\prod_{d|N}\prod_{i|\frac{N}{d}}\bigl(q^{-\frac{id}{24}}\eta(id\tau)\bigr)^{\hat c_{d,0}(0)\mu(i) }=\prod_{\ell|N}\prod_{d|\ell}\bigl(q^{-\frac{\ell}{24}}\eta(\ell\tau)\bigr)^{\hat c_{d,0}(0)\mu(\ell/d) }=\prod_{\ell|N}\bigl(q^{-\frac{\ell}{24}}\eta(\ell\tau)\bigr)^{\alpha(\ell)}\ ,
\end{align} where we used the property
\be \sum_{a|b}\mu(a)=\begin{cases}1 & \text{for }b=1\ ,\\ 0 & \text{otherwise}\ ,\end{cases}
\ee of the M\"obius function.
It follows that
\be T_{\id,g}(\tau)-\hat c_{1,0}(0)=q^{-1}\prod_{n>0} (1-q^n)^{\hat c_{n,0}(0)}=
q^{-\frac{24+\sum_{\ell|N}\ell\alpha(\ell)}{24}}
\prod_{\ell|N}\eta(\ell\tau)^{\alpha(\ell)}\ .
\ee 
In order for the right hand side to be modular under $\Gamma_0(N)$, the power of $q$ in front of the $\eta$-product must vanish, i.e.
\be \sum_{\ell|N}\ell\alpha(\ell)=-24\ .
\ee Finally, we notice that $\alpha(1)=\hat c_{1,0}(0)$ and obtain \eqref{etaproduct}.

\subsection{Proof of Eq.\eqref{WeylCompatible}}\label{pf:WeylCompatible}

Let $g\in \Aut(V^\natural)$ of order $N$ be such that the orbifold $V^\natural/\langle g\rangle$ is consistent ($\lambda=1$) and has no currents. In section \ref{s:Weyl}, we showed that, in this case, there is an isomorphism (corresponding to T-duality in the space direction for the CHL model based on $g$)
\be f:V^\natural_{n,m}\to V^\natural_{m,n}\qquad n,m\in \ZZ/N\ZZ\ ,
\ee of $(V^\natural)^g$ modules inducing an isomorphism
of the VOAs  $V'=\oplus_n V^\natural_{n,0}$ and $V^\natural=\oplus_m V^\natural_{0,m}$. The isomorphism maps $g\in \Aut(V^\natural)$ to the quantum symmetry $Q\in \Aut(V')$.  The BKM algebra corresponding to the CHL model based on $g$ has a real root $\gamma=(1,-1)$ with generator $a$ and the Weyl reflection $r_\gamma$ with respect to $\gamma$ is exactly the automorphism induced by $f$. Let $h\in \Aut(V^\natural)$ be a symmetry commuting with $g$ and of order coprime to $N$. In this appendix, we will prove that $h$ always fixes the root $a$.

The root $a$ corresponds to the ground state in the $g$-twisted sector of $V^\natural$. Recall that $T_{g,\id}(\tau)=q^{-\frac{1}{N\lambda}}+\ldots$, so that the $g$-twisted ground state is $1$-dimensional and $h$ can only act on $a$ by a phase $\xi$. The phase $\xi$ is the coefficient of the polar term $q^{-\frac{1}{N\lambda}}$ in the series $T_{g,h}(\tau)$.  Since $g$ and $h$ have coprime order, $T_{g,h}(\tau)$ is a $SL(2,\ZZ)$ transformation of $T_{\id,gh}(\tau)$. If $gh$ has order $M$ then, for each $n$ coprime to $M$, $g^nh^n$ is in the same Monster conjugacy class as $gh$ or $(gh)^{-1}$. In either case, one has $T_{\id,gh}=T_{\id,g^nh^n}$, so that $T_{g^n,h^n}=T_{g,h}$. Now, one can always find $n$ coprime to $M$ such that $g^n=g$ and $h^n=h^{-1}$, so that $T_{g,h}=T_{g,h^{-1}}$ and $\xi=\xi^{-1}$. This means that $\xi=\pm 1$. 

If $h$ has odd order, then the only possibility is $\xi=1$. Let us consider the case where $h$ has even order, so that the order $N$ of $g$ is odd. We will suppose by absurdity that $h(a)=-a$ and derive an inconsistency. The symmetry $h$ must act by $-1$ also on the generator $\theta(a)$, corresponding to the opposite root $(-1,1)$, and leave invariant the Cartan generator $H_a=\frac{1}{2}[a,\theta(a)]$. Consider a generator $b\in \g_{(1,m)}\otimes \CC$ in the (complexified) algebra component corresponding to a simple root $(1,m)$, $m> 0$, and assume it is an eigenvector of $h$ with eigenvalue $\zeta$. The vector $b$ is annihilated by $\theta(a)$ (because there are no generators with roots $(0,m+1)$) and has eigenvalue $1-m$ under the Cartan generator $H_a$. Therefore, $b$ is the highest weight vector of a $m$-dimensional representation of the $\mathfrak{sl}_2$ subalgebra generated by $a,\theta(a)$ and $H_a$, whose lowest weight vector is in the root component $\g_{(m,1)}$. Furthermore, this $\mathfrak{sl}_2$-representation is spanned by eigenvectors of $h$ with eigenvalues $\pm \zeta$; in particular, the lowest weight vector with root $(m,1)$ has eigenvalue $(-1)^{m-1}\zeta$. The Weyl reflection $r_\gamma$ exchanges the highest and lowest weight vector within each $\mathfrak{sl}_2$ representation. Therefore, it preserves the $h$-eigenvalues of all odd dimensional $\mathfrak{sl}_2$-representations  and multiplies by $-1$ the $h$-eigenvalues of all even dimensional $\mathfrak{sl}_2$-representations. In particular, a generator $b\in \g_{(1,m)}$ for a simple root $(1,m)$ and its Weyl transform $r_\gamma(b)\in \g_{(m,1)}$ have the same (respectively, the opposite) $h$-eigenvalue if $m$ is odd (respectively, even).

 This means that conjugation of $h$ by the Weyl reflection defines a new automorphism $r_\gamma hr_\gamma$ commuting with both $g$ and $h$ (since it has the same eigenvectors as $h$). An equivalent description of $r_\gamma hr_\gamma$ is as follows.  The symmetry $h\in \Aut(V^\natural)$ has a unique lift to a symmetry (automorphism of $(V^\natural)^g$-modules) of the same order acting on the twisted sectors $V^\natural_{n,m}$. This lift induces a symmetry $h'\in \Aut(V')$ of the VOA $V'=\oplus_n V^\natural_{n,0}$. By conjugating $h'$ by the isomorphism $f:V'\to V^\natural$, we obtain a new automorphism $\tilde h:=fh'f^{-1}\in \Aut(V^\natural)$ of $V^\natural$, possibly different from $h$. Then, $r_\gamma hr_\gamma$ is the symmetry induced on the BKM algebra by $\tilde h$.
Under the assumption that $h(a)=-a$, the symmetries $h,\tilde h\in \Aut(V^\natural)$ and the corresponding actions on the BKM algebra $\g$ are different. Indeed, the composition $t:=h^{-1}\tilde h$ acts by $(-1)^{m-1}$ on the simple roots in $\g_{1,m}$. Let us prove that such a symmetry $t$ cannot exist. Recall that the generators of $\g_{1,m}$ correspond to states in the $g$-twisted sector $V^\natural_{1,m}$ of $V^\natural$ and level $(L_0-1)=\frac{m}{N}$. These states are eigenvectors of $t$ with eigenvalue $(-1)^{m-1}$. Let $\chi$ be a $t$-invariant state in the $g$-twisted sector (for example the ground state) and level $(L_0-1)=\frac{m}{N}$, with $m$ odd. The Virasoro descendant $L_{-r}\chi$, with $r$ odd, has level $(L_0-1)=\frac{m}{N}+r=\frac{m+rN}{N}$ and since both $r$ and $N$ are odd, has $t$-eigenvalue $(-1)^{m+rN-1}=-1$. But this is absurd: every symmetry $t$ must commute with the Virasoro algebra, so $\chi$ and $L_{-r}\chi$ must have the same eigenvalue. Thus, the symmetry $t$ cannot exist, and the initial assumption that $h(a)=-a$ is inconsistent.

\subsection{Proof of Lemma \ref{t:lecuspe}}\label{pf:lecuspe}

Let $\frac{a}{\lambda c}$, with $a,c>0$, $c|\frac{N}{\lambda}$, $(a,c)=1$, be a representative for the cusp $\cc$. We have to prove that if $T_{\id,g}$ is unbounded at $\frac{a}{\lambda c}$, then $c$ is an exact divisor of $\frac{N}{\lambda}$. Indeed, in this case also $e=\frac{N}{\lambda c}$ is an exact divisor for $\frac{N}{\lambda}$ and $(ae,\frac{N}{\lambda e})=1$. Thus, we can find integers $b,d$ such that $ade-\frac{N}{\lambda e} b=1$, so that
\be w_e=\frac{1}{\sqrt{e}}\begin{pmatrix}
ae & \frac{b}{\lambda}\\ N & de
\end{pmatrix}
\ee is an Atkin-Lehner involution for $\Gamma_0(N|\lambda)$ such that $w_e\cdot \infty =\frac{ae}{N}=\frac{a}{\lambda c}$.

Choose integers $\alpha,\beta,\delta$ such that $(\begin{smallmatrix}
\alpha & \beta\\ \lambda c & \delta
\end{smallmatrix}
 )\in SL(2,\ZZ)$. The assumption that $T_{\id,g}$ is unbounded at $\frac{a}{\lambda c}$ implies that
\be T_{\id,g}\left((\begin{smallmatrix}
\alpha & \beta\\ \lambda c & \delta
\end{smallmatrix}
 )\cdot \tau\right)=\xi T_{g^{\lambda c},g^{\delta}}(\tau)=\xi Aq^{-\frac{\lambda c}{N}}+O(q^0)\ ,\ee where $\xi$ is a phase and $A$ is the eigenvalue of $g^{\delta}$ on the ground state of $\Hh_{g^{\lambda c}}$. (Notice that if $T_{\id,g}$ has multiplier $\lambda$, then $g^{\lambda c}$ has trivial multiplier).
Then
\be\label{phase1} T_{\id,g}\left((\begin{smallmatrix}
\alpha & \beta+k\alpha\\\lambda c & \delta+k\lambda c
\end{smallmatrix}
 )\cdot \tau\right)=\xi T_{g^{\lambda c},g^{\delta}}(\tau+k)=
\xi Ae^{-\frac{2\pi i k\lambda c}{N}}q^{-\frac{\lambda c}{N}}+O(q^0)\ .
\ee Set
\be k=\frac{\omega N}{\lambda^2 c^2}\ ,
\ee where $\omega=\frac{\lambda^2c^2}{(N,\lambda^2c^2)}$ is the minimal integer for which $(\lambda c)^2|N\omega$, and notice that
\be \begin{pmatrix}
\alpha & \beta+k\alpha\\ \lambda c & \delta+k\lambda c
\end{pmatrix}=\begin{pmatrix}
1-k\alpha\lambda c & k\alpha^2\\ -k\lambda^2 c^2 & 1+k\alpha\lambda c
\end{pmatrix}\begin{pmatrix}
\alpha & \beta\\ \lambda c & \delta
\end{pmatrix}=\begin{pmatrix}
1-\frac{\omega N}{\lambda c}\alpha & k\alpha^2\\ -\omega N & 1+k\alpha\lambda c
\end{pmatrix}\begin{pmatrix}
\alpha & \beta\\ \lambda c & \delta
\end{pmatrix}\ .
\ee Since $T_{\id,g}$ is $\Gamma_0(N)$ invariant up to a phase, we have
\be\label{phase2} T_{\id,g}\left((\begin{smallmatrix}
\alpha & \beta+k\alpha\\ \lambda c & \delta +k\lambda c
\end{smallmatrix}
 )\cdot \tau\right)=T_{\id,g}\left((\begin{smallmatrix}
1-\frac{\omega N}{\lambda c}\alpha & k\alpha^2\\ -\omega N & 1+k\alpha\lambda c
\end{smallmatrix})(\begin{smallmatrix}
\alpha & b\beta\\ \lambda c & \delta
\end{smallmatrix})\cdot \tau\right)=e^{\frac{2\pi i \omega \E_g}{\lambda} }T_{\id,g}\left((\begin{smallmatrix}
\alpha & \beta\\ \lambda c & \delta
\end{smallmatrix}
 )\cdot \tau\right)\ .
\ee By comparing \eqref{phase1} and \eqref{phase2}, we obtain 
\be e^{-\frac{2\pi i \omega}{c\lambda}}= e^{\frac{2\pi i \E_g\omega}{\lambda} }\ ,
\ee that is
\be\label{compare} \frac{\omega}{c\lambda}\equiv \frac{-\E_g\omega}{\lambda}\mod \ZZ\ .
\ee  
Now\footnote{We thank the anonymous referee for feedback that greatly simplified the latter part of this proof.}, since $\omega = \frac{\lambda^2 c^2}{(N, \lambda^2 c^2)}$ the latter congruence implies $\lambda c \equiv -\E_g \lambda c^2 \mod (N, \lambda^2 c^2)$. Next, for any prime $p$ that divides $c$, let $v_p(x)$ be the $p$-valuation of the integer $x$, which is defined as the number of times $x$ can be divided by $p$. We have immediately that $v_p((N, \lambda^2 c^2)) = \textrm{ min }(v_p(N), v_p(\lambda^2 c^2))$ and the congruence further implies that $\textrm{ min }(v_p(\lambda c), v_p(N)) = \textrm{ min }(v_p(-\E_g \lambda c^2), v_p(N), v_p(\lambda^2 c^2))$. Finally, since $v_p(c)> 0$ by assumption, we have $v_p(\lambda c) = v_p(N)$.

\section{Twisted-twining partition functions}\label{a:TwistAndTwin}

\noindent In this Appendix we summarize some basic properties of orbifolds of holomorphic conformal field theories (vertex operator algebras). We refer to \cite{Miyamoto:2013kx,Ekeren:2015kq} and references therein for recent new results on this topic.

Given a holomorphic VOA (chiral bosonic two-dimensional CFT, in physics language) $V$ of central charge $24$  with group of automorphisms $G=\Aut(V)$, we can consider the twisted-twining partition functions
\be\label{Twisted} T_{g,h}(\tau):=\Tr_{V_g}(hq^{L_0-1})\ ,
\ee where $g,h$ are any commuting elements of $G$ and $V_g$ is the $g$-twisted sector of $V$. Strictly speaking, the action of $h$ on the twisted sectors is only determined up to a $N$-th root of unity, where $N$ is the order of $g$: we implicitly assume that a choice has been made for this action. When $h=g$, we always make the standard choice
\be\label{gaction} g=e^{2\pi i L_0}\ .
\ee
 The partition functions \eqref{Twisted} are related to one each other by modular transformations
 \be\label{modulagenera} T_{g,h}\bigl(\frac{a\tau+b}{c\tau+d} \bigr)=\xi_{g,h}\begin{pmatrix}
 a & b\\ c &d
\end{pmatrix}   T_{g^ah^c,g^bh^d}(\tau)\ ,\qquad \begin{pmatrix}
 a & b\\ c &d
\end{pmatrix}\in SL(2,\ZZ)\ ,
 \ee where $\xi_{g,h}\bigl(\begin{smallmatrix}
 a & b\\ c &d
\end{smallmatrix}   \bigr)$ are non-zero complex numbers.
 We will only focus on the case where the group generated by $g$ and $h$ is cyclic. For a cyclic group $\langle g\rangle$ of order $N$, all $\xi_{g^i,g^j}$ are $N$-th roots of unity and are completely determined by the conformal weight ($L_0$-eigenvalue) of the $g$-twisted ground state. In general, the $L_0$-eigenvalue $\Delta_g$ of a $g$-twisted state takes value in
 \be\label{fundam1} \Delta_g\in \frac{\E_g}{N\lambda_g}+\frac{1}{N}\ZZ\ ,
 \ee for some positive integer $\lambda_g$, with $\lambda_g|N$, and $\E_g\in \ZZ/\lambda_g\ZZ$ coprime to $\lambda_g$ (we often omit the subscript in $\lambda_g$ when there is no ambiguity). The orbifold of $V$ by $g$ is a consistent CFT if and only if $\lambda=1$; when  $\lambda>1$ we have a failure of the level matching condition. 
The $g$-twisted and $g^{-1}$-twisted sector have the same conformal weights modulo $1/N$
 \be \Delta_g\equiv \Delta_{g^{-1}} \mod \frac{1}{N}\ZZ\ .
 \ee More generally, the conformal weights of the $g^n$-twisted sector take values in
 \be\label{fundam2} \Delta_{g^n}\in n^2\frac{\E_g}{N\lambda_g}+\frac{(n,N)}{N}\ZZ\ .
 \ee In particular, one has $\lambda_{g^n}=\frac{\lambda_g}{(\lambda_g,n)}$, so that the orbifold of $V$ by $\langle g^n\rangle$ is consistent if and only if $n$ is a multiple of $\lambda_g$. 
 
 The Monster CFT $V^\natural$ (or rather its group of automorphisms $\Aut(V^\natural)=\MM$) has another special property: given any $g\in \MM$ of order $N$ and any integer $a$ coprime to $N$, $g^a$ is always conjugated to either $g$ or $g^{-1}$
 \be\label{conjuga} (a,N)=1\qquad \rightarrow \qquad g^a=hgh^{-1}\text{ or } g^a=hg^{-1}h^{-1}\ ,
 \ee for some $h\in \MM$. The symmetry $h$ induces an isomorphism
 \be \phi_h:V_g\stackrel{\cong}{\longrightarrow} V_{g^a}\ ,
 \ee between the $g$-twisted and the $g^a$-twisted sector. In particular, the conformal weights must be the same
 \be \Delta_g\equiv \Delta_{g^a} \mod \frac{1}{N}\ZZ\ ,
 \ee and using \eqref{fundam1},\eqref{fundam2} one obtains that
 \be (a,N)=1 \qquad \Rightarrow \qquad  a^2\equiv 1\mod \lambda\ .
 \ee As observed in \cite{Conway:1979kx}, this condition holds for all $a\in \ZZ/\lambda\ZZ$ if and only if $\lambda$ is a divisor of $24$. In general,  for any holomorphic CFT $V$ and for all $g\in \Aut(V)$ such that \eqref{conjuga} holds, one has $\lambda_g|24$. 
 
For CFTs where \eqref{conjuga} holds, the twining partition function $T_{\id,g}$ is a modular form for 
 \be \Gamma_0(N):=\bigl\{ \begin{pmatrix}
 a & b\\ c & d
 \end{pmatrix}\in SL(2,\ZZ)\mid c\equiv 0\mod N  \bigr\} \ ,
 \ee up to a multiplier
 \be T_{\id,g}\bigl(\frac{a\tau+b}{c\tau+d} \bigr)=e^{-2\pi i \frac{\E_g cd}{N\lambda}}   T_{\id,g}(\tau)\ ,\qquad\qquad \begin{pmatrix}
 a & b\\ c & d
 \end{pmatrix}\in \Gamma_0(N)\ .
 \ee In particular, $T_{\id,g}$ is invariant under $\Gamma_0(N\lambda)$. \footnote{For a general CFT $V$ and $g\in \Aut(V)$, where \eqref{conjuga} does not hold, $T_{\id,g}$ is only modular under a subgroup $\Gamma_1(N)\subseteq \Gamma_0(N)$.} 
 
The fixed point subVOA $V^g$ has $N^2$ irreducible modules, given by the $g$-eigenspaces in the $g^r$-twisted sectors $r=1,\ldots, N$.   When $\lambda>1$, by our definition \eqref{gaction}, $g$ has order $N\lambda$ when acting on the twisted sectors, i.e. it generates a central extension of $\langle g\rangle$. In constructing the CHL orbifolds, it is useful to define $(N\lambda)^2$ irreducible $V^g$-modules, given by
\be\label{Vnm} V_{n,m}=\{ v\in V_{g^n}\mid g(v)=e^{\frac{2\pi i m}{N\lambda}}v \}\ ,\qquad n,m\in \ZZ/N\lambda\ZZ\ .
\ee Notice that 
\be m- n\E_g\neq 0\mod \lambda\qquad \Rightarrow \qquad V_{n,m}=0\ .
\ee so that only $N^2\lambda$ out of $(N\lambda)^2$ modules $V_{n,m}$ are actually non-zero. Furthermore, many of the non-zero $V_{n,m}$ are isomorphic. For example, the $g^N$-twisted sector is isomorphic, as a $V^g$-module, to the untwisted sector. However, as discussed in \cite{Persson:2015jka}, it is useful to define the action of $g$ on the $g^N$-twisted sector so that
\be g(\V_N)=e^{-2\pi i \frac{\E_g}{\lambda}}\V_N\ ,
\ee where $\V_N$ is the vertex operator relative to the ground state  of conformal weight $0$ in the $g^N$-twisted sector. This choice yields the simple fusion rules
\be V_{i,j} \boxtimes_{V^g} V_{k,l}\to V_{i+k,j+l}\ ,\qquad i,j,k,l\in \ZZ/N\lambda\ZZ\ ,
\ee and allows to eliminate the phases in the $SL(2,\ZZ)$ transformations of $T_{\id,g}$ 
\be T_{\id,g}\bigl(\frac{a\tau+b}{c\tau+d} \bigr)=   T_{g^c,g^d}(\tau)\ ,\qquad\qquad \begin{pmatrix}
 a & b\\ c & d
 \end{pmatrix}\in SL(2,\ZZ)\ .
\ee The CHL orbifold relative to $g$ can be simply defined by imposing that the strings with momentum $m$ and winding $w=\frac{n}{N\lambda}$ be tensored with states in $V_{n,m}$. The OPE with the $g^N$-twisted vertex operator $\V_N$ defines an isomorphism
\be \V_N:V_{n,m}\stackrel{\cong}{\longrightarrow} V_{n+N,m-\E_g N}\ .
\ee These equivalences further reduce the number of irreducible modules from $N^2\lambda$ to $N^2$, as expected. 

\bibliographystyle{utphys}

\bibliography{Refs}

\end{document}